\documentclass[12pt,a4paper]{article}

\usepackage{amssymb,amsmath,geometry,graphicx,color,setspace,natbib,array,hyperref,amsthm,bm,tikz}
\usepackage{dsfont}
\usepackage{float}
\usepackage{subcaption}
\usepackage{subfiles}
\usepackage{makecell}

\usepackage[T1]{fontenc}
\usepackage{newpxmath}
\usepackage{newpxtext}
\usepackage{threeparttable}

\allowdisplaybreaks

\theoremstyle{plain}\newtheorem{theorem}{Theorem}
\theoremstyle{plain}\newtheorem{proposition}{Proposition}
\theoremstyle{plain}\newtheorem{corollary}{Corollary}
\theoremstyle{plain}\newtheorem{lemma}{Lemma}
\theoremstyle{plain}

\theoremstyle{definition}
\newtheorem{definition}{Definition}

\newtheorem{remark}{Remark}
\newtheorem{example}{Example}

\newcommand{\supp}[1]{\text{supp}(#1)}

\newcommand{\argmax}{\mathop{\rm arg~max}\limits}

\usetikzlibrary{calc}
\usetikzlibrary{positioning}
\usetikzlibrary{matrix,fit}

\usepackage{booktabs}

\makeatletter
\let\@makefntextOring\@makefntext
\def\@makefntext#1{\@makefntextOring{\baselineskip=15pt#1}}
\makeatother

\definecolor{navyblue}{rgb}{0.0, 0.0, 0.5}
\definecolor{green}{rgb}{0.2, 0.5, 0.2}

\hypersetup{
    colorlinks=true,
    citecolor=navyblue,
    linkcolor=navyblue,
    urlcolor=navyblue,
}

\def\citeapos#1{\citeauthor{#1}\textcolor{navyblue}{'s} (\citeyear{#1})}
\usepackage{titlesec}

\usepackage{etoolbox}
\patchcmd{\thanks}{#1}{\protect\doublespacing}{}{}

\DeclareFontShape{OT1}{cmr}{m}{n}{<->cmr10}{}
\usepackage{fix-cm}

\usepackage{xr}

\title{Value of Information in Social Learning\thanks{\protect \onehalfspacing 
We thank the editor, an associate editor, the anonymous referees, Yu Awaya, Makoto Hanazono, Michihiro Kandori, Satoshi Kasamatsu, Daiki Kishishita, Akihiko Matsui, Satoshi Nakada, Shinpei Noguchi, Daisuke Oyama, Tadashi Sekiguchi, Ryo Shirakawa, Wataru Tamura, Yuichi Yamamoto, and seminar participants at Tokyo University of Science, Nagoya University, Keio University, the 2025 Japanese Economic Association Autumn Meeting, and Economic Theory Workshop in Morioka.
All remaining errors are our own.
Sato acknowledges the financial support from the JSPS KAKENHI Grant 24KJ0100.
Shimizu acknowledges the financial support from the JSPS KAKENHI Grant 23KJ0667.
}}
\author{
\Large Hiroto Sato\thanks{\protect \onehalfspacing 
Department of Economics, Nagoya University, Furo-cho, Chikusa-ku, Nagoya 464-8601, Japan. 
Email: \url{sato.hiroto.s9@f.mail.nagoya-u.ac.jp}.
}
\and
\Large Konan Shimizu\thanks{\protect \onehalfspacing 
Faculty of Economics, Keio University, 2-15-45 Mita, Minato-ku, Tokyo 108-8345, Japan.
Email: \url{shimizu-konan@keio.jp}.}
} 
\date{\today}

\begin{document}

\begin{titlepage}
\maketitle
\begin{abstract}
    This study extends \citeapos{blackwell1953equivalent} comparison of information to a sequential social learning model in which agents make decisions sequentially based on both private signals and observed actions of others.
    In this context, we introduce a binary relation over information structures: 
    an information structure is {\it more socially valuable} than another if it yields higher expected payoffs for {\it all} agents, regardless of their preferences and equilibrium realizations.
    First, we establish that this binary relation is strictly stronger than the Blackwell order.
    Next, we provide a necessary and sufficient condition for our binary relation and propose a simpler sufficient condition that is easier to verify.
    We further explore comparisons of information structures in terms of long-run payoffs, limit welfare, and canonical binary environments.
\end{abstract}

\bigskip

\setcounter{page}{0}
\thispagestyle{empty}
\end{titlepage}

\newcolumntype{C}[1]{>{\centering\arraybackslash}p{#1}}

\setlength{\abovedisplayskip}{5pt}
\setlength{\belowdisplayskip}{5pt}

\section{Introduction} \label{sec:introduction}
In classical decision theory, an information source is considered more valuable than another if it enables an individual decision maker to make better choices under uncertainty.
This has been established by \citeapos{blackwell1953equivalent} comparison of information structures, which evaluates information structures based on whether a single agent would always prefer one over another, regardless of preferences.
Thus, the Blackwell order provides a preference-independent criterion for comparing the value of information sources for a single decision maker.

However, in many real-world settings, decision makers do not rely solely on their private signals; they also acquire information from the observed actions of others.
This creates an information externality: an individual’s decision not only influences their own outcome but also transmits information to future decision makers.
Because of this externality, simply comparing information structures based on their value for individual decision-making is no longer sufficient to evaluate the value of information in society.

A motivating example arises in online shopping.\footnote{A similar example is provided in \citet{arieli2023herd}.}
Buyers can obtain information about product quality from advertisements or product descriptions on each product page.
In addition to these signals, buyers often have access to information about the volume of past purchases.
When deciding whether to purchase a product, buyers use both their private signal and the information based on the popularity of previous buyers.
As a result, when assessing and comparing advertisements or product pages for their informational efficiency, it is essential to consider the information externalities created by the observability of past purchasing behavior.

To address the value of information in the society, we extend \citeapos{blackwell1953equivalent} comparison of information structures to the classical sequential social learning model \citep{banerjee1992simple,bikhchandani1992theory,smith2000pathological}.
In this model, homogeneous agents make decisions sequentially based on the past actions of others (referred to as history) and their own private signals.
These private signals are independently drawn from an identical information structure.\footnote{We discuss extensions to correlated and heterogeneous private signals, as well as heterogeneous decision problems, in Online Appendix \ref{sec:extension}.}
Within this framework, we introduce a binary relation over information structures:
an information structure is {\it more socially valuable} than another if it yields higher expected payoffs for {\it all} agents, regardless of their preferences and equilibrium realizations, in the presence of social learning.

Within this model, we pose a benchmark question: {\it Does a Blackwell more informative source always make all agents better off in a social learning environment?}
The answer is no.
We observe that our binary relation is strictly stronger than the Blackwell order (Proposition \ref{prop: observation_strict_inclusion}). 
This follows intuitively because the history garbles signal realizations depending on the underlying decision problem.
Consequently, our binary relation requires a sufficiently informative signal to ensure that the joint value of history and the private signal increases.  
This highlights an essential feature of the observability assumption.
If agents could observe past signal realizations instead of actions, then a Blackwell more informative signal would always be more socially valuable.
Thus, this strict gap between our binary relation and the Blackwell order ultimately arises from whether agents can observe past signal realizations or only past actions, and consequently, from the resulting information externalities.

We then ask the main question: {\it When is one information structure more socially valuable than another?}
Theorem \ref{thm_characterization} provides a necessary and sufficient condition for our binary relation.
Specifically, one information structure is more socially valuable than another if and only if it yields higher expected payoffs for all agents across all decision problems and equilibria, even in settings where past signals (rather than actions) are observable under the alternative information structure.
The necessary condition, combined with classical results, indicates that an information structure is more socially valuable than another only if it induces unbounded (private) beliefs.
Thus, if an information structure induces an information cascade, then it is no longer more socially valuable than any other information structure.
Equivalently, the existence of conclusive signals, or asymptotic learning, is necessary for all agents to be better off.

Given the restrictive necessary condition, it is natural to ask: {\it Which pairs of information structures can be compared within our binary relation?}
This question naturally directs our focus to the sufficiency part of Theorem \ref{thm_characterization}, but verifying this condition is challenging, as it depends on the underlying decision problem.
To address this, we provide a clear and simple sufficient condition.
Specifically, Theorem \ref{thm_sufficient_condition} states that an information structure is more socially valuable than another if there exists a mixture of full and no information between them in the Blackwell order.
To verify the existence of such a mixture, Proposition \ref{prop_equivalent_sufficient_condition} provides an equivalent condition.
By combining these results, we show that an information structure is more socially valuable than another if it assigns a sufficiently high probability of disclosing conclusive signals about each state.
Therefore, a relatively higher degree of conclusiveness is sufficient for all agents to be better off.

This sufficient condition follows from the intrinsic properties of mixtures of two extreme information structures. 
Under any such mixture, the expected payoffs for all agents match those in a setting in which agents observe past signals rather than actions for any equilibrium and any decision problem.  
Moreover, any such mixture respects the Blackwell order:
If the mixture is Blackwell more informative than another information structure, then this mixture is more socially valuable.
Conversely, if an information structure is Blackwell more informative than the mixture, it is also more socially valuable.
Thus, if a mixture of full and no information exists between two information structures in the Blackwell order, they are also comparable in our binary relation.
Therefore, our proof finds and utilizes the unique property of mixtures of full and no information.
This avoids the complicated formula of expected payoffs for all agents, which is the main technical challenge of our analysis.

In Section \ref{sec:discussion}, we discuss several variants of our binary relation and their implications.
First, Section \ref{sec:eventual} introduces a binary relation based only on sufficiently later agents.
Formally, we say that one information structure is {\it eventually more socially valuable} than another if it yields weakly higher expected payoffs than another for all agents after some cutoff for all decision problems.
Proposition \ref{thm: asymptotic comparison} provides a characterization and highlights the distinction between this binary relation and the eventual Blackwell order by \citet{azrieli2014comment} and \citet{mu2021blackwell}, as shown in Proposition \ref{prop: large sample order and long run comparison}.
Moreover, by utilizing the eventual Blackwell order, Proposition \ref{thm: sufficient condition for long run comparison} provides a sufficient condition for this long-run comparison, which is weaker than the condition in Theorem~\ref{thm_sufficient_condition}.

We further introduce a binary relation comparing only the limit welfare, which is one of the central objectives in the social learning literature.
Specifically, we say that one information structure is {\it eventually more socially valuable in the limit} than another if it yields weakly higher expected payoffs than another in the limit for all decision problems.
Theorem \ref{thm: limit comparison} shows that one information structure is eventually more socially valuable in the limit than another if and only if the former induces unbounded beliefs or the latter is no information.
Thus, the necessary condition for the more socially valuable order becomes sufficient in this limit welfare comparison.

Next, our original binary relation requires the comparison to be robust to equilibrium multiplicity, that is, to hold for all equilibria, and therefore is not a partial order.
To relax this all-equilibria requirement, Section \ref{sec:weak_order} introduces a weaker binary relation.
We say that one information structure is {\it weakly more socially valuable} than another if it yields weakly higher payoffs for all agents under some equilibrium (rather than any equilibrium) than another.
This weaker version is a partial order.
Example \ref{example_1} also shows that this weaker relation is strictly stronger than the Blackwell order.
Furthermore, in Example \ref{example_weak}, we show that our original binary relation is strictly stronger than this weaker relation by providing a specific sufficient condition (Proposition \ref{observation_weak_order_3support}) that is independent of Theorem \ref{thm_sufficient_condition}.

Relatedly, in Section \ref{sec:restricted decision problem}, we relax the requirement of a universal domain of decision problems and focus on canonical binary environments that have been extensively studied in the social learning literature. 
Theorem \ref{prop: binary binary} shows that these environments exhibit a knife-edge property: for arbitrary binary information structures and generic binary decision problems, a Blackwell more informative signal weakly increases the expected payoff of every agent. 
Thus, the usual welfare intuition behind the Blackwell order survives social learning in this canonical setting, even though it fails more generally, as illustrated by Example \ref{example_1}. The reason is that, with binary signals and binary actions, an observed action either perfectly reveals the realized signal or conveys no additional information, precluding the adverse deterioration in the informational content of histories that underlies our negative examples.


\subsection{Related Literature} \label{sec:literature}
Pioneered by \citet{blackwell1951comparison,blackwell1953equivalent}, numerous studies have extended Blackwell’s comparison of experiments.\footnote{Some studies examine comparisons of experiments within a restricted domain of decision problems or a limited class of experiments \citep{lehmann1988comparing, persico2000information, athey2018value, ben2024new}.} 
Our study investigates comparisons in a game-theoretic setting, similar to \citet{lehrer2010signaling}, \citet{lehrer2013garbling}, \citet{gossner2000comparison}, \citet{pkeski2008comparison}, \citet{cherry2012strategically}, \citet{bergemann2016bayes}, and \citet[Section 6]{de2018blackwell}, but we focus specifically on the social learning model, where strategic interaction arises from information externalities rather than payoff externalities.

Beyond the game-theoretic setting, our study is closely related to two strands of literature on comparisons of experiments.
The first strand examines comparisons involving repeated samples \citep{stein1951notes, torgersen1970comparison, moscarini2002law, azrieli2014comment, mu2021blackwell}.
Although each agent in our model receives a private signal independently drawn from the identical information structure, they cannot observe past signal realizations. 
Thus, our analysis departs from the standard learning model, and our binary relation becomes strictly stronger than the Blackwell order.
This aspect highlights the observability assumption inherent in the social learning model.
We discuss the relation between (a weaker version of) our binary relation and the eventual Blackwell order in \citet{azrieli2014comment} and \citet{mu2021blackwell} in Section \ref{sec:eventual}.

The second strand explores comparisons of dynamic information structures in sequential decision problems, as studied by \citet{greenshtein1996comparison}, \citet[Section 5]{de2018blackwell}, \citet{renou2024comparing}, and \citet{whitmeyer2024comparisons}.\footnote{\citet{whitmeyer2024dynamic} also analyze comparisons in dynamic decision problems in the presence of additional information, following \citet{brooks2024comparisons}.}
Similar to these studies, the information observed by agents is correlated across periods, but in our model this correlation arises from the correlation of past actions.
However, unlike in previous studies, the information they observe crucially depends on past actions and the underlying decision problem.

Broadly, this study contributes to the literature on social learning.
Since, for example, \citet{banerjee1992simple}, \citet{bikhchandani1992theory}, and \citet{smith2000pathological},\footnote{For a recent comprehensive survey, see \citet{bikhchandani2024information}.} a fundamental question has been whether agents can eventually learn the true state in various settings.\footnote{Examples include cases with limited observations of past actions \citep{ccelen2004observational, acemoglu2011bayesian, lobel2015information, arieli2019multidimensional, arieli2021general, kartik2024beyond, xu2025social} (see also \citet{gale2003bayesian,banerjee2004word,callander2009wisdom,smith2013rational}), as well as cases where observing past actions is costly \citep{kultti2006herding, kultti2007herding, song2016social}, and cases involving the costly acquisition of private signals \citep{mueller2016social,ali2018herding}.}
Some recent studies, such as \citet{arieli2023herd} and \citet{arielipositive}, examine how information structures can be optimally designed or regulated to influence herding or asymptotic behavior (see also \citet{lorecchio2022persuading} and \citet{parakhonyak2023information}).\footnote{Other important questions include social learning with correlated signals \citep{liang2020complementary, awaya2025social}, the speed and efficiency of learning \citep{hann2018speed, rosenberg2019efficiency}, and learning about the informativeness \citep{huang2024learning}.}
To the best of our knowledge, the comparison of experiments, which is the focus of this study, remains largely unexplored in the literature.
The primary technical challenge arises from the complexity of expected payoffs when analyzing all agents, which is even more difficult than focusing solely on asymptotic agents.
Our approach addresses this issue by leveraging the properties of mixtures of full and no information.

\section{Model} \label{sec:model}
There is an infinite sequence of ordered, homogeneous agents indexed by $i=1,2,\dots$, who make decisions sequentially. 
The state space is binary, $\Omega=\{L,H\}$ with a common prior.\footnote{For simplicity, we assume a binary state space, but our results can be extended to a finite state space.
} 
Let $\mu_{0}\in (0,1)$ be the prior of $\omega=H$.   
The periods are discrete ($t=0,1,\dots$), and each agent $i$ takes an action at period $i$ from a finite action set $A$.
A common payoff function $u: A\times \Omega\to \mathbb{R}$ determines each agent's payoff.\footnote{We discuss extensions to a heterogeneous-agent model in Online Appendix \ref{sec:heterogeneous}.}
The payoff of agent $i$ depends solely on their own action and the state, independent of actions taken by other agents.

The timing of this game is as follows:
At period $0$, nature first determines the true state, which remains unchanged throughout the game. 
In each period $i$, agent $i$ first observes the entire history, which consists of the actions of all preceding agents ($1,2,\dots, i-1$). 
Additionally, agent $i$ receives a private signal $s\in S$, drawn independently from an identical information structure $\pi: \Omega\to\Delta(S)$.\footnote{We discuss relaxations of the assumptions of independent and identically distributed private signals in Online Appendix \ref{sec:extension}.}
For simplicity, we assume that $S$ is finite.\footnote{Although the proof holds even when both $A$ and $S$ are countable, we impose this assumption to simplify notation.}
Following these observations, agent $i$ selects an action from the action set $A$.

Given the decision problem $\mathcal{D}=(A,u)$ and the information structure $\pi:\Omega\to\Delta(S)$, the strategy of agent $i$ is denoted by $\sigma_{i}: A^{i-1}\times S \to \Delta(A)$.
Given $\mathcal{D}=(A,u)$, $\pi$, and the strategy profile $\bm{\sigma}=(\sigma_{i})_{i\in \mathbb{N}}$, let $\alpha_{\leq i}^{\omega}(\pi,\bm{\sigma})\in \Delta(A^{i})$ denote the distribution of actions taken by agents $1,2,\dots,i$ when the state is $\omega$, that is,
\begin{align*}
 \alpha_{\leq i}^{\omega}(\bm{a}|\pi,\bm{\sigma})=\sum_{(s_{1},\dots,s_{i}) \in S^{i}}\prod_{k=1}^{i}\sigma_{k}(a_{k}|a_{1},\dots,a_{k-1},s_{k})\pi(s_{k}|\omega).
\end{align*}
Similarly, let $\alpha_{i}^{\omega}(\pi,\bm{\sigma})\in \Delta(A)$ be the distribution of actions taken by agent $i$ when the state is $\omega$, that is,
\begin{align*}
 \alpha_{i}^{\omega}(a|\pi,\bm{\sigma})=\sum_{(a_{1}',\dots,a_{i-1}')\in A^{i-1}}\alpha_{\leq i}^{\omega} (a_{1}',\dots,a_{i-1}',a | \pi,\bm{\sigma}).
\end{align*}
Note that $\alpha_{i}^{\omega}(\pi,\bm{\sigma})$ does not depend on the strategies of agents after $i$.
Let $V_{i}^{\mathcal{D}}(\pi,\bm{\sigma})$ be the ex-ante expected payoff for agent $i$. 
Precisely,
\begin{align*}
    V_{i}^{\mathcal{D}}(\pi,\bm{\sigma})=\mathbb{E}_\omega\left[\sum_{a\in A} \alpha_{i}^{\omega}(a|\pi,\bm{\sigma})u(a,\omega)\right].
\end{align*}
We say that the strategy profile $\bm{\sigma}^{*}$ is a Bayes-Nash equilibrium (hereafter referred to simply as an equilibrium) under $(\mathcal{D},\pi)$ if
\begin{align*}
    V_{i}^{\mathcal{D}}(\pi,\bm{\sigma}^{*})\geq V_{i}^{\mathcal{D}}(\pi,(\sigma_{i},\bm{\sigma}_{-i}^{*}))
\end{align*}
for all $\sigma_{i}$ and $i$.

For two information structures $\pi:\Omega\to\Delta(S)$ and $\pi':\Omega\to\Delta(S')$, define their product $\pi \otimes\pi':\Omega\to\Delta(S\times S')$ as
\begin{align*}
    (\pi\otimes \pi')((s,s')|\omega)=\pi(s|\omega)\pi'(s'|\omega)
\end{align*}
for all $s\in S$, $s'\in S'$, and $\omega\in\Omega$.
We denote
\[
\pi^{\otimes i} = \pi \otimes \dots \otimes \pi
\]
as the information structure generated by $i$ conditionally independent observations from $\pi$.
Define $\overline{V}_{i}^{\mathcal{D}}(\pi)$ as
\begin{align*}
  \overline{V}_{i}^{\mathcal{D}}(\pi)=\max_{\sigma_{i}:S^{i}\to\Delta(A)} \mathbb{E}_\omega\left[\sum_{a\in A}\sum_{\bm{s}\in S^{i}}\sigma_{i}(a|\bm{s}) \pi^{\otimes i}(\bm{s}|\omega)u(a,\omega)\right]. 
\end{align*}
In other words, this represents the maximized expected payoff when agent $i$ independently observes the signal drawn from $\pi$ for $i$ times.

Given the information structure $\pi:\Omega\to\Delta(S)$, define $\mu\in \Delta[0,1]$ as the 
{\it private belief distribution} induced by $\pi$, which represents the distribution over private beliefs about the state being $H$ after observing private signals from $\pi$.
Specifically, for $x\in[0,1]$, 
\begin{align*}
    \mu(x)=\mathbb{E}_{\omega}\left[\sum_{s\in S(x)}\pi(s|\omega)\right],
\end{align*}
where $S(x)=\{s\in S\mid \frac{\mu_{0}\pi(s|H)}{\mu_{0}\pi(s|H)+(1-\mu_{0})\pi(s|L)}=x\}$.
With a slight abuse of notation, we define $\pi(\mu=x|\omega)=\sum_{s\in S(x)}\pi(s|\omega)$ for each $\omega\in \Omega$.
We say that a signal $s$ is a {\it conclusive signal about} $\omega=H$ (resp. $\omega=L$) if $s\in S(1)$ (resp. $s\in S(0)$).
Additionally, we say that an information structure $\pi$ is {\it no information} if $\supp{\mu}=\{\mu_{0}\}$, and that $\pi$ is {\it full information} if $\supp{\mu}=\{0,1\}$.
Slightly abusing the terminology, $\pi$ is a {\it mixture of full and no information} if $\supp{\mu}=\{0,\mu_{0},1\}$.\footnote{
This definition differs slightly from the usual definition of a mixture of information structures.
Formally, we say that $\pi$ is a {\it mixture} of $\pi'$ and $\pi''$ if $\mu=\lambda \mu' + (1-\lambda)\mu''$ for some $\lambda\in(0,1)$.
When $\supp{\pi'}\cap \supp{\pi''}=\emptyset$, this definition coincides with the standard one.
}
Given $\pi$, $\bm{\sigma}$, and $i\geq 2$, define $\rho_{i}\in \Delta[0,1]$ as the {\it public belief distribution}, which is the distribution over public beliefs regarding the state being $H$ induced by observing actions taken by agents $1,\dots,i-1$.
Formally, for each $x\in[0,1]$, 
\begin{align*}
    \rho_{i}(x)=\mathbb{E}_{\omega}\left[\sum_{\bm{a}\in A^{i-1}(x)}\alpha_{\leq i-1}^{\omega}(\bm{a}|\pi,\bm{\sigma})\right],
\end{align*}
where $A^{i-1}(x)=\{\bm{a}\in A^{i-1} \mid \frac{\mu_{0} \alpha_{\leq i-1}^{H}(\bm{a}|\pi,\bm{\sigma})}{\mu_{0} \alpha_{\leq i-1}^{H}(\bm{a}|\pi,\bm{\sigma})+(1-\mu_{0}) \alpha_{\leq i-1}^{L}(\bm{a}|\pi,\bm{\sigma})}=x\}$.
Then, based on their private and public beliefs, agents update their {\it posterior beliefs} about $\omega=H$ according to Bayes' rule.\footnote{
Given a public belief $x$ and a private belief $y$, the posterior belief is 
$
\frac{xy}{xy+\frac{\mu_{0}}{1-\mu_{0}}(1-x)(1-y)}.
$
}

\section{Results} \label{sec:main}
For comparison, we introduce the standard Blackwell order, denoted $\pi\succsim_{B} \pi'$, when $\pi$ is Blackwell more informative than $\pi'$, that is, when $\pi'$ is a garbling of $\pi$.\footnote{Formally, $\pi'$ is a garbling of $\pi$ if there exists a Markov kernel $\gamma:S\to \Delta(S')$ such that $\pi'(s'|\omega)=\sum_{s\in S}\gamma(s'|s)\pi(s|\omega)$.}
By the Blackwell's classical theorem, $\pi \succsim_{B} \pi'$ if and only if the single decision maker (agent $1$ in our model) prefers $\pi$ over $\pi'$ for all decision problems.

By contrast, our primary focus is on the following binary relation:
\begin{definition}\label{def: socially valuable order}
    $\pi$ is \textit{more socially valuable} than $\pi'$, denoted $\pi\succsim_{S} \pi'$, if for any decision problem $\mathcal{D}=(A,u)$, $V_{i}^{\mathcal{D}}(\pi,\bm{\sigma}^{*})\geq V_{i}^{\mathcal{D}}(\pi',\bm{\sigma}^{**})$ for any agent $i$, any equilibrium $\bm{\sigma}^{*}$ under $(\mathcal{D},\pi)$, and any equilibrium $\bm{\sigma}^{**}$ under $(\mathcal{D},\pi')$.
\end{definition}
Thus, $\succsim_{S}$ requires that \textit{all} agents weakly prefer $\pi$ over $\pi'$ in terms of expected payoff, across \textit{all} decision problems and \textit{all} equilibria.\footnote{In Section \ref{sec:eventual}, we introduce a weaker binary relation that requires all asymptotic agents to prefer one information structure to another for every decision problem.} 
This Pareto-type order is natural from a societal perspective and provides insights for the social learning literature.
Specifically, this strong requirement of $\succsim_{S}$ implies that if $\pi$ is more socially valuable than $\pi'$, then $\pi$ achieves higher welfare, defined as the discounted sum of payoffs, than $\pi'$ for all discount factors, decision problems, and equilibrium realizations.\footnote{
We discuss the comparison of limit average discounted payoffs for discount factors sufficiently close to $1$ in Section \ref{sec:eventual}.
}
Additionally, $\pi \succsim_{S} \pi'$ indicates that learning occurs faster under $\pi$ than under $\pi'$ for all decision problems, connecting our binary relation $\succsim_{S}$ to the literature on the speed of learning in social learning \citep{hann2018speed, rosenberg2019efficiency}.

Definition \ref{def: socially valuable order} also requires that every equilibrium under $\pi$ yields a higher payoff than every equilibrium under $\pi'$. Thus, our criterion is robust to equilibrium multiplicity.
In this sense, $\succsim_{S}$ provides a robust welfare-based criterion for comparing information structures in the social learning.\footnote{
Section \ref{sec:weak_order} examines a weaker binary relation that relaxes the requirement that the comparison be robust to equilibrium multiplicity.
Relatedly, Section \ref{sec:restricted decision problem} focuses on a restricted domain of decision problems, particularly those involving binary actions.
}
Consider the online shopping setting in the Introduction wherein the outside observer evaluates the informational efficiency of advertisements or product descriptions.
Then, they usually cannot predict either users' decision problems or the equilibrium that will be realized.
In this setting, our definition of $\succsim_{S}$ captures the most pessimistic and robust evaluation of the value of information under $\pi$.

The binary relation $\succsim_{S}$ in Definition \ref{def: socially valuable order} is defined with respect to a fixed prior $\mu_{0}$.
As in the Blackwell order, however, this dependence is inconsequential.
Specifically, if $\pi \succsim_{S} \pi'$ holds for one prior, then it holds for every prior, as can be shown by appropriately rescaling the decision problem.

Our first observation establishes that our binary relation is strictly stronger than the Blackwell order.
In other words, if one information structure is more socially valuable than another, then it is also Blackwell more informative, but not vice versa.
\begin{proposition} \label{prop: observation_strict_inclusion}
    $\succsim_{S}$ is a strictly stronger binary relation than $\succsim_{B}$.
\end{proposition}
Note that $\pi \succsim_{S} \pi'$ implies that $\pi\succsim_{B} \pi'$ as agent $1$ prefers $\pi$ over $\pi'$ for all decision problems. 
Thus, $\succsim_{S}$ is weakly stronger than $\succsim_{B}$.
To complete the proof, we show the following example in which $\pi\succsim_{B} \pi'$ holds, but $\pi\succsim_{S} \pi'$ does not.
\begin{example}\label{example_1}
    Consider two information structures $\pi,\pi':\Omega\to \Delta(\{s_{0},s_{1},s_{2}\})$ given by the following tables, where $0<\delta<\varepsilon<\varepsilon'<1$:
\begin{table}[H]
\centering
\begin{tabular}{c@{\hspace{1.5cm}}c}

\begin{tabular}{l|ccc}
\Xhline{1.2pt}
$\pi$ & $s_{0}$ & $s_{1}$ & $s_{2}$ \\
\hline
$H$ & $0$ & $1-\varepsilon$ & $\varepsilon$ \\
$L$ & $1-\delta$ & $0$ & $\delta$ \\
\Xhline{1.2pt}
\end{tabular}

&
\begin{tabular}{l|ccc}
\Xhline{1.2pt}
$\pi'$ & $s_{0}$ & $s_{1}$ & $s_{2}$ \\
\hline
$H$ & $0$ & $1-\varepsilon'$ & $\varepsilon'$ \\
$L$ & $1-\delta$ & $0$ & $\delta$ \\
\Xhline{1.2pt}
\end{tabular}
\end{tabular}
\end{table}
    \noindent
    Then, we have $\pi\succsim_{B}\pi'$.
    Now, consider the following decision problem $\mathcal{D}=(A,u)$: $A=\{a_{0},a_{1}\}$, $u(a_{0},H)=u(a_{0},L)=0$, $u(a_{1},H)=1-r$, and $u(a_{1},L)=-r$, where $r\in(\frac{\mu_{0}\varepsilon}{\mu_{0}\varepsilon+(1-\mu_{0})\delta},\min\{\frac{\mu_{0}\varepsilon'}{\mu_{0}\varepsilon'+(1-\mu_{0})\delta},\frac{\mu_{0}\varepsilon^{2}}{\mu_{0}\varepsilon^{2}+(1-\mu_{0})\delta^{2}}\})$.\footnote{Note that $\frac{\mu_{0}\varepsilon^{2}}{\mu_{0}\varepsilon^{2}+(1-\mu_{0})\delta^{2}}>\frac{\mu_{0}\varepsilon}{\mu_{0}\varepsilon+(1-\mu_{0})\delta}$ since $\varepsilon>\delta>0$.}
    
    Take any equilibrium $\bm{\sigma}^{*}$ under $(\mathcal{D},\pi)$.
    First, agent $1$ chooses action $a_{1}$ if and only if she receives $s=s_{1}$ as $r>\frac{\mu_{0}\varepsilon}{\mu_{0}\varepsilon+(1-\mu_{0})\delta}$.
    Then, agent 2 can perfectly infer that the true state is $H$ when either $s=s_{1}$ or agent $1$ has chosen action $a_{1}$.
    Otherwise, agent 2's posterior belief is below the cutoff $r$.
    Thus, agent $2$ chooses action $a_{1}$ if and only if she perfectly knows that the true state is $H$.
    Similarly, agent $i$ chooses action $a_{1}$ if and only if (i) $s=s_{1}$ or (ii) at least one agent before $i$ has chosen action $a_{1}$, as agent $i$ can perfectly infer that the true state is $H$ under these cases, while the posterior beliefs would otherwise be below $r$. 
    Thus, agent $i$'s expected payoff is $V_{i}^{\mathcal{D}}(\pi,\bm{\sigma}^{*})=\mu_{0}(1-\varepsilon^{i})(1-r)$.\footnote{The formal proof is provided in Lemma \ref{lemma_expected_payoff_special_case} in the Appendix.}
    
    Under $(\mathcal{D},\pi')$, there is an equilibrium in which agent $i$ takes $a_{0}$ if and only if he receives $s_{0}$ or at least one agent before $i$ has taken $a_{0}$ since $\frac{\mu_{0}\varepsilon'}{\mu_{0}\varepsilon'+(1-\mu_{0})\delta}>r$. 
    Let $\bm{\sigma}^{**}$ denote this equilibrium strategy profile.
    In this equilibrium, the ex-ante expected payoff of agent $i$ is $V_{i}^{\mathcal{D}}(\pi',\bm{\sigma}^{**})=\mu_{0}(1-r)-(1-\mu_{0})\delta^{i}r$.
    
    Thus, the difference in payoffs of agent $i$ $(\geq 2)$ is
    \begin{align*}
        V_{i}^{\mathcal{D}}(\pi',\bm{\sigma}^{**})-V_{i}^{\mathcal{D}}(\pi,\bm{\sigma}^{*})=&\mu_{0}\varepsilon^{i}(1-r)-(1-\mu_{0})\delta^{i}r\\
        =&\mu_{0}\varepsilon^{i}\left(1-\frac{\mu_{0}\varepsilon^{i}+(1-\mu_{0})\delta^{i}}{\mu_{0} \varepsilon^{i}}r\right)\\
        \geq& \mu_{0}\varepsilon^{i} \left(1- \frac{\mu_{0}\varepsilon^{2}+(1-\mu_{0})\delta^{2}}{\mu_{0} \varepsilon^{2}}r\right)\\
        >& 0.
    \end{align*}
    Therefore, $\pi \succsim_{B} \pi'$ but not $\pi\succsim_{S}\pi'$. \qed
\end{example}
Proposition \ref{prop: observation_strict_inclusion} intuitively follows because past actions provide coarser information than past signal realizations.
Consequently, our binary relation requires the information structure to be sufficiently informative to ensure that the joint value of history and private signals increases.
By contrast, if agents could observe past signal realizations instead of actions, then a Blackwell more informative signal would always be more socially valuable.
Essentially, the gap between $\succsim_{S}$ and $\succsim_{B}$ is driven by information externalities arising from whether agents can observe past signal realizations or only past actions.

In the setting described in Example \ref{example_1}, when signals are observable, the expected payoffs of agent $i\geq 2$ under $\pi$ and $\pi'$ are identical in this example.  
If past signals were observable, agent $i$ receiving $s = s_{2}$ would choose $a_{1}$ whenever all preceding agents also received $s = s_{2}$.  
However, in the observable action setting, agent $i$ with $s = s_{2}$ would choose $a_{0}$ if all predecessors had selected $a_{0}$ under $\pi$, even when all preceding agents receive $s = s_{2}$.
In contrast, under $\pi'$, agents can behave as if they can perfectly observe the past signal realizations.
Therefore, all agents except for agent $1$ strictly prefer $\pi'$ to $\pi$ in this decision problem.

How strong is our binary relation relative to the Blackwell order?
To answer this, we provide a characterization as follows:
\begin{theorem}[Characterization]\label{thm_characterization}
    $\pi\succsim_{S}\pi'$ holds if and only if
    \begin{align*}
        V_{i}^{\mathcal{D}}(\pi,\bm{\sigma}^{*})\geq \overline{V}_{i}^{\mathcal{D}}(\pi')
    \end{align*}
    for any decision problem $\mathcal{D}$, any agent $i$, and any equilibrium $\bm{\sigma}^{*}$ under $(\mathcal{D},\pi)$.
\end{theorem}
Thus, by Theorem \ref{thm_characterization}, one information structure is more socially valuable than another if and only if it yields higher expected payoffs for all agents, decision problems, and equilibria, even when past signals are observable under the alternative information structure.

The sufficiency part of Theorem~\ref{thm_characterization} is immediate, as the observable signal setting provides an upper bound on any equilibrium payoffs (Lemma \ref{lemma_signal_is_more_informative}).
For the necessity part of Theorem~\ref{thm_characterization}, we construct a decision problem in which agents can perfectly infer past signal realizations under $\pi'$. 
Specifically, take any decision problem $\mathcal{D} = (A, u)$ and an equilibrium $\bm{\sigma}^{*}$ under $\pi$. 
Then, we construct an auxiliary decision problem $\overline{\mathcal{D}}=(\overline{A},\overline{u})$ where $\overline{A}=\{(a,k)\mid a\in A,k\in S'\}$ and $\overline{u}((a, k), \omega)=u(a, \omega)$ for all $a \in A$, $\omega \in \Omega$, and $k\in S'$.
Next, we construct an equilibrium $\bm{\sigma}$ under $(\overline{\mathcal{D}},\pi)$ in which agents choose optimal actions from among the replicated actions, independently of their private signals $s\in S$. 
By contrast, under $(\overline{\mathcal{D}},\pi')$, we consider an equilibrium $\bm{\sigma}'$ in which each agent selects an optimal action that perfectly reflects their private signal realization $s'\in S'$.
Thus, the expected payoff under $\pi$ with $\bm{\sigma}$ is the same as in the original equilibrium $\bm{\sigma}^{*}$. 
However, under $\pi'$ with $\bm{\sigma}'$, the expected payoff corresponds to that in the observable signal setting.

Intuitively, by sufficiently enlarging the action space, we can reinterpret agents’ actions as messages that convey private signals in a cheap-talk environment. 
In this setting, the equilibrium selection effectively determines how informative these messages become. 
Under $\pi'$, we can select a truth-telling equilibrium in which each agent breaks ties in a manner that perfectly reveals their private signal realization. 
By contrast, under $\pi$, we construct a babbling equilibrium in which agents’ actions convey no additional information relative to the original setting. 
This comparison illustrates that when the action space is rich enough, the observable signal setting under $\pi'$ can be induced as an equilibrium outcome that achieves maximal information revelation, whereas the equilibrium under $\pi$ corresponds to the original observable action setting. 
Hence, by leveraging this strong equilibrium selection rule, this construction highlights the essential role of information externalities in determining how effectively private information is aggregated and transmitted in social learning, in its most extreme form.

By combining the classical result of \citet{smith2000pathological}, we derive a simple necessary condition from Theorem \ref{thm_characterization}.
We say that an information structure $\pi$ induces {\it unbounded beliefs} if $\text{co}(\supp{\mu})=[0,1]$.\footnote{When the signal space is finite, $\pi$ induces unbounded beliefs if and only if there exists a conclusive signal about each state.}
Since agents can eventually learn the true state in an observable signal setting, we obtain the following necessary condition:
\begin{corollary}[Necessary condition]\label{observation_necessity_of_conclusive}
    Suppose that $\pi'$ is not no information.
    If $\pi\succsim_{S} \pi'$, then $\pi$ induces unbounded beliefs.
\end{corollary}
Corollary \ref{observation_necessity_of_conclusive} states that, except in the trivial case, an information structure must induce unbounded beliefs to be more socially valuable than another.
Thus, if an information cascade occurs under a given information structure, it is no longer more socially valuable than any other information structure except in certain trivial cases.

For the proof sketch, we construct a decision problem in which only almost conclusive signals yield a positive payoff.
If $\pi$ does not generate unbounded beliefs, the equilibrium payoff becomes zero.
However, by Theorem \ref{thm_characterization}, we can construct an auxiliary problem in which the equilibrium payoff under $\pi$ remains unchanged, while $\pi'$ yields the payoff corresponding to the observable signal setting.
In this auxiliary problem, whenever $\pi'$ is not no information, repeated observations of past signal realizations induce asymptotic learning and thus strictly positive expected payoffs.
Hence, inducing unbounded beliefs under $\pi$ is a necessary condition for $\succsim_{S}$.\footnote{The construction of the original decision problem relies on “test” problems, which correspond to the extreme points of convex functions in the binary state space.
We leave open the question of whether considering only test problems is sufficient for the original problem (rather than the auxiliary one) when verifying $\succsim_{S}$.}

Theorem \ref{thm_characterization} and Corollary \ref{observation_necessity_of_conclusive} underscore the strong requirements inherent in our binary relation.
This naturally gives rise to the question: Which pairs of information structures can be compared within our binary relation?
Accordingly, we shift our focus to the sufficiency part of Theorem \ref{thm_characterization}.
However, verifying this condition is challenging, as it depends on the underlying decision problem.
Moreover, we can see that the necessary condition in Corollary \ref{observation_necessity_of_conclusive} is not a sufficient condition by Example \ref{example_1}.
To address this, we provide a sufficient condition that can be verified directly from the information structures.
\begin{theorem}[Sufficient condition]\label{thm_sufficient_condition}
    If there exists $\pi''$ such that $\supp{\mu''}=\{0,\mu_{0},1\}$ and $\pi \succsim_{B}\pi''\succsim_{B} \pi'$, then $\pi \succsim_{S} \pi'$.
\end{theorem}
Thus, Theorem \ref{thm_sufficient_condition} indicates that $\pi$ is more socially valuable than $\pi'$ if there exists a mixture of full and no information $\pi''$ such that $\pi \succsim_{B}\pi''\succsim_{B} \pi'$.
To verify the existence of such a mixture, we provide an equivalent condition.
\begin{proposition}\label{prop_equivalent_sufficient_condition}
     There exists $\pi''$ such that $\supp{\mu''}=\{0,\mu_{0},1\}$ and $\pi \succsim_{B}\pi''\succsim_{B} \pi'$ if and only if $\pi$ and $\pi'$ satisfy
         \begin{align*}
             1- \sum_{s\in \supp{\pi'}}\min\{\pi'(s|L),\pi'(s|H)\}\leq \min \{\pi(\mu=0|L),\pi(\mu=1|H)\}.
         \end{align*}
\end{proposition}
Recall that the necessary condition in Theorem \ref{thm_characterization} requires that $\pi$ induces unbounded beliefs if $\pi$ is more socially valuable than $\pi'$ and $\pi'$ is not no information.
Then, by Proposition \ref{prop_equivalent_sufficient_condition}, the sufficient condition in Theorem \ref{thm_sufficient_condition} indicates that $\pi$ is more socially valuable than $\pi'$ if $\pi$ assigns a sufficiently high probability to disclosing conclusive signals about each state.

As a remark, the sufficient condition in Theorem \ref{thm_sufficient_condition} is independent of the prior $\mu_{0}$, as the equivalent inequality condition in Proposition \ref{prop_equivalent_sufficient_condition} does not depend on $\mu_{0}$.
Thus, if there exists a mixture of full and no information between $\pi$ and $\pi'$ in the Blackwell order for some prior, then such a mixture exists for every prior.

For example, for any information structure $\pi'$, we can construct an information structure $\pi$ by taking a mixture of $\pi'$ and $\pi^{full}$ where $\pi^{full}$ is full information.
Then, there exists $\lambda\in (0,1)$ such that $\pi(\mu=x|\omega)=(1-\lambda)\pi'(\mu'=x|\omega)$ and
$\pi(\mu=y|\omega)=(1-\lambda)\pi'(\mu'=y|\omega)+\lambda\pi^{full}(\mu=y|\omega)$ for all $x\neq 0,1$, $y\in \{0,1\}$ and $\omega$. 
We can see that if $\lambda$ is above the threshold, $\pi\succsim_{S} \pi'$ holds.\footnote{
Note that $\min\{\pi(\mu=0|L),\pi(\mu=1|H)\}=\lambda+(1-\lambda)\min\{\pi'(\mu'=0|L),\pi'(\mu'=1|H)\}$.
Thus, the inequality in Proposition~\ref{prop_equivalent_sufficient_condition} holds if
\begin{align*}
    \lambda\geq 1-\frac{\sum_{s\in \supp{\pi'}}\min\{\pi'(s|L),\pi'(s|H)\}}{1-\min\{\pi'(\mu'=0|L),\pi'(\mu'=1|H)\}}.
\end{align*}
}
One of the most commonly studied classes of information structures in the social learning literature is the class of symmetric binary information structures.
Suppose that $\mu_{0}=1/2$.
A symmetric binary signal $\pi'$ consists of the binary signals $S'=\{s_{l},s_{h}\}$ with disclosure rule $\pi'(s_{l}|L)=\pi'(s_{h}|H)=1-p$ where $p\in[0,1/2]$.
Then, $\pi'$ is parameterized by the single parameter $p$, where lower $p$ implies a more informative signal in the Blackwell order.
Now, we define $\pi$ as the mixture of $\pi^{full}$ and $\pi'$ with probability $\lambda$ and $1-\lambda$.
By the previous discussion, $\pi$ is more socially valuable than $\pi'$ if $\lambda \geq 1-2p$.
Thus, $\pi$ is more socially valuable than $\pi'$ if $\pi$ assigns a sufficiently high probability to disclose full information relative to the informativeness of $\pi'$.

The formal proof of Theorem \ref{thm_sufficient_condition} is complex and provided in the Appendix.
The key step focuses on the intrinsic properties of mixtures of full and no information.
Specifically, if such a mixture is Blackwell more informative than another information structure, it is also more socially valuable (Lemma \ref{observation_strict_order_3support}).  
Moreover, if an information structure is Blackwell more informative than the mixture, it is also more socially valuable (Lemma \ref{observation_imitation}).
Therefore, whenever a mixture of full and no information exists between two information structures in the Blackwell order, they remain comparable in our binary relation.

We now briefly explain why the Blackwell order with a mixture of full and no information implies our binary relation.
The proof of Lemma \ref{observation_strict_order_3support} proceeds as follows.
First, as shown in Lemma \ref{lemma_expected_payoff_3support}, under any mixture of full and no information, all agents can achieve the same expected payoff as if they had observed past signal realizations, even though they cannot directly infer their predecessors' private signals.
Intuitively, this follows for the following reasons.
First, under any mixture of full and no information, it is optimal for all agents in any decision problem to choose the optimal action when receiving a conclusive signal and to mimic the immediately preceding agent when receiving an uninformative signal.
Second, this strategy profile ensures that agents act as if they could observe past signals, both when (i) some previous agent or the agent herself has received a conclusive signal, and when (ii) all agents have received uninformative signals.\footnote{This feature is nontrivial because even a slight deviation in the support of the private belief distribution from that of the mixture can result in decision problems and equilibria that violate this property, as can be inferred from the proof of Proposition \ref{observation_weak_order_3support} in Section \ref{sec:weak_order}.}
Given the above discussion, if $\pi''$ is Blackwell more informative than $\pi'$ and $\pi''$ is a mixture of full and no information, then the expected payoff of agent $i$ under $\pi''$ is weakly higher than that under $i$ conditionally independent observations of $\pi'$ (that is, $\pi'^{\; \otimes i}$).  
Since past signals are always Blackwell more informative than history (Lemma \ref{lemma_signal_is_more_informative}), this expected payoff remains higher than that in any equilibrium under $\pi'$.

For the second step, Lemma \ref{observation_imitation} constructs a strategy profile under $\pi$ that achieves a lower bound on any equilibrium payoff under $\pi$.
Additionally, this strategy profile induces the same expected payoff as that under $\pi''$ for any equilibrium when $\pi''$ is a mixture of full and no information. 
Intuitively, the construction follows this logic:  
Consider any equilibrium strategy under $\pi$.  
First, any other strategy weakly decreases the agent's payoff due to the equilibrium condition.  
In particular, take a strategy in which agent $i$ behaves as if she observes $\pi''$ rather than $\pi$.
Since $\pi''$ is a mixture of full and no information, such a strategy involves choosing the optimal actions upon receiving conclusive signals about each state and mimicking agent $i-1$'s action otherwise.  
Given this, we further modify agent $i-1$'s strategy to follow the same one.  
This change decreases agent $i-1$'s expected payoff, which, in turn, reduces agent $i$'s (conditional) payoff from mimicking agent $i-1$ as the private belief coincides with the prior.  
Repeating this process yields a strategy profile that induces the lower bound of any equilibrium payoff under $\pi$.
Moreover, this lower bound coincides with the expected payoff under $\pi''$ for any equilibrium by construction and Lemma \ref{lemma_expected_payoff_3support}.

Note that, in the single-agent decision problem, it is sufficient to construct a strategy in which an agent can behave as if she observes $\pi''$ rather than $\pi$ to prove that $\pi$ is Blackwell more informative than $\pi''$.
This follows from the free-disposal property of information, which underlies the logic of garbling.
By contrast, this is not sufficient in our social learning setting because of the information externalities.
In the proof of Lemma \ref{observation_imitation}, we further utilize the property of the equilibrium strategy under the mixtures of full and no information $\pi''$.
Then, behaving as if an agent observes $\pi''$ rather than $\pi$ negatively affects the payoffs of future agents, and thus, $\pi$ is more socially valuable than $\pi''$.

As a remark, we do not know whether the converse of Theorem \ref{thm_sufficient_condition} holds, or how close this sufficient condition is to being necessary in general.
The main difficulty, as discussed in the Conclusion, lies in conducting a general analysis of learning speeds in each period.
By circumventing this difficulty, one can derive a necessary condition using techniques from large deviation theory, when focusing only on sufficiently late agents and specific decision problems.
Although the resulting necessary condition might not be tight, the following Example \ref{example: converse} indicates that the sufficient condition in Theorem \ref{thm_sufficient_condition} is close to being necessary to some extent.
\begin{example} \label{example: converse}
    Suppose that $\mu_{0}=1/2$.
    We now focus on two mixtures of full information and symmetric binary information structures.
    Consider two information structures $\pi,\pi':\Omega\to\Delta(S)$, where $S=\{s_{0},s_{l},s_{h},s_{1}\}$, defined as follows:
    \begin{table}[H]
    \centering
    \setlength{\tabcolsep}{3pt}
    \begin{tabular}{c@{\hspace{0.5cm}}c}
    \begin{tabular}{l|cccc}
    \Xhline{1.2pt}
    $\pi$ & $s_{0}$ & $s_{l}$ & $s_{h}$ & $s_{1}$ \\
    \hline
    $H$ & $0$ & $(1-\lambda)p$ & $(1-\lambda)(1-p)$ & $\lambda$ \\
    $L$ & $\lambda$ & $(1-\lambda)(1-p)$ & $(1-\lambda)p$ & $0$ \\
    \Xhline{1.2pt}
    \end{tabular}
    &
    \begin{tabular}{l|cccc}
    \Xhline{1.2pt}
    $\pi'$ & $s_{0}$ & $s_{l}$ & $s_{h}$ & $s_{1}$ \\
    \hline
    $H$ & $0$ & $(1-\lambda')p'$ & $(1-\lambda')(1-p')$ & $\lambda'$ \\
    $L$ & $\lambda'$ & $(1-\lambda')(1-p')$ & $(1-\lambda')p'$ & $0$ \\
    \Xhline{1.2pt}
    \end{tabular}
    \end{tabular}
    \end{table}
    \noindent
    Here, $\lambda\in[0,1]$ and $\lambda'\in[0,1]$ denote the probabilities of receiving conclusive signals, while $p\in(0,1/2)$ and $p'\in(0,1/2)$ capture the accuracy of the symmetric binary signals. We assume that $\lambda>\lambda'$ and $p<p'$, so that $\pi\succsim_B\pi'$.

    Consider the following (original) decision problem $\mathcal{D} = (A, u)$.
    There are two actions, $A = \{a_{0}, a_{1}\}$, with payoffs given by $u(a_{0}, L) = u(a_{0}, H) = 0$, $u(a_{1}, H) = 1 - r$, and $u(a_{1}, L) = -r$, where $r \in (1 - p, 1)$.
    Then, by Theorem \ref{thm_characterization}, we can construct an auxiliary problem $\overline{\mathcal{D}}$ such that the expected payoff under $\pi$ is the same as in the original problem, while the expected payoff under $\pi'$ coincides with that in the observable-signal setting.

    Under the assumption that $r > 1 - p$, the payoffs under $\pi$ can be computed straightforwardly.
    Specifically, there exists a unique equilibrium $\bm{\sigma}$ under $\pi$, which yields a payoff of $V^{\overline{\mathcal{D}}}_{i}(\pi, \bm{\sigma}) = \frac{1}{2}[1 - (1 - \lambda)^{i}](1 - r)$ for each $i$.
    In the Appendix, we derive $\overline{V}^{\overline{\mathcal{D}}}_{i}(\pi')$ by focusing on sufficiently large $i$ and applying large deviation theory.
    The derivation implies that $V^{\overline{\mathcal{D}}}_{i}(\pi,\bm{\sigma}) \geq \overline{V}^{\overline{\mathcal{D}}}_{i}(\pi')$ for all $r\in(1-p,1)$ and sufficiently large $i$ if and only if 
    \[
    \lambda \geq 1-2(1-\lambda')\sqrt{p'(1-p')},
    \]
    where this is a necessary condition for $\pi \succsim_{S} \pi'$.
    
    Note that Theorem \ref{thm_sufficient_condition} and Proposition \ref{prop_equivalent_sufficient_condition} together imply that our sufficient condition is equivalent to $\lambda \geq 1 - 2(1 - \lambda')p'$.
    Thus, by comparing with the above necessary condition, this condition is, to some extent, close to being necessary.
    In particular, our sufficient condition becomes tight when $p'$ is large.\footnote{As $p' \to \frac{1}{2}$, the inequality converges to $\lambda \geq \lambda'$.
    This is consistent with Theorem \ref{thm_sufficient_condition}, since $\pi'$ approaches a mixture of full and no information.}
    \qed
\end{example}

In Example \ref{example: converse}, $\pi \succsim_{S} \pi'$ does not hold when $\pi$ and $\pi'$ have the same probability of disclosing the conclusive signals.
This observation holds more generally and yields a specific condition under which the converse of Theorem \ref{thm_sufficient_condition} holds, as shown below.
\begin{corollary} \label{prop: converse}
    Suppose that $\pi(\mu=1|H)=\pi'(\mu'=1|H)>0$ and $\pi(\mu=0|L)=\pi'(\mu'=0|L)>0$.
    Then, $\pi \succsim_{S} \pi'$ (if and) only if there exists $\pi''$ such that $\supp{\mu''}=\{0,\mu_{0},1\}$ and $\pi \succsim_{B} \pi'' \succsim_{B} \pi'$.
\end{corollary}

\section{Alternative Social Value Orders} \label{sec:discussion}
\subsection{Long-Run Comparison} \label{sec:eventual}
Our original binary relation  appears strong, as it requires that all agents prefer one information structure to another.
A plausible alternative definition would require only that all sufficiently late agents prefer one information structure over another.
We focus on this weaker version and show that our original characterization still provides insights into it.
\begin{definition}\label{def: eventual socially valuable order}
    $\pi$ is {\it eventually more socially valuable} than $\pi'$, denoted $\pi\succsim_{ES} \pi'$, if there is one threshold $N\in\mathbb{N}$ such that for any decision problem $\mathcal{D}=(A,u)$, $V_{i}^{\mathcal{D}}(\pi,\bm{\sigma}^{*})\geq V_{i}^{\mathcal{D}}(\pi',\bm{\sigma}^{**})$ for all $i\geq N$, any equilibrium $\bm{\sigma}^{*}$ under $(\mathcal{D},\pi)$, and any equilibrium $\bm{\sigma}^{**}$ under $(\mathcal{D},\pi')$.
\end{definition}

By utilizing the proof of Theorem \ref{thm_characterization}, we have the following characterization:
\begin{proposition}\label{thm: asymptotic comparison}
    $\pi \succsim_{ES} \pi'$ if and only if there exists $N\in\mathbb{N}$ such that $V_{i}^{\mathcal{D}}(\pi,\bm{\sigma}^{*})\geq \overline{V}_{i}^{\mathcal{D}}(\pi')$ for any decision problem $\mathcal{D}$, all $i\geq N$, and any equilibrium $\bm{\sigma}^{*}$ under $(\mathcal{D},\pi)$.
\end{proposition}
The proof is almost the same as the one in Theorem \ref{thm_characterization}, and thus it is omitted.
It turns out that the necessary condition in Corollary \ref{observation_necessity_of_conclusive} is also a necessary condition in this setting. 
Thus, it is necessary that $\pi$ induces unbounded beliefs if $\pi \succsim_{ES}\pi'$ for some $\pi'$ except for the trivial case.
\begin{corollary}\label{coro_asymptotic}
    Suppose that $\pi'$ is not no information.
    If $\pi\succsim_{ES} \pi'$, then $\pi$ induces unbounded beliefs.
\end{corollary}
In Example \ref{example_1}, $V_{i}^{\mathcal{D}}(\pi',\bm{\sigma}^{**})$ is strictly higher than $V_{i}^{\mathcal{D}}(\pi,\bm{\sigma}^{*})$ for all $i\geq 2$, and thus $\pi$ is not eventually more socially valuable than $\pi'$. 
This means that inducing unbounded beliefs is not a sufficient condition even in this case.

By Proposition \ref{thm: asymptotic comparison}, we can highlight the difference between $\succsim_{ES}$ and the eventual Blackwell order/the large sample order in \citet{azrieli2014comment} and \citet{mu2021blackwell}.
\begin{definition}[\citet{mu2021blackwell}]
    $\pi$ is {\it eventually Blackwell more informative than} $\pi'$, denoted $\pi \succsim_{EB} \pi'$, if there exists $N\in\mathbb{N}$ such that $\pi^{\otimes i}\succsim_{B} \pi'^{\; \otimes i}$ for all $i\geq N$.
\end{definition}
By the same logic behind Example \ref{example_1}, the following Proposition \ref{prop: large sample order and long run comparison} shows that $\succsim_{ES}$ is strictly stronger than the eventual Blackwell order $\succsim_{EB}$, mirroring the relationship between $\succsim_{S}$ and the Blackwell order $\succsim_{B}$.
\begin{proposition}\label{prop: large sample order and long run comparison}
    $\succsim_{ES}$ is a strictly stronger binary relation than $\succsim_{EB}$.
\end{proposition}
As $\succsim_{S}$ is stronger than $\succsim_{ES}$, Theorem \ref{thm_sufficient_condition} still provides a sufficient condition for $\succsim_{ES}$.
However, we can provide a weaker sufficient condition for $\succsim_{ES}$ by leveraging the eventual Blackwell order $\succsim_{EB}$.
\begin{proposition}\label{thm: sufficient condition for long run comparison}
    If there exists $\pi''$ such that $\supp{\mu''}=\{0,\mu_{0},1\}$ and $\pi\succsim_{EB} \pi'' \succsim_{EB} \pi'$, then $\pi\succsim_{ES}\pi'$.
\end{proposition}
In other words, if there exists a mixture of full and no information between $\pi$ and $\pi'$ in the eventual Blackwell order, then $\pi$ is eventually more socially valuable than $\pi'$.

Note that, as shown in \citet{azrieli2014comment} and \citet{mu2021blackwell}, there exists a mixture of full and no information $\pi''$ and a symmetric binary information structure $\pi'$ such that $\pi'' \succsim_{EB} \pi'$ holds, while $\pi'' \succsim_{B} \pi'$ does not.
Our proof demonstrates that, under the condition $\supp{\mu''} = \{0, \mu_{0}, 1\}$, $\pi'' \succsim_{EB} \pi'$ implies $\pi'' \succsim_{ES} \pi'$.
Therefore, combining these observations shows that there exists an example in which $\pi'' \succsim_{ES} \pi'$ holds even though $\pi'' \succsim_{B} \pi'$ does not.\footnote{This implies that $\succsim_{S}$ is strictly stronger than $\succsim_{ES}$ as $\succsim_{S}$ is strictly stronger than $\succsim_{B}$ by Proposition \ref{prop: observation_strict_inclusion}.}
Furthermore, together with Example \ref{example_1}, this implies that $\succsim_{ES}$ is independent of the Blackwell order $\succsim_{B}$.

As a remark, one may also consider an intermediate notion, in the spirit of \citet{moscarini2002law} and \citet{frick2023learning}, that allows the threshold horizon to depend on the decision problem. This formulation is weaker than Definition \ref{def: eventual socially valuable order} and stronger than Definition \ref{def: limit order}. Under this variant, the analogues of Propositions \ref{thm: asymptotic comparison}, \ref{prop: large sample order and long run comparison}, and Corollary \ref{coro_asymptotic} follow after appropriately modifying the order of quantifiers. However, the analogue of Proposition \ref{thm: sufficient condition for long run comparison} does not follow directly, because the corresponding eventual comparison need not imply the Blackwell order even for mixtures of full and no information.

As noted above, $\succsim_{ES}$ compares two information structures based on the expected payoffs of agents who arrive sufficiently late.
To further weaken this requirement, we now introduce a coarser binary relation that focuses solely on the expected payoffs in the limit.
\begin{definition}\label{def: limit order}
    $\pi$ is {\it eventually more socially valuable in the limit} than $\pi'$, denoted $\pi\succsim_{LES} \pi'$, if $\lim_{i\to \infty}V_{i}^{\mathcal{D}}(\pi,\bm{\sigma}^{*})\geq \lim_{i\to \infty}V_{i}^{\mathcal{D}}(\pi',\bm{\sigma}^{**})$ for any decision problem $\mathcal{D}=(A,u)$, any equilibrium $\bm{\sigma}^{*}$ under $(\mathcal{D},\pi)$, and any equilibrium $\bm{\sigma}^{**}$ under $(\mathcal{D},\pi')$.
\end{definition}
Note that $\succsim_{LES}$ is equivalent to the comparison of the limit average discounted payoffs.
Formally, $\pi \succsim_{LES} \pi'$ if and only if
\[
\lim_{\delta\to 1} (1-\delta)\sum_{i=1}^{\infty}\delta^{i-1}V_i^{\mathcal{D}}(\pi,\bm{\sigma}^*)
\geq
\lim_{\delta\to 1} (1-\delta)\sum_{i=1}^{\infty}\delta^{i-1}V_i^{\mathcal{D}}(\pi',\bm{\sigma}^{**})
\]
for any decision problem $\mathcal{D}=(A,u)$, any equilibrium $\bm{\sigma}^*$ under $(\mathcal{D},\pi)$, and any equilibrium $\bm{\sigma}^{**}$ under $(\mathcal{D},\pi')$.

Under this weaker notion, we have a tight characterization.
\begin{theorem}\label{thm: limit comparison}
    Suppose that $\pi'$ is not no information.
    $\pi\succsim_{LES} \pi'$ if and only if $\pi$ induces unbounded beliefs. 
\end{theorem}
Thus, if we are concerned only with limit welfare, it is sufficient to check whether one information structure induces unbounded beliefs or discloses conclusive signals about each state, which is a necessary condition for $\succsim_{S}$ as in Corollary \ref{observation_necessity_of_conclusive}.
Moreover, Theorem \ref{thm: limit comparison}, combined with Example~\ref{example_1}, indicates that $\succsim_{ES}$ is strictly stronger than $\succsim_{LES}$.
The relationships among the binary relations introduced so far are summarized in Figure~\ref{fig: binary relations}.
\begin{figure}[ht]
    \centering
\begin{tikzpicture}[
  node distance= 1.0cm and 1.5cm
]

\node (R1) {$\succsim_{S}$};
\node (R2) [right=of R1] {$\succsim_{ES}$};
\node (R3) [below=of R1] {$\succsim_{B}$};
\node (R4) [below=of R2] {$\succsim_{EB}$};
\node (R5) [right=of R2] {$\succsim_{LES}$};

\draw[->, very thick] (R1) -- (R2);
\draw[->, very thick] (R1) -- (R3);
\draw[->, very thick] (R2) -- (R4);
\draw[->, very thick] (R3) -- (R4);
\draw[->, very thick] (R2) -- (R5);

\end{tikzpicture}
    \caption{Relationships among Binary Relations}
    \par
  \makeatletter\def\TPT@hsize{}\makeatletter
  \begin{tablenotes}
       \footnotesize
       \item[] \textit{Notes:}
       $\succsim_{B}$ (resp. $\succsim_{EB}$) represent the (resp. eventual) Blackwell order.
    $\succsim_{S}$ denotes the socially valuable order defined in Definition \ref{def: socially valuable order}.
    $\succsim_{ES}$ represents the eventual socially valuable order introduced in Definition \ref{def: eventual socially valuable order}.
    $\succsim_{LES}$ corresponds to the limit eventual socially valuable order defined in Definition \ref{def: limit order}.
    Arrows indicate the strength of binary relations:
    an arrow from relation P to relation Q means that P is (strictly) stronger than Q.
       \end{tablenotes}
    \label{fig: binary relations}
\end{figure}

\subsection{Equilibrium Multiplicity} \label{sec:weak_order}
Our definitions so far are too strong, particularly because they require the comparison to hold for all equilibria.
As a result, $\succsim_{S}$ is not a partial order.
\begin{proposition}\label{prop_partial_order}
    $\pi \succsim_{S} \pi$ if and only if $\supp\mu\subseteq\{0,\mu_{0},1\}$.
\end{proposition}

An alternative binary relation considers a weaker notion of comparison.
This definition relaxes the pessimistic and robust attitude toward equilibrium multiplicity in Definition \ref{def: socially valuable order}.
\begin{definition}
    $\pi$ is {\it weakly more socially valuable} than $\pi'$, denoted $\pi\succsim_{W} \pi'$, if for any decision problem $\mathcal{D}=(A,u)$ and any equilibrium $\bm{\sigma}^{**}$ under $(\mathcal{D},\pi')$, there exists an equilibrium $\bm{\sigma}^{*}$ under $(\mathcal{D},\pi)$ such that $V_{i}^{\mathcal{D}}(\pi,\bm{\sigma}^{*})\geq V_{i}^{\mathcal{D}}(\pi',\bm{\sigma}^{**})$ for any agent $i$.
\end{definition}
Under this definition, it is straightforward to see that $\pi \succsim_{W} \pi$ holds for any $\pi$.
Note that Example \ref{example_1} does not involve multiple equilibria under $\pi$, and therefore, the same result holds even if we consider this weak order. 
This means that $\succsim_{W}$ is also a strictly stronger order than $\succsim_{B}$.

In the following example, we highlight the difference between $\succsim_{S}$ and $\succsim_{W}$.
\begin{example}\label{example_weak}
    Suppose that $\mu_{0}=1/2$.
    Consider two information structures $\pi:\Omega\to\Delta(\{s_{0},s_{1},s_{2}\})$ and $\pi':\Omega\to\Delta(\{s'_{0},s'_{1},s'_{2}\})$ defined as follows:
    \begin{table}[H]
    \centering
    \begin{tabular}{c@{\hspace{1.5cm}}c}
    \begin{tabular}{l|ccc}
    \Xhline{1.2pt}
    $\pi$ & $s_{0}$ & $s_{1}$ & $s_{2}$ \\
    \hline
    $H$ & $0$ & $1-\varepsilon$ & $\varepsilon$ \\
    $L$ & $1-\delta$ & $0$ & $\delta$ \\
    \Xhline{1.2pt}
    \end{tabular}
    &
    \begin{tabular}{l|ccc}
    \Xhline{1.2pt}
    $\pi'$ & $s'_{0}$ & $s'_{1}$ & $s'_{2}$ \\
    \hline
    $H$ & $0$ & $1-\varepsilon'$ & $\varepsilon'$ \\
    $L$ & $1-\delta'$ & $0$ & $\delta'$ \\
    \Xhline{1.2pt}
    \end{tabular}
    \end{tabular}
    \end{table}
    \noindent
    Assume that $\delta < \delta' < \varepsilon < \varepsilon'$.\footnote{
    Note that this violates the sufficient condition of Theorem \ref{thm_sufficient_condition} as 
    \[
    1- \sum_{s\in\supp{\pi'}}\min\{\pi'(s|L),\pi'(s|H)\} = 1-\delta' \quad \text{and} \quad
    \min \{\pi(\mu=0|L),\pi(\mu=1|H)\}=1-\varepsilon.
    \]
    }
    Thus, $\pi \succsim_{B} \pi'$ holds, as this condition is equivalent to $\varepsilon \leq \varepsilon'$ and $\delta \leq \delta'$.
    Moreover, we assume that $\frac{\varepsilon'}{\varepsilon'+\delta'} < \frac{\varepsilon}{\varepsilon+\delta} < \frac{\varepsilon'^{2}}{\varepsilon'^{2}+\delta'^{2}}$.

    We now construct a decision problem in which the necessary condition of Theorem \ref{thm_characterization} is violated, implying that $\pi \succsim_{S} \pi'$ does not hold.
    Consider decision problem $\mathcal{D}$ defined as follows:
    Let $x=\frac{\varepsilon}{\varepsilon+\delta}$.
    The action set is given by $A=\{a_{0},a_{1},a_{2}\}$, with payoffs specified as follows: $u(a_{0},L)=u(a_{0},H)=u(a_{2},L)=u(a_{2},H)=0$ and $u(a_{1},H)=1-x$, $u(a_{1},L)=-x$.

    Now consider equilibrium strategy $\bm{\sigma}^{*}$ under $\pi$ such that agent $1$ chooses action $a_{0}$ if $s=s_{0}$ or $s_{2}$ and $a_{1}$ if $s=s_{1}$.
    Given this strategy, the posterior belief of agent $2$ when agent $1$'s action is $a_{0}$ and $s=s_{2}$ is $\frac{\varepsilon^{2}}{\varepsilon^{2}+\delta}$, which is lower than $x$.
    Given this, agent $2$ optimally chooses action $a_{1}$ if and only if (i) $s=s_{1}$ or (ii) $s=s_{2}$ and agent $1$ chooses action $a_{1}$.
    Thus, the expected payoff for agent $2$ under this equilibrium is given by $V_{2}^{\mathcal{D}}(\pi,\bm{\sigma}^{*})=(1-\varepsilon^{2})(1-x)/2$.
    
    However, under $\pi'$, when agent $1$ chooses action $a_{0}$ if $s'=s'_{0}$, $a_{2}$ if $s'=s'_{2}$, and $a_{1}$ if $s'=s'_{1}$, agent $2$ can perfectly infer agent $1$'s private signal.\footnote{Recall that $a_{2}$ always induces the same payoffs as $a_{0}$.
    Thus, this strategy is also optimal for agent $1$.}
    Given the assumption $x=\frac{\varepsilon}{\varepsilon+\delta} < \frac{\varepsilon'^{2}}{\varepsilon'^{2}+\delta'^{2}}$, when agent $2$ observes that agent $1$ chooses action $a_{2}$ and receives the private signal $s'=s'_{2}$, the optimal action is $a_{1}$.
    Thus, agent $2$ optimally chooses action $a_{1}$ if and only if either (i) $s'=s'_{1}$ or (ii) $s'=s'_{2}$ and agent $1$ chooses either action $a_{1}$ or $a_{2}$.
    Let $\bm{\sigma}^{**}$ denote the equilibrium strategy profile following this tie-breaking rule.
    Then, the expected payoff of agent $2$ is 
    \begin{align*}
        V_{2}^{\mathcal{D}}(\pi',\bm{\sigma}^{**})& =\overline{V}_{2}^{\mathcal{D}}(\pi') \\
        & = \frac{1}{2}(1-\varepsilon')(1-x) + \frac{1}{2}\varepsilon'(1-\varepsilon')(1-x)  \\
        & \quad \quad \quad \quad + \frac{1}{2}(\varepsilon'^{2}+\delta'^{2})\left(\frac{\varepsilon'^{2}}{\varepsilon'^{2}+\delta'^{2}}(1-x) +  \frac{\delta'^{2}}{\varepsilon'^{2}+\delta'^{2}}(-x) \right) \\
        &=\frac{1}{2}(1-x)-\frac{1}{2}\delta'^{2}x.
    \end{align*}
    Since $V_{2}^{\mathcal{D}}(\pi,\bm{\sigma}^{*})<V_{2}^{\mathcal{D}}(\pi',\bm{\sigma}^{**})$ is equivalent to $ \frac{\varepsilon}{\varepsilon+\delta} < \frac{\varepsilon'^{2}}{\varepsilon'^{2}+\delta'^{2}}$, it follows that $\pi \succsim_{S} \pi'$ does not hold. \qed

    Next, we establish that $\pi \succsim_{W} \pi'$.
    By directly constructing the equilibrium, we have a slightly more general observation:
    \begin{proposition}\label{observation_weak_order_3support}
    Suppose $\pi \succsim_{B} \pi'$ and $\supp \mu=\{0,x,1\}$ such that $|\mu_{0}-x|\geq |\mu_{0}-y|$ for all $y\in \supp {\mu'}\cap(0,1)$, where $x,y\in [0,1]$.\footnote{If $x=0$ or $x=1$, then $\supp \mu=\{0,1\}$.} Then, $\pi \succsim_{W} \pi'$.
    \end{proposition}
    By applying Proposition \ref{observation_weak_order_3support}, we confirm that in this example, $\pi\succsim_{W} \pi'$ holds. 
    The key feature is that under $\pi$, if agent $1$ chooses action $a_{1}$ or $a_{2}$ when $s=s_{2}$, then agent $2$ can obtain the expected payoff as if she were able to observe the past signal realization.
    \qed 
\end{example}
    Beyond Example \ref{example_weak}, we cannot obtain a general characterization or a simple sufficient condition for the weaker order $\succsim_{W}$.  

\subsection{Canonical Binary Environments} \label{sec:restricted decision problem}
Up to this point, we have assumed that the domain of decision problems is unusually rich, following the original work of \citet{blackwell1953equivalent}.
This strong assumption simplifies our analysis in Theorem \ref{thm_characterization}, but it also rules out any comparison between binary information structures by Corollary \ref{observation_necessity_of_conclusive}.
We now show that this impossibility is sharply reversed once we restrict attention to the canonical binary environments studied in the social learning literature.
In such environments, binary signals have a special structure: along the equilibrium path, an action either reveals the agent's binary signal or becomes independent of it.
This structure prevents the kind of adverse informational externality that can overturn the Blackwell order in richer environments.

Specifically, we focus on the binary information structures under the "canonical" binary decision problems without tie issues.
As in the classical works of \citet{banerjee1992simple} and \citet{bikhchandani1992theory}, we consider the signal space $S=\{s_{l},s_{h}\}$ and assume without loss of generality that $\pi(s_{h}|H)\geq \pi(s_{h}|L)$. 
Let $\Pi^{B}$ denote the set of binary information structures satisfying $\pi(s_{h}|H)\geq \pi(s_{h}|L)$.

Here, we focus on all binary decision problems while excluding those that are nonessential. In particular, we exclude binary decision problems in which one action is always optimal regardless of the decision maker's belief. Moreover, for the action set $A=\{a_0,a_1\}$, we may assume without loss of generality that $a_0$ (resp. $a_{1}$) is the unique optimal action on state $L$ (resp. $H$). 
Let $\mathscr{D}^{B}$ denote the collection of all such binary decision problems. 
In other words, $\mathscr{D}^{B}=\{\mathcal{D}=(A,u)\mid A=\{a_0,a_1\},u(a_1,H)>u(a_0,H),u(a_0,L)>u(a_1,L)\}.$
Take any $\mathcal{D}\in \mathscr{D}^{B}$. Then, the optimal action is $a_0$ (resp. $a_1$) if the belief in state $H$ is less (resp. greater) than or equal to $\frac{u(a_0,L)-u(a_1,L)}{u(a_0,L)-u(a_1,L)+u(a_1,H)-u(a_0,H)}$

Fix any $N\in\mathbb{N}$.
When an agent is indifferent between the two actions, the way in which the tie is broken affects the beliefs of subsequent agents. Since our objective is to analyze situations in which no ties arise, we define the notion of \emph{no tie-break} through a sufficient condition that guarantees the absence of ties.
When $\mathcal{D}$ induces no tie-break under $\pi$, the equilibrium is essentially unique.\footnote{Here, by saying that the equilibrium is essentially unique, we mean that equilibrium strategies may differ only in their off-path prescriptions. Since such differences have no effect on agents' payoffs, we identify all such strategy profiles and denote the resulting equivalence class by $\bm{\sigma}^{*}$.}
\begin{definition}
 For $\pi \in \Pi^{B}$ and $\mathcal{D}\in \mathscr{D}^{B}$, we say $\mathcal{D}$ induces {\it no tie-break} under $\pi$ if $\frac{u(a_0,L)-u(a_1,L)}{u(a_0,L)-u(a_1,L)+u(a_1,H)-u(a_0,H)}\neq \frac{\mu_0[\pi(s_{h}|H)]^n[\pi(s_{l}|H)]^m}{\mu_0[\pi(s_{h}|H)]^n[\pi(s_{l}|H)]^m+(1-\mu_0)[\pi(s_{h}|L)]^n[\pi(s_{l}|L)]^m}$  for all $(n,m)\in \{(n,m)\in \mathbb{Z}_+^2\mid 1\leq n+m\leq N\}$.
\end{definition}

The main result is that, in the binary social learning model, an increase in the Blackwell informativeness of each agent's private information weakly increases every agent's payoff whenever the no-tie condition is satisfied. More precisely, we establish the following result.
\begin{theorem}\label{prop: binary binary}
    Suppose that $\pi,\pi'\in \Pi^{B}$ satisfy $\pi \succsim_B \pi'$. Take any $\mathcal{D}\in \mathscr{D}^{B}$ which induces no tie-break under $\pi$ and $\pi'$. Then, $V_{i}^{\mathcal{D}}(\pi,\bm{\sigma}^{*})\geq V_{i}^{\mathcal{D}}(\pi',\bm{\sigma}^{**})$ for all $i=1,2,\dots,N$, where $\bm{\sigma}^{*}$ and $\bm{\sigma}^{**}$ denote the unique equilibria under $\pi$ and $\pi'$, respectively.
\end{theorem}
The proof is relatively long and is therefore provided in the Online Appendix \ref{sec: proof binary}.
The key idea behind the proof of Theorem \ref{prop: binary binary} is to decompose the comparison into local and global arguments. We first show that if two binary information structures $\pi$ and $\pi'$ are sufficiently "close", then they induce the same equilibrium behavior on the equilibrium path. Within such a region, a Blackwell improvement weakly increases every agent's payoff. We then consider an arbitrary pair of information structures satisfying $\pi \succsim_B \pi'$ and connect $\pi$ and $\pi'$ by a path along which information gradually becomes less informative in the Blackwell sense. Payoffs vary monotonically as long as equilibrium behavior remains unchanged and, under binary information structures, remain continuous when equilibrium behavior changes or ties arise. This allows us to extend the local monotonicity argument to arbitrary pairs of information structures and thereby conclude that agents' payoffs respect the Blackwell order.\footnote{Note that, even with binary signals, these properties may fail when tie-breaking occurs, depending on the equilibrium selection. 
We illustrate this point in Online Appendix \ref{sec: proof binary}.}

This result contrasts sharply with Example \ref{example_1}, in which a Blackwell more informative signal can lower some agents' payoffs because observed actions garble private signals in a way that worsens the informational content of histories. 
With binary signals and binary actions, however, each on-path action either perfectly reveals the realized signal or conveys no additional information. 
Consequently, a Blackwell improvement in private signals cannot induce the kind of discontinuous deterioration in public beliefs that drives Example \ref{example_1}. 
In this sense, our result uncovers a knife-edge property of the binary environments most commonly analyzed in the literature.

At the same time, this observation clarifies the strength of the requirements imposed by $\succsim_{S}$. On the universal domain of decision problems, Corollary \ref{observation_necessity_of_conclusive} implies that no binary information structure can be ranked by $\succsim_{S}$. 
By contrast, binary environments restore the usual welfare intuition behind the Blackwell order as in Theorem \ref{prop: binary binary}. 
This conclusion, however, relies critically on the restriction to binary information structures: as illustrated by Example \ref{example_1}, once more general information structures are admitted, a more informative signal $\pi$ may still yield lower payoffs for some agents than a less informative signal $\pi'$, even under canonical and generic binary decision problems.
Thus, restricting attention to binary environments substantially weakens our requirements, but the resulting relation remains strictly stronger than the Blackwell order. More generally, it remains unclear how our results extend to other restricted classes of simple decision problems, owing to the complexity of the dynamics of public beliefs and agents' optimal strategies even in binary environments.

\section{Conclusion} \label{sec:conclusion}
In this study, we analyzed the comparison of information structures in the classical social learning model.
Our main binary relation, $\succsim_{S}$, defines one information structure as more socially valuable than another if it yields weakly higher expected payoffs for all agents, across all decision problems and equilibrium realizations.
Proposition \ref{prop: observation_strict_inclusion} shows that a Blackwell more informative information structure may be worse for some agents.
Theorem \ref{thm_characterization} provides a characterization and establishes a necessary condition, demonstrating that inducing unbounded beliefs is required for all agents to be better off.
Building on this, Theorem \ref{thm_sufficient_condition} derives a simple sufficient condition by exploiting the properties of mixtures of full and no information.
Furthermore, our framework naturally extends to long-run and limit comparisons, and we discuss other extensions in Online Appendix \ref{sec:extension}.

We leave open a number of questions.
Our limitations stem primarily from two main sources of analytical difficulty: information externalities and learning speeds.
First, public beliefs depend heavily on the underlying decision problem, and there is generally no Blackwell order among the induced public beliefs, even when one information structure yields higher expected payoffs than another for a particular decision problem.
As a result, explicitly characterizing or comparing the expected payoffs or posterior beliefs becomes highly complex.
This complexity is the technical reason we rely on a strong robustness requirement with respect to equilibrium multiplicity in Theorem \ref{thm_characterization}, and why we cannot provide a more general characterization of $\succsim_{W}$ and $\succsim_{S}$ in terms of the underlying information structures (even under the restricted domain of decision problems).

Second, even when we employ Theorem \ref{thm_characterization}, a general analysis of learning speeds over all periods is still required, as this is closely related to the distribution induced by repeated sampling from the information structure.
This necessity prevents us from establishing or negating the converse of Theorem \ref{thm_sufficient_condition}, and more broadly, from determining whether the expected payoffs under one information structure are always higher than those under the observable signal setting of another in cases in which the payoffs and beliefs can be explicitly derived.

Therefore, deriving the payoffs or beliefs is difficult, and moreover, making comparisons remains challenging even when these can be computed exactly.
A natural direction for future research is to pursue a more complete characterization of the social value of information by deepening its connection with and further developing the analysis of repeated Blackwell comparisons.
\titleformat{\section}
		{\Large\bfseries}     
         {Appendix \thesection:}
        {0.5em}
        {}
        []

\renewcommand{\thetheorem}{A.\arabic{theorem}}
\setcounter{theorem}{0}

 \appendix 

\section{Omitted Proofs}
\subsection{Preliminaries}
In this subsection, we present some preliminary results that will be used in subsequent proofs.

For each $a^{*}\in A$ and $z\in[0,1]$, define
\[
B(a^{*})=\left\{z\in[0,1] \mid a^{*}\in \argmax_{a\in A} \ [zu(a,H)+(1-z)u(a,L)]\right\},
\]
and 
\[
B^{-1}(z)=\left\{a\in A \mid z\in B(a)\right\}=\argmax_{a\in A} \ [zu(a,H)+(1-z)u(a,L)].
\]

\begin{lemma}\label{lemma_closed_interval}
   Fix any $\mathcal{D}=(A,u)$.     
   For each $a^{*}\in A$, $B(a^{*})$ is a closed interval.
\end{lemma}
\begin{proof}[Proof of Lemma \ref{lemma_closed_interval}]
    Since $zu(a,H)+(1-z)u(a,L)$ is continuous with respect to $z$, $B(a^{*})$ is a closed set. 
    Suppose $z_1\in B(a^{*})$ and $z_2\in B(a^{*})$. 
    It follows that $z_1u(a^{*},H)+(1-z_1)u(a^{*},L)\geq z_1u(a,H)+(1-z_1)u(a,L)$ and  $z_2u(a^{*},H)+(1-z_2)u(a^{*},L)\geq z_2u(a,H)+(1-z_2)u(a,L)$ for any $a\in A$. 
    Take any $t\in [0,1]$, then we have
    \begin{align*}
        &[tz_1+(1-t)z_2]u(a^{*},H)+[1-tz_1-(1-t)z_2]u(a^{*},L)\\
        &= t[z_1u(a^{*},H)+(1-z_1)u(a^{*},L)]+(1-t)[z_2u(a^{*},H)+(1-z_2)u(a^{*},L)]\\
        &\geq t[z_1u(a,H)+(1-z_1)u(a,L)]+(1-t)[z_2u(a,H)+(1-z_2)u(a,L)]\\
        &=[tz_1+(1-t)z_2]u(a,H)+[1-tz_1-(1-t)z_2]u(a,L).
    \end{align*}
    Hence, $tz_1+(1-t)z_2\in B(a^{*})$.
\end{proof}

\begin{lemma}\label{lemma_best_action}
 Fix any $\mathcal{D}=(A,u)$.    
    Suppose  $B^{-1}(z_1)\cap B^{-1}(z_2)\neq \emptyset$ for some $0\leq z_1< z_2\leq 1$. 
    Then, $B^{-1}(w)=B^{-1}(z_1)\cap B^{-1}(z_2)$ for all $w\in (z_1,z_2)$.
      \end{lemma}
      \begin{proof}[Proof of Lemma \ref{lemma_best_action}]
      Take any $a_{0}\in B^{-1}(z_1)\cap B^{-1}(z_2)$.
          Then, $z_1u(a_{0},H)+(1-z_1)u(a_{0},L)\geq z_1u(a,H)+(1-z_1)u(a,L)$ and $ z_2u(a_{0},H)+(1-z_2)u(a_{0},L)\geq z_2u(a,H)+(1-z_2)u(a,L)$  for all $a\in A$. Note that at least one inequality holds strictly if $a\notin B^{-1}(z_1)\cap B^{-1}(z_2)$.
          Hence, for any  $w\in (z_1,z_2)$, 
 \begin{align*}
 &wu(a_{0},H)+(1-w)u(a_{0},L)\\
 & = \frac{w-z_2}{z_1-z_2}[ z_1u(a_{0},H)+(1-z_1)u(a_{0},L)] \\
 & \quad \quad \quad \quad \quad \quad \quad \quad+\left(1-\frac{w-z_2}{z_1-z_2}\right) [ z_2u(a_{0},H)+(1-z_2)u(a_{0},L)]\\
 &\geq \frac{w-z_2}{z_1-z_2}[z_1u(a,H)+(1-z_1)u(a,L)] \\
 & \quad \quad \quad \quad \quad \quad \quad \quad+\left(1-\frac{w-z_2}{z_1-z_2}\right)[ z_2u(a,H)+(1-z_2)u(a,L)]\\
 &= wu(a,H)+(1-w)u(a,L)
 \end{align*}
  for all $a\in A$ and strict inequality holds for all $a\notin  B^{-1}(z_1)\cap B^{-1}(z_2)$. 
  Thus, $B^{-1}(w)=B^{-1}(z_{1})\cap B^{-1}(z_{2})$. 
      \end{proof}
\begin{lemma}\label{lemma_expected_payoff_special_case}
        Suppose $(\mathcal{D},\pi)$ satisfies $\supp \mu\cap(x,1)=\emptyset$ and  $B^{-1}(0)= B^{-1}(x)=\{a_{0}\}$ for some $x\geq \mu_{0}$ and $a_{0}\in A$. 
        Take arbitrary equilibrium $\bm{\sigma}^{*}$ under $(\mathcal{D},\pi)$. 
        Then,
        \[
            V_{i}^{\mathcal{D}}(\pi,\bm{\sigma}^{*})=\mu_{0}[(1-p^{i})u(a_{1},H)+p^{i}u(a_{0},H)] +(1-\mu_{0})u(a_{0},L),
        \]
        where $p=1-\pi(\mu=1|H)$ and $a_{1}\in B^{-1}(1)$.
      \end{lemma}
      \begin{proof}[Proof of Lemma \ref{lemma_expected_payoff_special_case}]
          If $a_{0}\in B^{-1}(1)$, the statement holds because
          \begin{align*}
             V_{i}^{\mathcal{D}}(\pi,\bm{\sigma}^{*})
             &=\mu_{0}u(a_{0},H)+(1-\mu_{0})u(a_{0},L)\\
             &=\mu_{0}[(1-p^{i})u(a_{0},H)+p^{i}u(a_{0},H)]+(1-\mu_{0})u(a_{0},L)\\
             &=\mu_{0}[(1-p^{i})u(a_{1},H)+p^{i}u(a_{0},H)]+(1-\mu_{0})u(a_{0},L).
        \end{align*}
          Suppose $a_{0}\notin B^{-1}(1)$. 
          Take any equilibrium under $\pi$.  
          Then, by Lemma \ref{lemma_best_action}, agent 1 chooses $a_{0}$ if and only if he receives $s\notin S(1)$. 
          Agent 2 chooses an action from $B^{-1}(1) $ if she receives $s\in S(1)$ or agent 1 takes an action other than $a_{0}$ because she knows that the state is $H$. 
          Notably, the public belief after observing $a_{0}$ is less than $\mu_{0}$ and Lemma~\ref{lemma_best_action} implies that $B^{-1}(z)=\{a_{0}\}$ for all $z\in[0,x]$. 
          Hence, agent 2 must choose $a_{0}$ if she receives $s\notin S(1)$ and agent 1 chooses $a_{0}$. 
          Analogously, agent $i$ takes action from $B^{-1}(1)$ if and only if he receives $s\in S(1)$ or at least one previous agent chooses an action other than $a_{0}$. 
          Otherwise, agent $i$ takes $a_{0}$. 
          Therefore, we have
                    \begin{align*}
             V_{i}^{\mathcal{D}}(\pi,\bm{\sigma}^{*})=\mu_{0}[(1-p^{i})u(a_{1},H)+p^{i}u(a_{0},H)]+(1-\mu_{0})u(a_{0},L).
        \end{align*}
      \end{proof}

\subsection{Proofs of Theorem \ref{thm_characterization} and Corollary \ref{observation_necessity_of_conclusive}}
We first provide a self-contained proof of the following lemma.
\begin{lemma}\label{observation_blackwell}
    If $\pi\succsim_{B} \pi'$ and $\rho \succsim_{B} \rho'$, then $\pi \otimes\rho \succsim_{B} \pi'\otimes \rho'$.
\end{lemma}
\begin{proof}[Proof of Lemma \ref{observation_blackwell}]
    Suppose $\pi\succsim_{B} \pi'$ and $\rho \succsim_{B} \rho'$. 
    Then, there exist Markov kernels $\gamma_{1}$ and $\gamma_{2}$ such that
    \[
    \pi'(s'|\omega)=\sum_{s\in \supp\pi}\gamma_{1}(s'|s)\pi(s|\omega) 
    \quad \text{and} \quad \rho'(t'|\omega)=\sum_{t\in \supp\rho}\gamma_{2}(t'|t)\rho(t|\omega)
    \]
    for all $s'\in \supp{\pi'} $ and $t'\in \supp{\rho'}$.
    Then, we have
    \begin{align*}
        (\pi'\otimes \rho')((s',t')|\omega)&=\pi'(s'|\omega)\rho'(t'|\omega)\\
        &=\sum_{s\in \supp\pi}\gamma_{1}(s'|s)\pi(s|\omega)\sum_{t\in \supp\rho}\gamma_{2}(t'|t)\rho(t|\omega)\\
        &=\sum_{(s,t)\in \supp{\pi\otimes \rho}}\gamma_{1}(s'|s)\gamma_{2}(t'|t)\pi(s|\omega)\rho(t|\omega)\\
        &=\sum_{(s,t)\in \supp{\pi\otimes \rho}}\gamma((s',t')|(s,t))(\pi\otimes\rho)((s,t)|\omega),
    \end{align*}
    where $\gamma((s',t')|(s,t))=\gamma_{1}(s'|s)\gamma_{2}(t'|t)$. Since $\gamma$ is a Markov kernel, $\pi'\otimes \rho'$ is a garbling of $\pi \otimes \rho$.
\end{proof}

Then, the next lemma establishes that the expected payoff under the observable signal setting provides an upper bound for each agent, any decision problem, and any equilibrium.
\begin{lemma}\label{lemma_signal_is_more_informative}
   $\overline{V}_{i}^{\mathcal{D}}(\pi)\geq V_{i}^{\mathcal{D}}(\pi,\bm{\sigma})$ for all $i, \mathcal{D},\pi$, and $\bm{\sigma}$.
\end{lemma}
\begin{proof}[Proof of Lemma~\ref{lemma_signal_is_more_informative}]
Take any $\mathcal{D},\pi$, and $\bm{\sigma}$. 
Note that $ \overline{V}_{1}^{\mathcal{D}}(\pi)= V_{1}^{\mathcal{D}}(\pi,\bm{\sigma})$.
Fix $i\geq 2$. 
For each $\bm{s}\in S^{i-1}$, define $f_{i-1}(\cdot|\bm{s})\in\Delta(A^{i-1})$ as 
\begin{align*}
    f_{i-1}(\bm{a}|\bm{s})=\prod_{k=1}^{i-1}\sigma_{k}(a_{k}|a_{1},\dots,a_{k-1},s_{k}).
\end{align*}
Hence, $f_{i-1}(\bm{a}|\bm{s})$ is the probability that agent $1$ to agent $i-1$ takes actions $\bm{a}=(a_{1},\dots,a_{i-1})$ when agent $1$ to agent $i-1$ receives private signal $\bm{s}=(s_{1},\dots,s_{i-1})$. Then,
\begin{align*}
    \alpha_{\leq i-1}^{\omega}(\bm{a}|\pi,\bm{\sigma})=&\sum_{\bm{s}\in S^{i-1}}\prod_{k=1}^{i-1}\sigma_{k}(a_{k}|a_{1},\dots,a_{k-1},s_{k})\pi(s_{k}|\omega)\\=&\sum_{\bm{s}\in S^{i-1}}f_{i-1}(\bm{a}|\bm{s})\pi^{\otimes i-1}(\bm{s}|\omega).
\end{align*}
Thus, $\alpha_{\leq i-1}(\cdot|\pi,\bm{\sigma})$ is a garbling of $\pi^{\otimes i-1}$, where $\alpha_{\leq i-1}$ is an ex-ante action distribution induced by $\alpha_{\leq i-1}^{\omega}$.\footnote{Precisely, $\alpha_{\leq i-1}(\bm{a}|\pi,\bm{\sigma})=(1-\mu_{0})\alpha_{\leq i-1}^{L}(\bm{a}|\pi,\bm{\sigma})+\mu_{0}\alpha_{\leq i-1}^{H}(\bm{a}|\pi,\bm{\sigma})$.} 
By Lemma~\ref{observation_blackwell}, we have $\pi^{\otimes i-1}\otimes \pi \succsim_{B} \alpha_{\leq i-1}(\cdot|\pi,\bm{\sigma})\otimes\pi$.
Hence, $\overline{V}_{i}^{\mathcal{D}}(\pi)\geq V_{i}^{\mathcal{D}}(\pi,\bm{\sigma})$ holds for all $i$, $\mathcal{D}$, $\pi$, and $\bm{\sigma}$.  
\end{proof}
\begin{proof}[Proof of Theorem \ref{thm_characterization}]
    Since $\overline{V}_{i}^{\mathcal{D}}(\pi')\geq V_{i}^{\mathcal{D}}(\pi',\bm{\sigma}')$ holds for every strategy profile $\bm{\sigma}'$ by Lemma \ref{lemma_signal_is_more_informative}, $\pi\succsim_{S} \pi'$ holds if $V_{i}^{\mathcal{D}}(\pi,\bm{\sigma}^{*})\geq \overline{V}_{i}^{\mathcal{D}}(\pi')$.
    
    Conversely, suppose $\pi\succsim_{S}\pi'$.
    Take any $\mathcal{D}=(A,u)$, equilibrium $\bm{\sigma}^{*}$ under $\pi:\Omega\to \Delta(S)$, and equilibrium $\bm{\sigma}^{**}$ under $\pi':\Omega\to \Delta(S')$. 
    Then, $V_{i}^{\mathcal{D}}(\pi,\bm{\sigma}^{*})\geq V_{i}^{\mathcal{D}}(\pi',\bm{\sigma}^{**})$ by  $\pi\succsim_{S}\pi'$. 
    Consider the decision problem $\overline{\mathcal{D}}=(\overline{A},\overline{u})$, where $\overline{A}=\{(a,k)\mid a\in A,k\in S'\}$ and $\overline{u}((a,k),\omega)=u(a,\omega)$ for all $a\in A,\omega\in \Omega$. 
    Fix $s_{1}\in S'$ and define strategy profile $\bm{\sigma}=(\sigma_{i})_{i\in \mathbb{N}}$ under $(\overline{\mathcal{D}},\pi)$ as follows: 
    \begin{align*}
    \begin{cases}
        \sigma_{i}((a,s_{1})|(a_{1},k_{1}),(a_{2},k_{2}),\dots,(a_{i-1},k_{i-1}),s)=\sigma^{*}_{i}(a| a_{1},a_{2},\dots,a_{i-1},s)\\
        \sigma_{i}((a,k)|(a_{1},k_{1}),(a_{2},k_{2}),\dots,(a_{i-1},k_{i-1}),s)=0,
    \end{cases}
    \end{align*}
    for all $a\in A$, $s\in S$, $(a_{1},\dots,a_{i-1})\in A^{i-1}$, $k_{1},k_{2},\dots,k_{i-1}\in S'$, and $k\in S'\setminus \{s_{1}\}$. 
    Note that $\bm{\sigma} $ is an equilibrium under $(\overline{\mathcal{D}},\pi)$. 
    Moreover, it follows that $V_{i}^{\overline{\mathcal{D}}}(\pi,\bm{\sigma})=V_{i}^{\mathcal{D}}(\pi,\bm{\sigma}^{*})$.
    Under $(\overline{\mathcal{D}},\pi')$, if we consider the following equilibrium $\bm{\sigma}'$, the expected payoff of agent $i$ at equilibrium ($V_{i}^{\overline{\mathcal{D}}}(\pi',\bm{\sigma}')$) coincides with $\overline{V}_{i}^{\mathcal{D}}(\pi')$.
    Specifically, each agent $i$ chooses an action that maximizes his expected payoff on the equilibrium path, but always chooses an action of the form $(a,k)$ ($a\in A$) when the received signal is $k\in S'$.
    Since each agent can observe signals received by their predecessor on the equilibrium path, it follows that $V_{i}^{\overline{\mathcal{D}}}(\pi',\bm{\sigma}')=\overline{V}_{i}^{\overline{\mathcal{D}}}(\pi')=\overline{V}_{i}^{\mathcal{D}}(\pi')$. 
    Therefore,
    \[
      V_{i}^{\mathcal{D}}(\pi,\bm{\sigma}^{*})=  V_{i}^{\overline{\mathcal{D}}}(\pi,\bm{\sigma})\geq V_{i}^{\overline{\mathcal{D}}}(\pi',\bm{\sigma}')=\overline{V}_{i}^{\mathcal{D}}(\pi').
    \]
\end{proof}

\begin{proof}[Proof of Corollary \ref{observation_necessity_of_conclusive}]
Prove by contradiction. 
Suppose co$(\supp{\mu})\neq [0,1]$. 
Then, either $1\notin \supp{\mu}$ or $0 \notin \supp{\mu}$. 
By symmetry, it suffices to consider the case where $1\notin \supp{\mu}$.
Since $\supp{\mu}$ is a closed set, 
there exists $r\in[\mu_{0},1)$ such that $\supp{\mu}\subseteq [0,r]$.
Consider the following decision problem $\mathcal{D}=(A,u)$: $A=\{a_{0},a_{1}\}$, $u(a_{0},L)=u(a_{0},H)=0$, $u(a_{1},H)=1-r$ and $u(a_{1},L)=-r$. 
Then, the strategy profile $\bm{\sigma}^{*}$ that all agents always choose $a_{0}$ is an equilibrium under $(\mathcal{D},\pi)$. 
It follows that $ V_{i}^{\mathcal{D}}(\pi,\bm{\sigma}^{*})=0$.
Since $\pi'$ is not no information, repeated observations of $\pi'$ allow agents to learn the state in the limit. 
Hence, $\overline{V}_{i}^{\mathcal{D}}(\pi')>0$ holds for sufficiently large $i$. 
By Theorem \ref{thm_characterization}, $\pi$ is not more socially valuable than $\pi'$. 
\end{proof}

\subsection{Proof of Theorem \ref{thm_sufficient_condition}}
The following lemma shows that the expected payoff under the mixture of full and no information is the same as that under observable signal settings for any decision problem and equilibrium.
\begin{lemma}\label{lemma_expected_payoff_3support}
    Suppose $\supp \mu=\{0,\mu_{0},1\}$. 
    Fix the decision problem $\mathcal{D}=(A,u)$. 
    Take arbitrary equilibrium $\bm{\sigma}^{*}$ under $(\mathcal{D},\pi)$. 
    Then,
    \begin{align*}
         V_{i}^{\mathcal{D}}(\pi,\bm{\sigma}^{*})
         &=\overline{V}_{i}^{\mathcal{D}}(\pi)\\
         &=\mu_{0}[(1-p^{i})U_{1}]+(1-\mu_{0})[(1-p^{i})U_{0}]+p^{i}U_{\mu_{0}},
    \end{align*}
    where $U_{1}=\max_a u(a,H)$, $U_{0}=\max_a u(a,L)$, $U_{\mu_{0}}=\max_a [\mu_{0} u(a,H)+(1-\mu_{0})u(a,L)]$, and $p=\pi(\mu=\mu_{0}|H)=\pi(\mu=\mu_{0}|L)$. 
\end{lemma}
\begin{proof}[Proof of Lemma \ref{lemma_expected_payoff_3support}]
First, it is easily calculated that
 \begin{align*}
     \overline{V}_{i}^{\mathcal{D}}(\pi)=\mu_{0}[(1-p^{i})U_{1}]+(1-\mu_{0})[(1-p^{i})U_{0}]+p^{i}U_{\mu_{0}}.
\end{align*}
We now show that $V_{i}^{\mathcal{D}}(\pi,\bm{\sigma}^{*})=\overline{V}_{i}^{\mathcal{D}}(\pi)$. 
First, this obviously holds for agent $1$:
\begin{align*}
      V_{1}^{\mathcal{D}}(\pi,\bm{\sigma}^{*})= \overline{V}_{1}^{\mathcal{D}}(\pi) = \mu_{0} (1-p)U_{1}+(1-\mu_{0})(1-p)U_{0}+pU_{\mu_{0}}.
\end{align*}
Then, for each $i\geq 2$, we consider a strategy in which agent $i$ chooses the optimal actions upon receiving conclusive signals about each state and mimics agent $i-1$'s action otherwise.
By the optimality of the equilibrium strategy, for each $i\geq 2$, we have
\begin{align*}
    V_{i}^{\mathcal{D}}(\pi,\bm{\sigma}^{*})&\geq \mu_{0} (1-p)U_{1}+(1-\mu_{0})(1-p)U_{0} \\
     & \quad +p\left[\mu_{0}\sum_{a}\alpha_{i-1}^{H}(a|\pi,\bm{\sigma}^{*})u(a,H) +(1-\mu_{0})\sum_{a}\alpha_{i-1}^{L}(a|\pi,\bm{\sigma}^{*})u(a,L)\right]\\
     &= \mu_{0} (1-p)U_{1}+(1-\mu_{0})(1-p)U_{0}+pV_{i-1}^{\mathcal{D}}(\pi,\bm{\sigma}^{*}).
\end{align*}
Conversely, from Lemma~\ref{lemma_signal_is_more_informative}, for each $i$
\begin{align*}
      V_{i}^{\mathcal{D}}(\pi,\bm{\sigma}^{*})
      \leq \overline{V}_{i}^{\mathcal{D}}(\pi) = \mu_{0} (1-p^{i})U_{1}+(1-\mu_{0})(1-p^{i})U_{0}+p^{i}U_{\mu_{0}}.
\end{align*}
Fix some $j\geq 2$ and suppose $V_{j-1}^{\mathcal{D}}(\pi,\bm{\sigma}^{*})=\overline{V}_{j-1}^{\mathcal{D}}(\pi)$. 
Then,
\begin{align*}
    V_{j}^{\mathcal{D}}(\pi,\bm{\sigma}^{*})&\geq  \mu_{0} (1-p)U_{1}+(1-\mu_{0})(1-p)U_{0}+pV_{j-1}^{\mathcal{D}}(\pi,\bm{\sigma}^{*})\\
    &=\mu_{0} (1-p)U_{1}+(1-\mu_{0})(1-p)U_{0} \\
    & \quad \quad \quad \quad +p\left[\mu_{0} (1-p^{j-1})U_{1}+(1-\mu_{0})(1-p^{j-1})U_{0}+p^{j-1}U_{\mu_{0}}\right]\\
    &=\mu_{0} (1-p^j)U_{1}+(1-\mu_{0})(1-p^j)U_{0}+p^jU_{\mu_{0}}\\
    &=\overline{V}_j^{\mathcal{D}}(\pi).
\end{align*}
Hence, we have $V_{j}^{\mathcal{D}}(\pi,\bm{\sigma}^{*})=\overline{V}_j^{\mathcal{D}}(\pi)$. 
By mathematical induction, it follows that $V_{i}^{\mathcal{D}}(\pi,\bm{\sigma}^{*})=\overline{V}_{i}^{\mathcal{D}}(\pi)$ for all $i$.
\end{proof}
Utilizing Lemma \ref{lemma_expected_payoff_3support} and Blackwell's theorem (Lemma \ref{observation_blackwell}), we can show that the expected payoff under $\pi$ is weakly higher than the upper bound under $\pi'$ for any decision problem if $\pi$ is a mixture of full and no information.
\begin{lemma}\label{observation_strict_order_3support}
    Suppose $\pi \succsim_{B} \pi'$ and $\supp \mu=\{0,\mu_{0},1\}$. 
    Then, $\pi \succsim_{S} \pi'$.
\end{lemma}
\begin{proof}[Proof of Lemma \ref{observation_strict_order_3support}]
Take any $\mathcal{D}=(A,u)$. Take arbitrary equilibrium $\bm{\sigma}^{*}$ under $(\mathcal{D},\pi)$. 
From Lemma~\ref{lemma_expected_payoff_3support}, we have $V_{i}^{\mathcal{D}}(\pi,\bm{\sigma}^{*})=\overline{V}_{i}^{\mathcal{D}}(\pi)$.
Hence, the expected payoff of agent $i$ in any equilibrium is the same as the expected payoff of agent $i$ when agent $i$ can observe the signal realizations of past agents rather than the actions taken by them.
 
Next, take any equilibrium $\bm{\sigma}^{**}$ under $(\mathcal{D},\pi')$. 
Note that $\overline{V}_{i}^{\mathcal{D}}(\pi')\geq V_{i}^{\mathcal{D}}(\pi',\bm{\sigma}^{**})$ holds by Lemma~\ref{lemma_signal_is_more_informative}.
Since $\pi^{\otimes i} \succsim_{B} \pi'^{\ \otimes i}$ by Lemma~\ref{observation_blackwell}, we have $\overline{V}_{i}^{\mathcal{D}}(\pi)\geq \overline{V}_{i}^{\mathcal{D}}(\pi')$.
Hence, it follows that
\begin{align*}
    V_{i}^{\mathcal{D}}(\pi,\bm{\sigma}^{*})= \overline{V}_{i}^{\mathcal{D}}(\pi)\geq  \overline{V}_{i}^{\mathcal{D}}(\pi')\geq V_{i}^{\mathcal{D}}(\pi',\bm{\sigma}^{**}).
\end{align*}
Therefore, we obtain $\pi \succsim_{S} \pi'$.
\end{proof}
We then construct a strategy profile under $\pi$ that achieves the same equilibrium expected payoff under $\pi'$ when $\pi'$ is a mixture of full and no information.
Additionally, we show that this strategy profile provides a lower bound for the payoffs of all agents under $\pi$.
\begin{lemma}\label{observation_imitation}
    Suppose that $\pi$ and $\pi'$ satisfy $\supp{\mu'}=\{0,\mu_{0},1\}$ and $\min\{\pi(\mu=0|L),\pi(\mu=1|H)\}\geq 1-p$, where $p=\pi'(\mu'=\mu_{0}|L)=\pi'(\mu'=\mu_{0}|H)$.
    Then, $\pi \succsim_{S} \pi'$.
\end{lemma}
\begin{proof}[Proof of Lemma \ref{observation_imitation}]
    Let $q_{L}=\frac{\pi'(\mu'=0|L)}{\pi(\mu=0|L)} $ and $q_{H}=\frac{\pi'(\mu'=1|H)}{\pi(\mu=1|H)}$. Take any $\mathcal{D}$ and
    define $\bm{\sigma}^{**}=(\sigma_{i}^{**})_{i\in \mathbb{N}}$ as the following strategy under $(\mathcal{D},\pi)$. 
    Agent 1 chooses $a_{0}\in B^{-1}(0)$ with probability $q_{L}$ and chooses $a_{2}\in B^{-1}(\mu_{0})$ with probability $1-q_{L}$ if he receives conclusive signal about $\omega=L$. 
    Agent 1 chooses $a_{1}\in B^{-1}(1)$ with probability $q_{H}$ and chooses $a_{2}$ with probability $1-q_{H}$ if he receives conclusive signal about $\omega=H$. 
    Otherwise, agent 1 chooses $a_{2}$. 
    For $i\geq 2$, agent $i$ chooses $a_{0}$ with probability $q_{L}$ and chooses the same action as agent $i-1$ with probability $1-q_{L}$ if he receives a conclusive signal about $\omega=L$. 
    Agent $i$ chooses $a_{1}$ with probability $q_{H}$ and chooses the same action as agent $i-1$ with probability $1-q_{H}$ if he receives a conclusive signal about $\omega=H$. 
    Otherwise, agent $i$ chooses the same action as agent $i-1$.
    First, note that 
    \begin{align*}
        V_{i}^{\mathcal{D}}(\pi,\bm{\sigma}^{**})
        &=\mu_{0}[(1-p^{i})U_{1}]+(1-\mu_{0})[(1-p^{i})U_{0}]+p^{i}U_{\mu_{0}}\\
        &=\overline{V}_{i}^{\mathcal{D}}(\pi'),
    \end{align*}
    where the last equality comes from Lemma~\ref{lemma_expected_payoff_3support}. 
    
    Fix any equilibrium $\bm{\sigma}^{*}$ under $(\mathcal{D},\pi)$ and 
    define $\bm{\sigma}(k)$ as
    \begin{align*}
        \bm{\sigma}(k)=(\sigma^{*}_1,\sigma^{*}_2,\dots,\sigma^{*}_k,\sigma^{**}_{k+1},\sigma^{**}_{k+2},\dots).
    \end{align*}
    We show that if $i\geq k+1$,
    \begin{align*}
         V_{i}^{\mathcal{D}}(\pi,\bm{\sigma}(k))=\mu_{0}(1-p)U_{1}+(1-\mu_{0})(1-p)U_{0}+p V_{i-1}^{\mathcal{D}}(\pi,\bm{\sigma}(k)).
    \end{align*}
Note that
    \begin{align*}
    V_{i-1}^{\mathcal{D}}(\pi,\bm{\sigma}(k))=&\mu_{0}\sum_{a\in A}\alpha_{i-1}^{H}(a|\pi,\bm{\sigma}(k))u(a,H)\\
   & \quad \quad \quad \quad +(1-\mu_{0})\sum_{a\in A}\alpha_{i-1}^{L}(a|\pi,\bm{\sigma}(k))u(a,L).
    \end{align*}
    $i\geq k+1$ implies that $\bm{\sigma}(k)_{i}=\sigma^{**}_{i}$. 
    Hence,
     \begin{align*}
        V_{i}^{\mathcal{D}}(\pi,\bm{\sigma}(k))
        &= \mu_{0}\left[\pi(\mu=1|H)q_{H}U_{1}+(1-\pi(\mu=1|H)q_{H})\sum_{a\in A}\alpha_{i-1}^{H}(a|\pi,\bm{\sigma}(k))u(a,H)\right]\\
        & + (1-\mu_{0})\left[\pi(\mu=0|L)q_{L}U_{0}+(1-\pi(\mu=0|L)q_{L})\sum_{a\in A}\alpha_{i-1}^{L}(a|\pi,\bm{\sigma}(k))u(a,L)\right]\\
        & = \mu_{0}\left[(1-p)U_{1}+p\sum_{a\in A}\alpha_{i-1}^{H}(a|\pi,\bm{\sigma}(k))u(a,H)\right]\\
        &\quad \quad \quad \quad \quad \quad +(1-\mu_{0})\left[(1-p)U_{0}+p\sum_{a\in A}\alpha_{i-1}^{L}(a|\pi,\bm{\sigma}(k))u(a,L)\right]\\
        & = \mu_{0} (1-p)U_{1}+(1-\mu_{0})(1-p)U_{0}+pV_{i-1}^{\mathcal{D}}(\pi,\bm{\sigma}(k)).
    \end{align*}
    By the definition of $\bm{\sigma}(k)$,
    \begin{align*}
        \begin{cases}
            V_{i}^{\mathcal{D}}(\pi,\bm{\sigma}^{*})=V_{i}^{\mathcal{D}}(\pi,\bm{\sigma}(k)) & \text{ if } i<k+1,\\
            V_{i}^{\mathcal{D}}(\pi,\bm{\sigma}^{*})\geq V_{i}^{\mathcal{D}}(\pi,\bm{\sigma}(k)) & \text{ if } i=k+1.
        \end{cases}
    \end{align*}
    The second inequality is held by the optimality of $\sigma^{*}_{i}$.
    We now show that 
    \begin{align*}
        V_{i}^{\mathcal{D}}(\pi,\bm{\sigma}(k+1))\geq V_{i}^{\mathcal{D}}(\pi,\bm{\sigma}(k))
    \end{align*}
    for all $i$ and $k$. 
    First, if $k\geq i-1$, we have $V_{i}^{\mathcal{D}}(\pi,\bm{\sigma}(k+1))= V_{i}^{\mathcal{D}}(\pi,\bm{\sigma}^{*}) \geq V_{i}^{\mathcal{D}}(\pi,\bm{\sigma}(k))$. 
    Next, we have $ V_{i}^{\mathcal{D}}(\pi,\bm{\sigma}(i-1))\geq V_{i}^{\mathcal{D}}(\pi,\bm{\sigma}(i-2))$ for $i\geq 2$ since
    \begin{align*}
         V_{i}^{\mathcal{D}}(\pi,\bm{\sigma}(i-2))&=\mu_{0}(1-p)U_{1}+(1-\mu_{0})(1-p)U_{0}+pV_{i-1}^{\mathcal{D}}(\pi,\bm{\sigma}(i-2))\\
         &\leq \mu_{0}(1-p)U_{1}+(1-\mu_{0})(1-p)U_{0}+pV_{i-1}^{\mathcal{D}}(\pi,\bm{\sigma}(i-1))\\
         &=V_{i}^{\mathcal{D}}(\pi,\bm{\sigma}(i-1)).
    \end{align*}
    Then, we have $ V_{i}(\pi,\bm{\sigma}(i-2))\geq V_{i}(\pi,\bm{\sigma}(i-3))$ for $i\geq 3$ since
     \begin{align*}
         V_{i}^{\mathcal{D}}(\pi,\bm{\sigma}(i-3))&=\mu_{0}(1-p)U_{1}+(1-\mu_{0})(1-p)U_{0}+pV_{i-1}^{\mathcal{D}}(\pi,\bm{\sigma}(i-3))\\
         &\leq \mu_{0}(1-p)U_{1}+(1-\mu_{0})(1-p)U_{0}+pV_{i-1}^{\mathcal{D}}(\pi,\bm{\sigma}(i-2))\\
         &=V_{i}^{\mathcal{D}}(\pi,\bm{\sigma}(i-2)).
    \end{align*}
    Analogously, it follows that $V_{i}^{\mathcal{D}}(\pi,\bm{\sigma}(i-m))\geq V_{i}^{\mathcal{D}}(\pi,\bm{\sigma}(i-m-1))$ for all $i,m$ that satisfies $i-m-1\geq 0$. 
    Hence, $V_{i}^{\mathcal{D}}(\pi,\bm{\sigma}(k+1))\geq V_{i}^{\mathcal{D}}(\pi,\bm{\sigma}(k))$ for all $i,k$. 
    Therefore, we have
    \[
    V_{i}^{\mathcal{D}}(\pi,\bm{\sigma}^{*})
    = V_{i}^{\mathcal{D}}(\pi,\bm{\sigma}(i)) \geq V_{i}^{\mathcal{D}}(\pi,\bm{\sigma}(0)) =V_{i}^{\mathcal{D}}(\pi,\bm{\sigma}^{**})=\overline{V}_{i}^{\mathcal{D}}(\pi').
    \]
\end{proof}
\begin{proof}[Proof of Theorem \ref{thm_sufficient_condition}]
    Suppose that $\pi\succsim_{B} \pi'' \succsim_{B} \pi'$ and $\supp{\mu''}=\{0,\mu_{0},1\}$.
    From Lemma \ref{observation_strict_order_3support}, we conclude that $\pi''\succsim_{S} \pi'$ holds. Since $\pi\succsim_{B} \pi''$ and $\supp{\mu''}=\{0,\mu_{0},1\}$, it follows that $\min\{\pi(\mu=0|L),\pi(\mu=1|H)\}\geq \pi''(\mu''=0|L)=\pi''(\mu''=1|H)$. 
    Thus,
    from Lemma \ref{observation_imitation}, we also conclude that $\pi\succsim_{S} \pi''$ holds. 
    Therefore, we have $\pi\succsim_{S} \pi'$.
\end{proof}
\subsection{Proof of Proposition~\ref{prop_equivalent_sufficient_condition}}
\begin{proof}[Proof of Proposition~\ref{prop_equivalent_sufficient_condition}]
     Let $S'=\supp{\pi'}$.
     Suppose $\supp{\mu''}=\{0,\mu_{0},1\}$. 
     We now show that $\pi \succsim_{B} \pi''$ is equivalent to
     \[
     \pi(\mu=0|L)\geq \pi''(\mu''=0|L) \quad \text{and} \quad \pi(\mu=1|H)\geq \pi''(\mu''=1|H).
     \]
     Note that 
     \begin{align*}
         &\pi''(\mu''=0|\omega)=\frac{\pi''(\mu''=0|L)}{\pi(\mu=0|L)}\pi(\mu=0|\omega)\\
         &\pi''(\mu''=1|\omega)=\frac{\pi''(\mu''=1|H)}{\pi(\mu=1|H)}\pi(\mu=1|\omega)\\
         &\pi''(\mu''=\mu_{0}|\omega)=\left[1-\frac{\pi''(\mu''=0|L)}{\pi(\mu=0|L)}\right]
         \pi(\mu=0|\omega)+
         \left[1-\frac{\pi''(\mu''=1|H)}{\pi(\mu=1|H)}\right]\pi(\mu=1|\omega)\\
         &\hspace{7cm} 
         +\sum_{x\in \supp{\mu}\setminus\{0,1\}}\pi(\mu=x|\omega).
     \end{align*}
     Thus, $\pi''$ is a garbling of $\pi$ if $\pi(\mu=0|L)\geq \pi''(\mu''=0|L)$ and $ \pi(\mu=1|H)\geq \pi''(\mu''=1|H)$. 
     Conversely, suppose $\pi''$ is a garbling of $\pi$. 
     Then, we have 
     \begin{align*}
         &\pi''(\mu''=0|\omega)=\sum_{x\in \supp{\mu}}\gamma(x)\pi(\mu=x|\omega)
     \end{align*}
     for some $\gamma:\supp{\mu}\to[0,1]$. Since $\pi''(\mu''=0|H)=0$ and $\pi(\mu=x|H)>0$ for all $x\neq 0$, it follows that $\gamma(x)=0$ for all $x\neq 0$. 
     Hence, $\pi(\mu=0|L)\geq \pi''(\mu''=0|L)$. 
     Similarly, we have $\pi(\mu=1|H)\geq \pi''(\mu''=1|H).$ 
     
     Next, we show that $\pi'' \succsim_{B} \pi'$ is equivalent to
       \begin{align*}
     \pi''(\mu''=\mu_{0}|L)&=\pi''(\mu''=\mu_{0}|H)\leq \sum_{s\in S'}\min\{\pi'(s|L),\pi'(s|H)\}.
      \end{align*}
      Suppose $\pi''(\mu''=\mu_{0}|L)=\pi''(\mu''=\mu_{0}|H)
     \leq \sum_{s\in S'}\min\{\pi'(s|L),\pi'(s|H)\}$. 
     Define $\rho:\Omega\to \Delta\{s_{0},s_{1},s_{2}\}$ that satisfies
     \begin{align*}
         \rho(s_{1}|L)=0, \quad \rho(s_{0}|H)=0, \quad \rho(s_{2}|H)=\rho(s_{2}|L)=\sum_{s\in S'}\min\{\pi'(s|L),\pi'(s|H)\}.
     \end{align*}
     Then, we have $\pi''\succsim_{B}\rho$ as $\supp{\mu''}=\{0,\mu_{0},1\}$.
     
     If $\rho(s_{2}|L)=\rho(s_{2}|H)=1$, $\rho\succsim_{B} \pi'$ as $\pi'$ is no information. 
     If $\rho(s_{2}|L)=\rho(s_{2}|H)=0$, $\rho\succsim_{B} \pi'$ as both $\rho$ and $\pi'$ are full information. 
     Otherwise, 
     \begin{align*}
         \pi'(s|\omega)=
         &\frac{\max\{\pi'(s|L)-\pi'(s|H),0\}}{\rho(s_{0}|L)}\rho(s_{0}|\omega)+\frac{\max\{\pi'(s|H)-\pi'(s|L),0\}}{\rho(s_{1}|H)}\rho(s_{1}|\omega)\\&\quad\quad\quad\quad\quad\quad+\frac{\min\{\pi'(s|L),\pi'(s|H)\}}{\rho(s_{2}|L)}\rho(s_{2}|\omega)
     \end{align*}
     and
      \begin{align*}
         &\sum_{s\in S'}\frac{\max\{\pi'(s|L)-\pi'(s|H),0\}}{\rho(s_{0}|L)}
         =\frac{\sum_{s\in S'}\max\{\pi'(s|L)-\pi'(s|H),0\}}{1-\sum_{s\in S'}\min\{\pi'(s|L),\pi'(s|H)\}}=1\\
           &\sum_{s\in S'}\frac{\max\{\pi'(s|H)-\pi'(s|L),0\}}{\rho(s_{1}|H)}
         =\frac{\sum_{s\in S'}\max\{\pi'(s|H)-\pi'(s|L),0\}}{1-\sum_{s\in S'}\min\{\pi'(s|L),\pi'(s|H)\}}=1\\
         &\sum_{s\in S'}\frac{\min\{\pi'(s|L),\pi'(s|H)\}}{\rho(s_{2}|L)}=1.
     \end{align*}
     Hence, $\pi'$ is a garbling of $\rho$ and we have $\rho\succsim_{B} \pi'$. 
     Note that $\pi''\succsim_{B} \rho$ and $\rho\succsim_{B} \pi'$ implies $\pi''\succsim_{B}\pi'$. 
     Therefore,  $ \pi''(\mu''=\mu_{0}|H)=\pi''(\mu''=\mu_{0}|L)\leq \sum_{s\in S'}\min\{\pi'(s|L),\pi'(s|H)\}$ is a sufficient condition for $\pi'' \succsim_{B} \pi'$.

     Conversely, suppose  $\pi'' \succsim_{B} \pi'$. 
     Then, there exists probability distribution $\gamma_{0}$, $\gamma_{1}$, and $\gamma_{\mu_{0}}$ over $S'$ such that
     \[
     \pi'(s|\omega)=\gamma_{0}(s)\pi''(\mu''=0|\omega)+\gamma_{1}(s)\pi''(\mu''=1|\omega)+\gamma_{\mu_{0}}(s)\pi''(\mu''=\mu_{0}|\omega)
     \]
     for all $s\in S'$ and $\omega\in \Omega$.
     Then, for each $\omega\in\Omega$,
     \begin{align*}
        & \sum_{s\in S'}\min\{\pi'(s|L),\pi'(s|H)\}\\
        &= \sum_{s\in S'} \min\left\{
\begin{aligned}
    & \gamma_{0}(s)\pi''(\mu''=0|L)+\gamma_{\mu_{0}}(s)\pi''(\mu''=\mu_{0}|L), \\
    & \quad \quad \quad  \gamma_{1}(s)\pi''(\mu''=1|H)+\gamma_{\mu_{0}}(s)\pi''(\mu''=\mu_{0}|H)
\end{aligned}
\right\}\\
        &=
        \sum_{s\in S'}\left[\min\left\{\gamma_{0}(s)\pi''(\mu''=0|L),\gamma_{1}(s)\pi''(\mu''=1|H)\right\}+\gamma_{\mu_{0}}(s)\pi''(\mu''=\mu_{0}|L)\right]\\
        & \geq
        \sum_{s\in S'}\gamma_{\mu_{0}}(s)\pi''(\mu''=\mu_{0}|L)\\
        &= \pi''(\mu''=\mu_{0}|\omega).
     \end{align*}  
      Hence, $ \pi''(\mu''=\mu_{0}|L)=\pi''(\mu''=\mu_{0}|H)\leq \sum_{s\in S'}\min\{\pi'(s|L),\pi'(s|H)\}$  is a necessary condition for $\pi'' \succsim_{B} \pi'$.
      Therefore, $\pi'' \succsim_{B} \pi'$ is equivalent to  $ \pi''(\mu''=\mu_{0}|L)=\pi''(\mu''=\mu_{0}|H)\leq \sum_{s\in S'}\min\{\pi'(s|L),\pi'(s|H)\}$, or  $ \pi''(\mu''=0|L)=\pi''(\mu''=1|H)\geq 1- \sum_{s\in S'}\min\{\pi'(s|L),\pi'(s|H)\}$. 
      By combining the first half and the second half, it can be seen that Proposition~\ref{prop_equivalent_sufficient_condition} holds.
\end{proof}
\subsection{Derivation in Example \ref{example: converse} and Proof of Corollary \ref{prop: converse}}
\begin{proof}[Derivation in Example \ref{example: converse}]
        For the first step, we derive the expected payoff $\overline{V}^{\overline{\mathcal{D}}}_{i}(\pi')$ for each sufficiently large $i$.
        We first decompose the expected payoff into two parts.
        First, if an agent observes a conclusive signal about $H$ (resp. $L$), then she chooses the optimal action $a_{1}$ (resp. $a_{0}$).
        Hence, the expected payoff from the events in which at least one of the first $i$ agents observes a conclusive signal is $\frac{1}{2}\cdot [1-(1-\lambda')^{i}](1-r)$.
        We then focus on the event in which all agents from $1$ to $i$ observe either $s_{l}$ or $s_{h}$ and derive the corresponding interim payoff.
        Note that this interim payoff coincides with that obtained from $i$ conditionally independent observations of the symmetric binary signal $\pi'_{p'}$ where $S_{p'}=\{s_{l},s_{h}\}$, $\pi'_{p'}(s_{l}|L)=\pi'_{p'}(s_{h}|H)=1-p'$, and $\pi'_{p'}(s_{h}|L)=\pi'_{p'}(s_{l}|H)=p'$ with $p'\in(0,1/2)$.
        
        Let $l(s)$ denote the log-likelihood ratio after observing signal $s$ from $\pi'_{p'}$, that is, $l(s)=\log \frac{\pi'_{p'}(s|H)}{\pi'_{p'}(s|L)}$.
        Then, we have $\mathbb{E}[l(s)|L] < 0 < \mathbb{E}[l(s)|H]$.
        Let $K^{L}:\mathbb{R}\to\mathbb{R}$ denote the {\it cumulant generating function} conditional on state $L$, defined as follows:
        \[
            K^{L}(t) 
         = \log \sum_{s\in S_{p'}} \pi'_{p'}(s|L)\left(\frac{\pi'_{p'}(s|H)}{\pi'_{p'}(s|L)} \right)^{t} =\log\left((1-p')\left(\frac{p'}{1-p'}\right)^{t} + p' \left(\frac{1-p'}{p'}\right)^{t}\right),
        \]
        for every $t\in \mathbb{R}$.
        
        By Cram{\'e}r's theorem \citep{cramer1938nouveau}, for each $r\in(1-p,1)$ and $p\in (0,1/2)$ (the parameter for $\pi$), we have
        \begin{align*}
           \pi_{p'}'^{\; \otimes i}(\mu'\geq r|L) 
           & = \mathbb{P}\left[l(s_{1})+\dots+l(s_{i})\geq \log \frac{r}{1-r}\mid L\right] \\
           & = e^{i\min_{t}K^{L}(t) + o(i)}.
        \end{align*}
        As $\min_{t}K^{L}(t)=\log(2\sqrt{p'(1-p')})$ by the arithmetic–geometric mean inequality, we have $\pi_{p'}'^{\; \otimes i}(\mu'\geq r|L)=(2\sqrt{p'(1-p')})^{i}e^{o(i)}$.
        By the symmetric argument, we obtain $\pi_{p'}'^{\; \otimes i}(\mu'\leq r|H)=(2\sqrt{p'(1-p')})^{i}e^{o(i)}$.
        
        Therefore, the interim payoff is 
        \begin{align*}
            & \frac{1}{2}\cdot \left((1-\pi_{p'}'^{\; \otimes i}(\mu'\leq r|H))(1-r)+\pi_{p'}'^{\; \otimes i}(\mu'\geq r|L)(-r)\right) \\
            & = \frac{1}{2}\cdot \left(1-r-\left(2\sqrt{p'(1-p')}\right)^{i}e^{o(i)} \right).
        \end{align*}
        Thus, for each sufficiently large $i$, the expected payoff is
        \[\overline{V}^{\overline{\mathcal{D}}}_{i}(\pi')=\frac{1}{2}\left[ [1-(1-\lambda')^{i}](1-r)+(1-\lambda')^{i}\left(1-r-\left(2\sqrt{p'(1-p')}\right)^{i}e^{o(i)} \right)\right].\]
    
    Then, a necessary condition for $\pi \succsim_{S} \pi'$ is $V^{\overline{\mathcal{D}}}_{i}(\pi,\bm{\sigma}) \geq \overline{V}^{\overline{\mathcal{D}}}_{i}(\pi')$ for all $r\in(1-p,1)$ and sufficiently large $i$.
    By the above derivation, the necessary condition is
    \begin{align*}
   (1-\lambda')^{i}\left(2\sqrt{p'(1-p')}\right)^{i}e^{o(i)} \geq (1-\lambda)^{i}(1-r),
   \end{align*}
   which is equivalent to
   \begin{align*} r \geq 1- \left(\frac{1-\lambda'}{1-\lambda} \cdot 2\sqrt{p'(1-p')}\right)^{i}e^{o(i)}.
    \end{align*}
    The left-hand-side is smallest when $r$ is close to $1-p$, and $\pi \succsim_{S} \pi'$ does not hold if the inequality does not hold at $r=1-p$.
    Thus, applying the $1/i$-th power to both sides at $r=1-p$ yields the necessary condition as  
    \[
    \frac{1-\lambda'}{1-\lambda} \cdot 2\sqrt{p'(1-p')}e^{\frac{o(i)}{i}} \geq p^{\frac{1}{i}}.
    \]
    Therefore, for sufficiently large $i$, we can derive a necessary condition as 
    \[
    \lambda \geq 1-2(1-\lambda')\sqrt{p'(1-p')}.
    \]
\end{proof}

\begin{proof}[Proof of Corollary \ref{prop: converse}]
    By Theorem \ref{thm_sufficient_condition}, it is enough to show the necessity part.
    Take any $\pi$ and $\pi'$ such that $\pi(\mu=1|H)=\pi'(\mu'=1|H)>0$ and $\pi(\mu=0|L)=\pi'(\mu'=0|L)>0$.
    As a contraposition, suppose that there is no $\pi''$ such that $\supp{\mu''}=\{0,\mu_{0},1\}$ and $\pi \succsim_{B} \pi'' \succsim_{B} \pi'$.
    Then, under $\pi'$, there exists a signal $s'$ inducing a private belief in $(0,1)\setminus \{\mu_{0}\}$.\footnote{Note that $\pi'$ is not full or no information.}
    As in the construction behind the proof of Corollary \ref{observation_necessity_of_conclusive}, there exists some decision problem and its auxiliary problem, in which the observable signal setting of $\pi'$ yields strictly higher payoff than $\pi$ for some sufficiently later agent $i$ because observing signal $s'$ $i$-times induce the strictly better payoff by $\pi(\mu=1|H)=\pi'(\mu'=1|H)$ and $\pi(\mu=0|L)=\pi'(\mu'=0|L)$.  
\end{proof}
\subsection{Proofs of Proposition \ref{prop: large sample order and long run comparison}, Proposition \ref{thm: sufficient condition for long run comparison}, and Theorem \ref{thm: limit comparison}}
\begin{proof}[Proof of Proposition \ref{prop: large sample order and long run comparison}]
Suppose $\pi \succsim_{ES} \pi'$. 
Take $N\in \mathbb{N}$ such that $V_{i}^{\mathcal{D}}(\pi,\bm{\sigma}^{*})\geq \overline{V}_{i}^{\mathcal{D}}(\pi')$ for all $\mathcal{D}$, $i\geq N$, and equilibrium $\bm{\sigma}^{*}$ under $(\mathcal{D},\pi)$. 
By Lemma \ref{lemma_signal_is_more_informative}, we have $\overline{V}_{i}^{\mathcal{D}}(\pi)\geq \overline{V}_{i}^{\mathcal{D}}(\pi')$ for all $\mathcal{D}$ and $i\geq N$. 
Thus, $\pi \succsim_{EB} \pi'$.

Consider the same example as Example \ref{example_1}. 
Then, we have $V_{i}^{\mathcal{D}}(\pi',\bm{\sigma}')>V_{i}^{\mathcal{D}}(\pi,\bm{\sigma})$ for all $i\geq 2$. 
Hence, $\pi$ is not eventually more socially valuable than $\pi'$, although $\pi \succsim_{EB} \pi'$ holds as $\pi \succsim_{B} \pi'$ and Lemma \ref{observation_blackwell}.
\end{proof}
\begin{proof}[Proof of Proposition \ref{thm: sufficient condition for long run comparison}]
    Since $\succsim_{ES}$ satisfies transitivity, showing $\pi \succsim_{ES} \pi''$ and $\pi'' \succsim_{ES} \pi'$ is sufficient.
    First, we have $\pi'' \succsim_{ES} \pi'$ because
    \[
        V_{i}^{\mathcal{D}}(\pi'',\bm{\sigma}^{*})=\overline{V}_{i}^{\mathcal{D}}(\pi'')
        \geq \overline{V}_{i}^{\mathcal{D}}(\pi')\geq  V_{i}^{\mathcal{D}}(\pi',\bm{\sigma}^{**})
    \]
    for sufficiently large $i$, where $\bm{\sigma}^{*}$ is an arbitrary equilibrium under $(\mathcal{D},\pi'')$ and $\bm{\sigma}^{**}$ is an arbitrary equilibrium under $(\mathcal{D},\pi')$.
    Each equality or inequality follows from Lemma \ref{lemma_expected_payoff_3support}, $\pi'' \succsim_{EB} \pi'$, and Lemma \ref{lemma_signal_is_more_informative} respectively. 
    Thus, $\pi'' \succsim_{ES} \pi'$.
    
    To show $\pi \succsim_{ES} \pi''$, note that if $\supp{\mu''}=\{0,\mu_{0},1\}$, the support of private belief distribution induced by $\pi''^{\; \otimes i}$ is also $\{0,\mu_{0},1\}$ for all $i$.
    Then by the proof of Proposition \ref{prop_equivalent_sufficient_condition}, $\pi\succsim_{EB}\pi''$ if and only if there exists $N\in\mathbb{N}$ such that
    $\pi^{\otimes i}(\mu=0|L) \geq \pi''^{\; \otimes i}(\mu''=0|L)$ and $\pi^{\otimes i}(\mu=1|H) \geq \pi''^{\; \otimes i}(\mu''=1|H)$ for all $i\geq N$.
    By the nature of conclusive signals, agents form the extreme beliefs $\mu=1$ or $0$ if they observe the conclusive signal at least once.
    Thus, these conditions are equivalent to
    \[
     1-[1-\pi(\mu=0|L)]^{i}\geq 1-[1-\pi''(\mu''=0|L)]^{i},
     \]
     and 
     \[1-[1-\pi(\mu=1|H)]^{i}\geq 1-[1-\pi''(\mu''=1|H)]^{i},\]
     for all sufficiently large $i$.
     Furthermore, these inequalities are equivalent to 
     \[
     \pi(\mu=0|L)\geq \pi''(\mu''=0|L) \quad \text{and} \quad \pi(\mu=1|H)\geq \pi''(\mu''=1|H).
     \]
     Thus, under $\supp{\mu''}=\{0,\mu_{0},1\}$, $\pi\succsim_{EB} \pi''$ if and only if $\pi\succsim_{B} \pi''$.
     Therefore, we also have $\pi \succsim_{S} \pi''$ by Lemma \ref{observation_imitation}.
\end{proof}

\begin{proof}[Proof of Theorem \ref{thm: limit comparison}]
    Suppose that $\pi$ does not induce unbounded beliefs. 
    By symmetry, we can assume $1\notin \supp{\mu}$ without loss of generality. 
    Take $r\in (0,1)$ such that $ \supp{\mu} \subset [0,r]$.
    Let $\mathcal{D}=(A,u)$ be the decision problem that is the same as in the proof of Corollary \ref{observation_necessity_of_conclusive}. 
    Define $\overline{\mathcal{D}}$ as the replicated decision problem in the same way as in the proof of Theorem \ref{thm_characterization}. 
    Then, we can construct the equilibrium $\bm{\sigma}^{**}$ under $(\overline{\mathcal{D}},\pi')$ whose history completely reveals the signals received by each agent. 
    Fix $k_{1}\in \supp{\pi'}.$
    Note that it is always an equilibrium under $(\overline{\mathcal{D}}, \pi)$ for all agents to choose $(a_{0}, k_{1})$.
    Denote this equilibrium by $\bm{\sigma}^{*}$.
    Then, it follows that $  V_{i}^{\overline{\mathcal{D}}}(\pi,\bm{\sigma}^{*})=0$ for all $i$.
    Since $\pi'$ is not no information, for any $\varepsilon>0$, $V_{i}^{\overline{\mathcal{D}}}(\pi',\bm{\sigma}^{**})>\mu_{0}(1-r)-\varepsilon$ for sufficiently large $i$. 
    Thus, $\lim_{i\to \infty}V_{i}^{\overline{\mathcal{D}}}(\pi',\bm{\sigma}^{**})>\lim_{i\to \infty}V_{i}^{\overline{\mathcal{D}}}(\pi,\bm{\sigma}^{*})$, and hence $\pi \succsim_{LES} \pi'$ does not hold.
    
    Conversely, suppose that $\pi$ induces unbounded beliefs.
    Note that $V_{i}^{\mathcal{D}}(\pi',\bm{\sigma}^{**})\leq \mu_{0} U_{1}+(1-\mu_{0})U_{0}$, where $U_{0}=\max_{a\in A} u(a,L)$ and  $U_{1}=\max_{a\in A} u(a,H)$, because the right-hand side is an expected utility under full information. 
    Hence, it is enough to show that, for any $\mathcal{D}$, $\lim_{i\to \infty}V_{i}^{\mathcal{D}}(\pi,\bm{\sigma}^{*})=\mu_{0} U_{1}+(1-\mu_{0})U_{0}$ for any equilibrium $\bm{\sigma}^{*}$ under $(\mathcal{D},\pi)$.
    Take any decision problem $\mathcal{D}=(A,u)$. 
    Let $q_{L}=\frac{\min\{\pi(\mu=1|H),\pi(\mu=0|L)\}}{\pi(\mu=0|L)}$ and $q_{H}=\frac{\min\{\pi(\mu=1|H),\pi(\mu=0|L)\}}{\pi(\mu=1|H)}$. 
    Note that $q_{L},q_{H}>0$ holds since $\pi$ induces unbounded beliefs. 
    Define $\bm{\sigma}^{**}$ in the same way as in the proof of Lemma \ref{observation_imitation}. 
    Then, by the same argument as the proof of Lemma \ref{observation_imitation}, we have $V_{i}^{\mathcal{D}}(\pi,\bm{\sigma}^{*})\geq V_{i}^{\mathcal{D}}(\pi,\bm{\sigma}^{**})$ for any equilibrium $\bm{\sigma}^{*}$ under $(\mathcal{D},\pi)$.
    Since 
    \[
   V_{i}^{\mathcal{D}}(\pi,\bm{\sigma}^{**})
        =\mu_{0}[(1-p^{i})U_{1}]+(1-\mu_{0})[(1-p^{i})U_{0}]+p^{i}U_{\mu_{0}},
    \]
    where $p=1-\min\{\pi(\mu=1|H),\pi(\mu=0|L)\}$ and $U_{\mu_{0}}=\max_{a\in A} [\mu_{0}u(a,H)+(1-\mu_{0})u(a,L)]$, it follows that
    \begin{align*}
    \mu_{0} U_{1}+(1-\mu_{0})U_{0} & \geq V_{i}^{\mathcal{D}}(\pi,\bm{\sigma}^{*})\\
    & \geq \mu_{0}[(1-p^{i})U_{1}]+(1-\mu_{0})[(1-p^{i})U_{0}]+p^{i}U_{\mu_{0}}.
    \end{align*}
    Since $p<1$, we have $\mu_{0}[(1-p^{i})U_{1}]+(1-\mu_{0})[(1-p^{i})U_{0}]+p^{i}U_{\mu_{0}}\to \mu_{0} U_{1}+(1-\mu_{0})U_{0}$ as $i\to \infty$. 
    Thus, by the squeeze theorem, it follows that $\lim_{i\to \infty}V_{i}^{\mathcal{D}}(\pi,\bm{\sigma}^{*})=\mu_{0} U_{1}+(1-\mu_{0})U_{0}$ for any $\mathcal{D}$ and any equilibrium $\bm{\sigma}^{*}$ under $(\mathcal{D},\pi)$.
    Thus, $\pi \succsim_{LES} \pi'$ for any $\pi'$.
\end{proof}

\subsection{Proofs of Proposition \ref{prop_partial_order} and Proposition \ref{observation_weak_order_3support}}
\begin{proof}[Proof of Proposition \ref{prop_partial_order}]
    Suppose $\supp\mu\subseteq\{0,\mu_{0},1\}$. 
    Then, we have
    \begin{align*}
        1-\sum_{s\in\supp{\pi}}\min\{\pi(s|L),\pi(s|H)\} 
        & = 1-\pi(\mu=\mu_{0}|L) \\
        & = \pi(\mu=0|L)=\min \{\pi(\mu=0|L),\pi(\mu=1|H)\}.
    \end{align*}
    Therefore, by Theorem \ref{thm_sufficient_condition} and Proposition \ref{prop_equivalent_sufficient_condition}, we have $\pi\succsim_{S} \pi$.
    
    Now, we show that $\pi \succsim_{S} \pi$ does not hold if $\supp{\mu}\nsubseteq\{0,\mu_{0},1\} $. 
    It is sufficient to show for the case where there exists some $x\in(\mu_{0},1)$ such that $x\in \supp\mu $. 
    Take $r\in [0,1]$ that satisfies $x<r<\frac{x^2}{x^2+\frac{\mu_{0}}{1-\mu_{0}}(1-x)^2}$.
    Consider the decision problem $\mathcal{D}=(A,u) $: $A=\{a_{0},a_{1}\}$ and the payoff function is defined as $u(a_{0},H)=u(a_{0},L)=0$, $u(a_{1},H)=1-r$, and $u(a_{1},L)=-r$. 
    Take any equilibrium $\bm{\sigma}^{*}=(\sigma^{*}_{i})_{i\in\mathbb{N}}$ and $s_{1},s_{2}\in S(x)$.
    Then, it follows that $\sigma^{*}_1(a_{0}|s_{1})=1$ and $\sigma^{*}_2(a_{0}|a_{0},s_{2})=1$. 
    Thus, $V_{2}^{\mathcal{D}}(\pi,\bm{\sigma}^{*}|(s_{1},s_{2}))=0$.
   Additionally, we have $\overline{V}_2^{\mathcal{D}}(\pi|(s_{1},s_{2}))>0$ since $r<\frac{x^2}{x^2+\frac{\mu_{0}}{1-\mu_{0}}(1-x)^2}$. 
   Note that for all $s_{1}',s_{2}'\in S$, $ \overline{V}_2^{\mathcal{D}}(\pi|(s_{1}',s_{2}'))\geq V_{2}^{\mathcal{D}}(\pi,\bm{\sigma}^{*}|(s_{1}',s_{2}'))$.\footnote{This statement follows from the same argument as in Lemma \ref{lemma_signal_is_more_informative}.} 
   Therefore, $\overline{V}_2^{\mathcal{D}}(\pi)>V_{2}^{\mathcal{D}}(\pi,\bm{\sigma}^{*})$.
   By Theorem \ref{thm_characterization}, it follows that $\pi\succsim_{S}\pi$ does not hold.
\end{proof}

\begin{proof}[Proof of Proposition \ref{observation_weak_order_3support}]
  Without loss of generality, assume that $x>\mu_{0}$, $\supp\pi=\{s_{0},s_{1},s_{2}\}$ and $\pi(s_{0}|H)=0$, $\pi(s_{1}|H)=1-\varepsilon$, $\pi(s_{2}|H)=\varepsilon$, $\pi(s_{0}|L)=1-\delta$, $\pi(s_{1}|L)=0$, and $\pi(s_{2}|L)=\delta$, where $\varepsilon$ and $\delta$ satisfy the condition that $x=\frac{\mu_{0}\varepsilon}{\mu_{0}\varepsilon+(1-\mu_{0})\delta}$. 
  We divide decision problem $\mathcal{D}$ into three cases and construct the following equilibrium $\bm{\sigma}^{*}$ under $(\mathcal{D},\pi)$. 
  
    \textit{Case (i): $B^{-1}(0)\cap B^{-1}(1)\neq \emptyset$.}
    Fix $a^{*}\in B^{-1}(0)\cap B^{-1}(1)$.
    In this case, all agents choose $a^{*}$ regardless of private signal and action histories.  
    
     
    \textit{Case (ii): $B^{-1}(1)\cap B^{-1}(x)= \emptyset$, 
    $B^{-1}(0)= B^{-1}(x)=\{a_{0}\}$ for some $a_{0}\in A$.}
    Fix any $a_{1}\in B^{-1}(1)$.
    Agent 1 chooses $a_{0}$ if he receives $s_{0}$ or $s_{2}$ and chooses $a_{1}$ otherwise. For $i\geq 2$, agent $i$ chooses $a_{0}$ if she receives $s_{0}$, or receives $s_{2}$ and all previous agents take $a_{0}$. 
    Otherwise, $i$ chooses $a_{1}$.
     
     \textit{Case (iii): Otherwise.}
     First, fix $a_{0}\in B^{-1}(0)$ such that for all $z\in [x,1]$, $B^{-1}(z)\neq \{a_{0}\}$. (Such $a_{0}$ must exist by Lemma~\ref{lemma_best_action}.)
     In this case, agent $1$ chooses action $a_{0}$ if he receives $s_{0}$, chooses action from $B^{-1}(1)$ if he receives $s_{1}$, and chooses action from $B^{-1}(x)$ if he receives $s_{2}$.
     For $i\geq 2,$ agent $i$ chooses action $a_{0}$ if she receives $s_{0}$ or at least one agent before $i$ has taken $a_{0}$, chooses action from $B^{-1}(1)$ if she receives $s_{1}$, and chooses action from $B^{-1}(\frac{x^{i}}{x^{i}+(\frac{\mu_{0}}{1-\mu_{0}})^{i-1}(1-x)^{i}})\setminus\{a_{0}\}$ if she receives $s_{2}$ and no one before $i$ has taken action $a_{0}$ or action from $B^{-1}(1)$. Otherwise, she chooses the same action as agent $i-1$.

    In \textit{Case (i)}, it is always optimal to take $a^{*}$ regardless of the posterior belief. Hence, this strategy $\bm{\sigma}^{*}$ is an equilibrium and we have $V_{i}^{\mathcal{D}}(\pi,\bm{\sigma}^{*})=\overline{V}_{i}^{\mathcal{D}}(\pi)$.
    In \textit{Case (iii)}, action $a_{0}$ is taken if someone has received the signal $s_{0}$ in the past, an action from $B^{-1}(1)$ is taken if someone has received the signal $s_{1}$ in the past, and an action from $B^{-1}(\frac{x^{i}}{x^{i}+(\frac{\mu_{0}}{1-\mu_{0}})^{i-1}(1-x)^{i}})$ or an action yielding the same expected payoff is taken when all past agents have received $s_{2}$.
    Therefore, we have $ V_{i}^{\mathcal{D}}(\pi,\bm{\sigma}^{*})=\overline{V}_{i}^{\mathcal{D}}(\pi)$. Hence, $\bm{\sigma}^{*}$ is an equilibrium. 
    Then, in \textit{Case (i)} and \textit{Case (iii)}, by the same argument as Lemma~\ref{observation_strict_order_3support},
    \begin{align*}
       V_{i}^{\mathcal{D}}(\pi,\bm{\sigma}^{*})= \overline{V}_{i}^{\mathcal{D}}(\pi)\geq  \overline{V}_{i}^{\mathcal{D}}(\pi')\geq V_{i}^{\mathcal{D}}(\pi',\bm{\sigma}^{**}),
    \end{align*}
    for any equilibrium $\bm{\sigma}^{**}$ under $(\mathcal{D},\pi')$.
   
    The only case left is \textit{Case (ii)}. 
    In \textit{Case (ii)}, from Lemma \ref{lemma_expected_payoff_special_case}, 
    \begin{align*}
    V_{i}^{\mathcal{D}}(\pi',\bm{\sigma}^{**})
    & = \mu_{0}[(1-(1-\pi'(\mu'=1|H))^{i})u(a_{1},H)\\
    &\qquad + (1-\pi'(\mu'=1|H))^{i}u(a_{0},H)]+(1-\mu_{0})u(a_{0},L),
    \end{align*}
    for any equilibrium $\bm{\sigma}^{**}$ under $(\mathcal{D},\pi')$.
    Since $\pi'(\mu'=1|H)\leq 1-\varepsilon$ (by $\pi \succsim_{B} \pi'$) and $u(a_{1},H)>u(a_{0},H)$, it follows that
    \begin{align*}
        &\mu_{0}[(1-(1-\pi'(\mu'=1|H))^{i})u(a_{1},H)+(1-\pi'(\mu'=1|H))^{i}u(a_{0},H)]+(1-\mu_{0})u(a_{0},L)\\
        &= \mu_{0}u(a_{1},H)-\mu_{0}(1-\pi'(\mu'=1|H))^{i}[u(a_{1},H)-u(a_{0},H)]+(1-\mu_{0})u(a_{0},L)\\
        & \leq \mu_{0}u(a_{1},H)-\mu_{0}\varepsilon^{i}[u(a_{1},H)-u(a_{0},H)]+(1-\mu_{0})u(a_{0},L)\\
        & = V_{i}^{\mathcal{D}}(\pi,\bm{\sigma}^{*}),
    \end{align*}
    where $\bm{\sigma}^{*}$ is an equilibrium described above.
    Therefore, $\pi \succsim_{W} \pi'$.
\end{proof}


\begingroup
\singlespacing
\def\bibinfo#1#2{#2\ignorespaces}
\setlength{\bibsep}{4.5pt plus 0.3ex}
\bibliography{Reference}
\addcontentsline{toc}{section}{References}
\endgroup

\clearpage
\begingroup
\onehalfspacing

\newcommand{\mainref}[1]{%
  \ifSubfilesClassLoaded{\ref{main-#1}}{\ref{#1}}}

\clearpage

\ifSubfilesClassLoaded{%
\title{Online Appendix \\
Value of Information in Social Learning
}
\author{
\Large Hiroto Sato\thanks{\protect \onehalfspacing 
Department of Economics, Nagoya University, Furo-cho, Chikusa-ku, Nagoya 464-8601, Japan. 
Email: \url{sato.hiroto.s9@f.mail.nagoya-u.ac.jp}.
}
\and
\Large Konan Shimizu\thanks{\protect \onehalfspacing 
Faculty of Economics, Keio University, 2-15-45 Mita, Minato-ku, Tokyo 108-8345, Japan.
Email: \url{shimizu-konan@keio.jp}.}
} 
\date{\today}

\maketitle
}{%
  \begin{center}
    {\LARGE\bfseries Online Appendix\par}
    \vspace{0.5em}
    {\Large Value of Information in Social Learning\par}
    \vspace{1.25em}
    {\large Hiroto Sato \qquad Konan Shimizu\par}
    \vspace{0.75em}
    {\small
      Department of Economics, Nagoya University, Furo-cho, Chikusa-ku,
      Nagoya 464-8601, Japan.\\
      Email: \url{sato.hiroto.s9@f.mail.nagoya-u.ac.jp}.\\[0.5em]
      Faculty of Economics, Keio University, 2-15-45 Mita, Minato-ku,
      Tokyo 108-8345, Japan.\\
      Email: \url{shimizu-konan@keio.jp}.\par}
    \vspace{0.75em}
    {\today\par}
  \end{center}
  \vspace{1em}
}

\setlength{\abovedisplayskip}{3pt}
\setlength{\belowdisplayskip}{3pt}

\appendix 

\titleformat{\section}
  {\Large\bfseries}
  {Appendix \thesection:}
  {0.5em}
  {}
\renewcommand{\theHsection}{onlineappendix.\Alph{section}}

\numberwithin{theorem}{section}
\numberwithin{proposition}{section}
\numberwithin{corollary}{section}
\numberwithin{lemma}{section}
\numberwithin{claim}{section}
\numberwithin{definition}{section}
\numberwithin{assumption}{section}
\numberwithin{remark}{section}
\numberwithin{example}{section}
\numberwithin{conjecture}{section}
\numberwithin{observation}{section}

\renewcommand{\theHtheorem}{onlineappendix.\thetheorem}
\renewcommand{\theHproposition}{onlineappendix.\theproposition}
\renewcommand{\theHcorollary}{onlineappendix.\thecorollary}
\renewcommand{\theHlemma}{onlineappendix.\thelemma}
\renewcommand{\theHclaim}{onlineappendix.\theclaim}
\renewcommand{\theHdefinition}{onlineappendix.\thedefinition}
\renewcommand{\theHassumption}{onlineappendix.\theassumption}
\renewcommand{\theHremark}{onlineappendix.\theremark}
\renewcommand{\theHexample}{onlineappendix.\theexample}
\renewcommand{\theHconjecture}{onlineappendix.\theconjecture}
\renewcommand{\theHobservation}{onlineappendix.\theobservation}

\section{Extensions} \label{sec:extension}
\subsection{Dynamic Information Structures} \label{sec:adaptive_info}
Throughout Section \ref{sec:main}, we have assumed that agents receive private signals independently from an identical information structure.
In this section, we extend our analysis to dynamic information structures, following the dynamic extensions of the Blackwell comparison developed by \citet{renou2024comparing} and \citet{whitmeyer2024comparisons} under adaptive decision problems.
We show that our main results continue to hold in this setting.

Let $S=\prod_{i=1}^{\infty}S_{i}$, and let $\pi:\Omega\to \Delta(S)$ denote the dynamic information structure.
For each $i$, let $\pi_{i}$ be the marginal distribution of $\pi$ over $S_{i}$, and let $\mu_{i}$ denote the induced distribution of private beliefs.
For a decision problem $\mathcal{D}$ and a strategy profile $\bm{\sigma}$, define $V_{i}^{\mathcal{D}}(\pi,\bm{\sigma})$ as the expected payoff of agent $i$ when she observes both the past action history and her private signal.
Let $\overline{V}_{i}^{\mathcal{D}}(\pi)$ denote the maximized expected payoff under the imaginary setting wherein agent $i$ can observe signals drawn from $\pi_{\leq i}$, the marginal distribution over $S_{1}\times \dots\times S_{i}$ induced by $\pi$.

Note that private signals $s_{i}$ and $s_{j}$ may be correlated, so the public belief can depend on the realization of each agent’s private signal.
Moreover, the marginal distributions $\pi_{i}$ may differ across agents, allowing our framework to naturally encompass settings with heterogeneous private signals.

Then, Theorem \ref{thm_characterization} is naturally extended as follows:
\begin{proposition}\label{prop_characterization_adaptive}
    $\pi\succsim_{S}\pi'$ holds if and only if $V_{i}^{\mathcal{D}}(\pi,\bm{\sigma}^{*})\geq \overline{V}_{i}^{\mathcal{D}}(\pi')$ for any decision problem $\mathcal{D}$, any agent $i$, and any equilibrium $\bm{\sigma}^{*}$ under $(\mathcal{D},\pi)$.
\end{proposition}
\begin{proof}[Proof of Proposition \ref{prop_characterization_adaptive}]
It suffices to show that $\overline{V}_{i}^{\mathcal{D}}(\pi)\geq V_{i}^{\mathcal{D}}(\pi,\bm{\sigma})$ holds for all $i$, $\mathcal{D}$, $\pi$, and $\bm{\sigma}$ because the remaining part of the proof is identical to that of Theorem \ref{thm_characterization}.
Take any $\mathcal{D},\pi$, and $\bm{\sigma}$. 
Note that $ \overline{V}_{1}^{\mathcal{D}}(\pi)= V_{1}^{\mathcal{D}}(\pi,\bm{\sigma})$.
Fix $i\geq 2$. 
For each $\bm{s}\in S_{1}\times \dots \times S_{i-1}$, define $f_{i-1}(\cdot|\bm{s})\in \Delta (A^{i-1})$ as 
\begin{align*}
    f_{i-1}(\bm{a}|\bm{s})=\prod_{k=1}^{i-1}\sigma_{k}(a_{k}|a_{1},\dots,a_{k-1},s_{k}).
\end{align*}
Thus, $f_{i-1}(\bm{a}|\bm{s})$ is the probability that agents $1,\dots,i-1$ take action $\bm{a}=(a_{1},\dots,a_{i-1})$ when they receive private signals $\bm{s}=(s_{1},\dots,s_{i-1})$. 
Let $g^{\omega}_{i}(\bm{a},s_{i}'|\pi,\bm{\sigma})$ denote the probability that agent $i$ observes $(\bm{a},s_{i}')$ when the state is $\omega$ and agents follow strategy profile $\bm{\sigma}$.
Note that 
\begin{align*}
    g^{\omega}_{i}(\bm{a},s_{i}'|\pi,\bm{\sigma})
    & =\sum_{\bm{s}\in S_{1}\times\dots \times S_{i-1}}\prod_{k=1}^{i-1}\sigma_{k}(a_{k}|a_{1},\dots,a_{k-1},s_{k})\pi_{\leq i}(\bm{s},s_{i}'|\omega) \\
    &= \sum_{\bm{s}\in S_{1}\times\dots \times S_{i-1}}f_{i-1}(\bm{a}|\bm{s})\pi_{\leq i}(\bm{s},s_{i}'|\omega) \\
    & = \sum_{\bm{s}\in S_{1}\times\dots \times S_{i}}f_{i-1}(\bm{a}|\bm{s})\pi_{\leq i}(\bm{s}|\omega)\mathds{1}_{s_{i}=s_{i}'}.
\end{align*}
Then, $g_{i}(\cdot|\pi,\bm{\sigma})$ is a garbling of $\pi_{\leq i}$, where $g_{i}(\cdot|\pi,\bm{\sigma})$ is the ex-ante distribution induced by $g_{i}^{\omega}(\cdot|\pi,\bm{\sigma})$.
\end{proof}
The following example shows that the sufficient condition in Theorem \ref{thm_sufficient_condition} does not naturally extend to correlated information structures by imposing its Blackwell comparisons separately on each agent's marginal experiment.\footnote{We thank the Associate Editor for providing this example.}
\begin{example}\label{example_correlated_signals}
Fix $\varepsilon\in(0,1)$.
For each $i$, let $S_{i}=\{s_{0},s_{\mu_{0}},s_{1}\}$, where $s_{0}$ and $s_{1}$ conclusively reveal states $L$ and $H$, respectively, whereas $s_{\mu_{0}}$ is uninformative.
Let $\pi''$ generate signals independently across agents according to the following common marginal experiment:
\begin{table}[H]
\centering
\begin{tabular}{l|ccc}
\Xhline{1.2pt}
$\pi_i''$ & $s_{0}$ & $s_{\mu_{0}}$ & $s_{1}$ \\
\hline
$H$ & $0$ & $\varepsilon$ & $1-\varepsilon$ \\
$L$ & $1-\varepsilon$ & $\varepsilon$ & $0$ \\
\Xhline{1.2pt}
\end{tabular}
\end{table}
\noindent
Set $\pi'=\pi''$.

To define $\pi$, suppose that all agents receive the perfectly same signal realization.
Conditional on the realized state, they all receive the corresponding conclusive signal with probability $1-\varepsilon$ and all receive $s_{\mu_0}$ with probability $\varepsilon$.
Although the joint distributions under $\pi$ and $\pi''$ differ, their one-agent marginals coincide.
Since $\pi'=\pi''$, for every agent $i$ we have
\[
    \pi_i=\pi_i''=\pi_i'
    \qquad\text{and}\qquad
    \supp{\mu_i''}=\{0,\mu_0,1\}.
\]
Thus, both marginal Blackwell comparisons
$\pi_i\succsim_B\pi_i''\succsim_B\pi_i'$ hold with equality, exactly as required by the direct marginal extension of Theorem \ref{thm_sufficient_condition}.

We now show that these conditions do not imply $\pi\succsim_S\pi'$.
Consider the decision problem with $A=\{a_{0},a_{1}\}$ and, for some $r\in(\mu_{0},1)$,
\[
u(a_{0},L)=u(a_{0},H)=0,\qquad
u(a_{1},H)=1-r,\qquad
u(a_{1},L)=-r.
\]
Take any equilibrium $\bm{\sigma}^{*}$ under $(\mathcal{D},\pi)$ and any equilibrium $\bm{\sigma}^{**}$ under $(\mathcal{D},\pi')$.
In either equilibrium, on the equilibrium path, an agent selects $a_1$ whenever she receives $s_1$ or some predecessor has selected $a_1$, and selects $a_0$ in all other cases.
No indifference is involved in these choices, so all equilibria induce the same on-path behavior.
Under $\pi$, the realization in which all agents receive $s_{\mu_0}$ occurs with probability $\varepsilon$ in either state, so observing $s_{\mu_0}$ leaves the posterior equal to $\mu_0<r$.
Under $\pi'$, if all predecessors chose $a_0$ and agent $i$ receives $s_{\mu_0}$, her posterior is
\[
    \frac{\mu_{0}\varepsilon^{i}}{\mu_{0}\varepsilon^{i}+(1-\mu_{0})\varepsilon}
    =\frac{\mu_{0}\varepsilon^{i-1}}{\mu_{0}\varepsilon^{i-1}+1-\mu_{0}}
    \leq \mu_{0}<r.
\]

Under $\pi$, action $a_1$ is taken in state $H$ precisely when the common signal is $s_1$. Hence,
\[
    V_{i}^{\mathcal{D}}(\pi,\bm{\sigma}^{*})=\mu_{0}(1-r)(1-\varepsilon).
\]
Under the independent structure $\pi'=\pi''$, agent $i$ takes $a_1$ in state $H$ whenever at least one of the first $i$ signals is $s_1$. It follows that
\[
    V_{i}^{\mathcal{D}}(\pi',\bm{\sigma}^{**})=\mu_{0}(1-r)(1-\varepsilon^{i}).
\]
For every $i\geq2$,
\[
    V_{i}^{\mathcal{D}}(\pi',\bm{\sigma}^{**})-V_{i}^{\mathcal{D}}(\pi,\bm{\sigma}^{*})
    =\mu_{0}(1-r)(\varepsilon-\varepsilon^{i})>0.
\]
Therefore, the marginal Blackwell comparisons are satisfied, but $\pi\succsim_S\pi'$ does not hold.
The failure comes from information persistence across agents.
Under $\pi$, an uninformative realization indicates that no agent in the sequence receives conclusive information, whereas under $\pi''$ each additional agent provides another independent opportunity to reveal the state. \qed
\end{example}

\subsection{Heterogeneous Decision Problems} \label{sec:heterogeneous}
We assume that all agents are homogeneous, that is, each agent faces the same decision problem.
In this section, we discuss how our results persist when we relax this assumption.

Let $\mathcal{D}_{i}=(A_{i},u_{i})$ be the agent $i$'s decision problem.
\begin{proposition} \label{prop_characterization_hetero}
    $\pi\succsim_{S}\pi'$ holds if and only if $V_{i}^{\mathcal{D}}(\pi,\bm{\sigma}^{*})\geq \overline{V}_{i}^{\mathcal{D}}(\pi')$ for any decision problem $\mathcal{D}=(\mathcal{D}_{i})_{i\in\mathbb{N}}$, any agent $i$, and any equilibrium $\bm{\sigma}^{*}$ under $(\mathcal{D},\pi)$.
\end{proposition}
Thus, the same consequence from Theorem \ref{thm_characterization} holds, and hence, Corollary \ref{observation_necessity_of_conclusive} is still a necessary condition.

However, the following Example \ref{example_hetero} shows that our sufficient condition in Theorem \ref{thm_sufficient_condition} does not hold under the heterogeneous decision problems.
\begin{example}\label{example_hetero}
Consider two information structures $\pi:\Omega\to\Delta(\{s_{0},s_{\mu_{0}},s_{1}\})$ and $\pi':\Omega\to\Delta(\{s'_{l},s'_{1}\})$ defined for fixed $\varepsilon\in(0,1)$ as follows:
\begin{table}[H]
\centering
\makebox[\textwidth][c]{%
\makebox[0pt][r]{%
\begin{tabular}{l|ccc}
\Xhline{1.2pt}
$\pi$ & $s_{0}$ & $s_{\mu_{0}}$ & $s_{1}$ \\
\hline
$H$ & $0$ & $\varepsilon$ & $1-\varepsilon$ \\
$L$ & $1-\varepsilon$ & $\varepsilon$ & $0$ \\
\Xhline{1.2pt}
\end{tabular}
}%
\hspace{1.5cm}%
\makebox[0pt][l]{%
\begin{tabular}{l|cc}
\Xhline{1.2pt}
$\pi'$ & $s'_{l}$ & $s'_{1}$ \\
\hline
$H$ & $\varepsilon$ & $1-\varepsilon$ \\
$L$ & $1$ & $0$ \\
\Xhline{1.2pt}
\end{tabular}
}%
}
\end{table}
\noindent
The information structures induce $\supp{\mu}=\{0,\mu_{0},1\}$ and $\supp{\mu'}=\{l,1\}$, respectively, where $l=\frac{\mu_{0}\varepsilon}{\mu_{0}\varepsilon+1-\mu_{0}}$.
Then, by the construction, we have $\pi\succsim_{B} \pi'$ and $\pi(\mu=1)=\pi'(\mu=1)$.

We now consider the following decision problem $(\mathcal{D}_{1},\mathcal{D}_{2})$:
Let $\mathcal{D}_{1}=(A_{1}, u_{1})$, where $A_{1}=\{a_{0},a_{1}\}$, $u_{1}(a_{0},\omega)=0$, $u_{1}(a_{1},H)=1-r_{1}$, $u_{1}(a_{1},L)=-r_{1}$, and $r_{1}\in (l,\mu_{0})$. 
Additionally, let $\mathcal{D}_{2}=(A_{2}, u_{2})$, where $A_{2}=\{b_0,b_1\}$, $u_{2}(b_0,\omega)=0$, $u_{2}(b_1,H)=1-r_{2}$, $u_{2}(b_1,L)=-r_{2}$, and $r_{2}\in (\frac{\mu_{0}}{\mu_{0}+(1-\mu_{0})\varepsilon},1)$.
Agent $1$ faces $\mathcal{D}_{1}$ and agent 2 faces $\mathcal{D}_{2}$.

Take any equilibrium $\bm{\sigma}$ under $\pi$.
Then, agent $1$ chooses action $a_{1}$ if and only if she receives signals either $s_{1}$ or $s_{\mu_{0}}$.
Now, agent $2$'s posterior belief after observing $a_{1}$ and $s_{\mu_{0}}$ is below the cutoff $r_{2}$, and thus agent $2$ chooses action $b_{1}$ if and only if he observes a conclusive signal $s_{1}$. 
Thus, the expected payoffs are $V_{1}(\pi,\bm{\sigma})=\mu_{0}(1-r_{1})-(1-\mu_{0})\varepsilon r_{1}$ and $V_{2}(\pi,\bm{\sigma})=\mu_{0}(1-\varepsilon)(1-r_{2})$.

By contrast, take any equilibrium $\bm{\sigma}'$ under $\pi'$.
Then, agent $1$ chooses action $a_{1}$ if and only if she observes a conclusive signal $s'_{1}$.
Given this, agent 2 can infer agent 1's private signal from her action.
Thus, agent $2$ chooses action $b_{1}$ if and only if agent $1$ chose action $a_{1}$ or agent 2's private signal is $s'_{1}$.
The expected payoffs are given by $V_{1}(\pi',\bm{\sigma}')=\mu_{0}(1-\varepsilon)(1-r_{1})$ and $V_{2}(\pi',\bm{\sigma}')=\mu_{0}(1-\varepsilon^{2})(1-r_{2})$.

Therefore, $V_{1}(\pi,\bm{\sigma})>V_{1}(\pi',\bm{\sigma}')$ but $V_{2}(\pi,\bm{\sigma})<V_{2}(\pi',\bm{\sigma}')$ even though $\supp{\mu}=\{0,\mu_{0},1\}$ and $\pi\succsim_{B} \pi'$.
This indicates that Theorem \ref{thm_sufficient_condition} does not hold under the heterogeneous decision problems. \qed
\end{example}

\section{Proof of Theorem \ref{prop: binary binary}} \label{sec: proof binary}
When the equilibrium strategy profile need not be specified explicitly, we sometimes write $V_i^{\mathcal{D}}(\pi)$ in place of $V_i^{\mathcal{D}}(\pi,\bm{\sigma})$.

We first formalize the notion of equilibrium uniqueness without tie-breaking.
\begin{proposition}\label{prop: pure strategy equilibrium}
    Take $\pi\in \Pi^{B}$ and $\mathcal{D}\in \mathscr{D}^{B}$. Let $\bm{\sigma}^{*}$ be a pure-strategy equilibrium of $(\mathcal{D},\pi)$. Then, under $(\mathcal{D},\pi,\bm{\sigma}^{*})$, each agent's posterior belief after observing her history and private signal belongs to the set
    \begin{align*}
\left\{\frac{\mu_0[\pi(s_{h}|H)]^n[\pi(s_{l}|H)]^m}{\mu_0[\pi(s_{h}|H)]^n[\pi(s_{l}|H)]^m+(1-\mu_0)[\pi(s_{h}|L)]^n[\pi(s_{l}|L)]^m}\,\middle|\,
\begin{array}{c}
(n,m)\in\mathbb{Z}_+^2\\
1\leq n+m\le N
\end{array}
\right\}.
\end{align*}
\end{proposition}
\begin{proof}[Proof of Proposition \ref{prop: pure strategy equilibrium}]
    Since $\bm{\sigma}^{*}$ is a pure strategy equilibrium and $\pi(s_{h}|H)\geq \pi(s_{h}|L)$, for every agent $i$ and every history of observed actions $\bm{a}\in A^{i-1}$, the equilibrium action rule of agent $i$ must be one of the following:
    (i) choose $a_1$ after observing $s_{h}$ and $a_0$ after observing $s_{l}$, (ii) choose $a_0$ regardless of the realized signal, or (iii) choose $a_1$ regardless of the realized signal.
    
    In cases (ii) and (iii), agent $i$'s action is independent of her private signal and thus provides no information to subsequent agents. In case (i), agent $i$'s action perfectly reveals her private signal, so every subsequent agent can infer that signal from the observed action.
    Hence, for every agent $i$, her posterior belief coincides with the posterior that would arise from observing $k$ private signals for some $k\leq i$. Consequently, the posterior must be an element of the set above. 
\end{proof}

To establish the theorem, we first reduce the class of binary decision problems that need to be considered. Since agents' behavior depends only on the cutoff belief
\[
\frac{u(a_0,L)-u(a_1,L)}
{u(a_0,L)-u(a_1,L)+u(a_1,H)-u(a_0,H)},
\]
it is without loss of generality to focus on the family of decision problems $(\mathcal{D}_r)_{r\in(0,1)}\subset\mathscr{D}^{B}$ defined as follows:
for each $r\in (0,1),$ let $\mathcal{D}_r=(A,u_r)$ be the following decision problem: $u_r(a_0,H)=u_r(a_0,L)=0$, $u_r(a_1,H)=1-r$, and $u_r(a_1,L)=-r$.
\begin{lemma}\label{lem: simple decision problem}
  Let $\pi,\pi'\in \Pi^{B}$. 
  Take any $\mathcal{D}\in \mathscr{D}^{B}$ which induces no tie-break under $\pi$ and $\pi'$. 
  Then, for each $i$, $V_{i}^{\mathcal{D}}(\pi)\geq V_{i}^{\mathcal{D}}(\pi')$ if $V_i^{\mathcal{D}_{r}}(\pi)\geq V_i^{\mathcal{D}_{r}}(\pi')$, where $r=\frac{u(a_0,L)-u(a_1,L)}{u(a_0,L)-u(a_1,L)+u(a_1,H)-u(a_0,H)}$.
\end{lemma}
\begin{proof}[Proof of Lemma~\ref{lem: simple decision problem}]
Since the equilibrium is essentially unique under both $\pi$ and $\pi'$ in $\mathcal{D}$, and since each agent's optimal action as a function of her posterior belief is identical in $\mathcal{D}$ and $\mathcal{D}_r$, the equilibrium strategy profile under $\pi$ in $\mathcal{D}$ must coincide with the equilibrium strategy profile under $\pi$ in $\mathcal{D}_r$. Likewise, the equilibrium strategy profile under $\pi'$ in $\mathcal{D}$ must coincide with the equilibrium strategy profile under $\pi'$ in $\mathcal{D}_r$. Therefore, for every agent $i$, her equilibrium strategy is the same in $\mathcal{D}$ and $\mathcal{D}_r$ under both information structures.

Fix an arbitrary agent $i$. Let $p$ and $q$ be the equilibrium probabilities that agent $i$ chooses action $a_1$ under $\pi$ when the true state is $H$ and $L$, respectively. Likewise, let $p'$ and $q'$ be the corresponding equilibrium probabilities under $\pi'$.
Then, we have
\begin{align*}
    & V_i^{\mathcal{D}_{r}}(\pi)=\mu_0p(1-r)-(1-\mu_0)qr, \\
    & V_i^{\mathcal{D}_{r}}(\pi')=\mu_0p'(1-r)-(1-\mu_0)q'r, \\
    & V_{i}^{\mathcal{D}}(\pi)=\mu_0[pu(a_1,H)+(1-p)u(a_0,H)]+(1-\mu_0)[qu(a_1,L)+(1-q)u(a_0,L)], \\
    & V_{i}^{\mathcal{D}}(\pi')=\mu_0[p'u(a_1,H)+(1-p')u(a_0,H)]+(1-\mu_0)[q'u(a_1,L)+(1-q')u(a_0,L)].
\end{align*}
Since $V_{i}^{\mathcal{D}_r}(\pi)\geq V_{i}^{\mathcal{D}_r}(\pi')$, we have
\begin{align*}
V_i^{\mathcal{D}_{r}}(\pi)-V_i^{\mathcal{D}_{r}}(\pi') = \mu_0(1-r)(p-p') - (1-\mu_0)r(q-q')\geq 0.
\end{align*}
By the definition of $r$, we have
\begin{align*}
(1-r)\bigl[u(a_0,L)-u(a_1,L)\bigr]=r\bigl[u(a_1,H)-u(a_0,H)\bigr].
\end{align*}
Hence,
\begin{align*}
V_i^{\mathcal D}(\pi)-V_i^{\mathcal D}(\pi')&=\mu_0(p-p')\bigl[u(a_1,H)-u(a_0,H)\bigr] \\
&\quad\quad+(1-\mu_0)(q-q')\bigl[u(a_1,L)-u(a_0,L)\bigr]\\
&= \bigl[u(a_0,L)-u(a_1,L)\bigr]
\left(
\mu_0\frac{1-r}{r}(p-p')
-(1-\mu_0)(q-q')
\right) \\
&\geq 0
\end{align*}
\end{proof}
Theorem \ref{prop: binary binary} is easy to prove when $\pi$ contains a conclusive signal, that is, when $\pi(s_{h}|H)=1$ or $\pi(s_{h}|L)=0$, or when $\pi'$ is completely uninformative, that is, when $\pi'(s_{h}|H)=\pi'(s_{h}|L)$. Thus, it suffices to consider the remaining cases.

In the following, the proof proceeds as follows. 
We first show that if $\pi$ and $\pi'$ are sufficiently "close", then the equilibrium strategies under $\pi$ and $\pi'$ coincide on the equilibrium path. For such pairs of information structures, we prove that $\pi \succsim_B \pi'$ implies $V_i^{\mathcal{D}_{r}}(\pi)\geq V_i^{\mathcal{D}_{r}}(\pi')$ for every agent $i$ (Proposition \ref{prop: same r region}).

Next, for an arbitrary pair $\pi \succsim_B \pi'$, we construct a path from $\pi$ to $\pi'$ along which information becomes gradually less informative in the sense of Blackwell. 
By Proposition~\ref{prop: same r region}, agents' payoffs decrease as long as the on-path equilibrium strategy remains unchanged. 
Therefore, it remains to show that payoffs do not jump at points where the equilibrium strategy changes, or at points where ties arise. 
Proposition \ref{prop: continuity} establishes precisely this continuity property.
Combining them, it follows that agents' payoffs respect the Blackwell order in all cases.

To analyze the case in which $\pi$ and $\pi'$ are sufficiently "close", we first define the notion of the same \(r\)-region associated with each decision problem \(\mathcal{D}_{r}\). 
For each $r\in (0,1)$, let $K_r=\frac{1-\mu_{0}}{\mu_{0}}\cdot \frac{r}{1-r}$. 
Fix $i$ and agent $i$'s information $(\bm{a},s)\in A^{i-1}\times S$. 
Then, agent $i$ takes $a_0$ under $(\mathcal{D}_{r},\pi)$ if  $\frac{\mathbb{P}((\bm{a},s)|H)}{\mathbb{P}((\bm{a},s)|L)}\leq K_r$ and takes $a_1$ if $\frac{\mathbb{P}((\bm{a},s)|H)}{\mathbb{P}((\bm{a},s)|L)}\geq K_r$ at the equilibrium.
In addition, if we only focus on pure strategy equilibrium, we have
\begin{align*}
    \frac{\mathbb{P}((\bm{a},s)|H)}{\mathbb{P}((\bm{a},s)|L)}\in \{x^ny^m\mid (n,m)\in\mathbb{Z}_+^2, 1\leq n+m\leq N\},
\end{align*}
where $x=\frac{\pi(s_{h}|H)}{\pi(s_{h}|L)}$ and $y=\frac{\pi(s_{l}|H)}{\pi(s_{l}|L)}$ by the same argument as Proposition \ref{prop: pure strategy equilibrium}.

Consider the subset $X=\{(x,y)\in\mathbb{R}^2 \mid x>1,\ 0<y<1\}$ of $\mathbb{R}^{2}$. 
For every point $(x,y)\in X$, there exists a unique pair of signal probabilities satisfying $
x=\frac{\pi(s_{h}|H)}{\pi(s_{h}|L)}$ and $y=\frac{1-\pi(s_{h}|H)}{1-\pi(s_{h}|L)}$, with $0<\pi(s_{h}|L)<\pi(s_{h}|H)<1$.
Hence, the set of full-support and not no information in $\Pi^{B}$ (hereafter denoted by $\overline{\Pi}^{B}$) can be identified one-to-one with the points in $X$. 
Define $\varphi:\overline{\Pi}^{B}\to X$ as such mapping, that is, $\varphi(\pi)=(\frac{\pi(s_{h}|H)}{\pi(s_{h}|L)},\frac{1-\pi(s_{h}|H)}{1-\pi(s_{h}|L)})$. 
Then, we have $\varphi^{-1}(x,y)=\pi^{*}$, where $\pi^{*}(s_{h}|H)=\frac{x(1-y)}{x-y}$ and $\pi^{*}(h|L)=\frac{1-y}{x-y}$.

Moreover, for any $\pi,\pi'\in\overline{\Pi}^{B}$, let $(x,y)=\varphi(\pi)$ and $(x',y')=\varphi(\pi')$. 
Then, $\pi \succsim_B \pi'$ if and only if $x\geq x'$ and $y\leq y'$.

Identify each information structure $\pi\in\overline{\Pi}^{B}$ with the pair $(\pi(s_{h}|H),\pi(s_{h}|L))$. 
Then, we endow $\overline{\Pi}^{B}$ with the subspace topology inherited from $\mathbb{R}^2$. 
Note that the map $\varphi:\overline{\Pi}^{B}\to X$ is a diffeomorphism.

Now fix $r\in(0,1)$.
Partition $X$ into finitely many regions by the curves $x^n y^m = K_r$, where $(n,m)\in\mathbb{Z}_+^2$ and $1\le n+m\le N$.
Let $\operatorname{sgn}:\mathbb{R}\to\{-1,0,1\}$ denote the sign function, defined by $\operatorname{sgn}(z)=-1$ for $z<0$, $\operatorname{sgn}(z)=0$ for $z=0$, and $\operatorname{sgn}(z)=1$ for $z>0$.
\begin{definition}
Fix $r\in(0,1)$.
We say that $(x,y)\in X$ and $(x',y')\in X$ belong to the {\it same $r$-region} if (i) $\operatorname{sgn}(x^ny^m-K_r) \neq 0$ and (ii) $\operatorname{sgn}(x^ny^m-K_r) =\operatorname{sgn}((x')^n(y')^m-K_r)$ for every $(n,m)\in\mathbb Z_+^2$ with $1\le n+m\le N$.
Also, we say that $\pi\in \overline{\Pi}^{B}$ and $\pi'\in \overline{\Pi}^{B}$ belong to the same $r$-region if $\varphi(\pi)$ and $\varphi(\pi')$ belong to the same $r$-region.
\end{definition}
Note that if $\pi\in \overline{\Pi}^{B}$ and $\pi'\in \overline{\Pi}^{B}$ belong to the same $r$-region, $\mathcal{D}_r$ induces no tie-break under $\pi$ and $\pi'$. 
In addition, as discussed above, in the absence of tie-breaking, the action chosen by each agent is determined solely by whether $x^n y^m<K_r$ or $x^n y^m>K_r$. 
Therefore, any two information structures whose corresponding points on $X$ lie in the interior of the same region induce the same equilibrium strategy profile on the path.

We now establish monotonicity and continuity of payoffs within the same $r$-region.
\begin{proposition}\label{prop: same r region}
    Fix $r\in(0,1)$.
Suppose that $\pi\in \overline{\Pi}^{B}$ and $\pi'\in \overline{\Pi}^{B}$ belong to the same $r$-region.
 If $\pi \succsim_B \pi'$, it follows that $V_i^{\mathcal{D}_{r}}(\pi)\geq V_i^{\mathcal{D}_{r}}(\pi')$. In addition, $V_i^{\mathcal{D}_{r}}(\pi)$ is continuous on each $r$-region.
\end{proposition}
At first glance, Proposition \ref{prop: same r region} may appear straightforward. Indeed, within the same $r$-region, equilibrium strategies remain unchanged. Nevertheless, this observation alone is insufficient. Even when the equilibrium strategy profile is fixed, the resulting history may fail to be Blackwell ordered across information structures. Consequently, the monotonicity of agents' payoffs cannot be established directly and requires a more delicate argument.

To prove Proposition \ref{prop: same r region}, we first show that if agent $i$ follows her private signal, the signal profile $(s_1,\dots,s_{i-1})$ is uniquely determined.

\begin{lemma}\label{lem: unique history under no cascade}
    Take $\pi\in \overline{\Pi}^{B}$ and $r\in (0,1)$. Suppose that $y<K_r<x$. Let $\bm{\sigma}^*$ be a pure-strategy equilibrium of $(\mathcal{D}_r,\pi)$ such that, whenever a tie occurs, all agents choose the same action.  In this case, for each $i\geq2$, there exists a unique signal profile $(s_1,\ldots,s_{i-1})\in S^{i-1}$
such that, after the action history induced by $(s_1,\ldots,s_{i-1})$, agent $i$ follows her private signal; that is, agent $i$ chooses $a_1$ upon receiving $s_h$ and chooses $a_0$ upon receiving $s_l$.  
\end{lemma}
\begin{proof}[Proof of Lemma \ref{lem: unique history under no cascade}]
Without loss of generality, we assume that each agent takes $a_0$ whenever a tie occurs.
We prove this claim by mathematical induction. First, consider $i=2$. Note that agent 2 can perfectly infer agent 1's private signal from agent 1's action. Indeed, since $y<K_r<x$,
agent 1 strictly prefers to follow her private signal: she chooses $a_1$ upon receiving $s_h$ and $a_0$ upon receiving $s_l$. Then, if $xy\leq K_r$, agent 2 always chooses $a_0$ if $s_1=s_l$, and agent 2 follows her private signal if $s_1=s_h$ because $x^2>x>K_r$. If $xy>K_r$, agent 2 always chooses $a_1$ if $s_1=s_h$, and agent 2 follows her private signal if  $s_1=s_l$ because $y^2<y<K_r$.

Next, suppose that the statement holds at $i=j$. Consider the case of $i=j+1$. If agent $j$ chooses an action independently of her private signal, then her action is uninformative. As a result, agent $j+1$ faces exactly the same situation as $j$, and therefore chooses the same action as $j$ regardless of her own private signal. Hence, whenever agent $j+1$ follows her private signal, agent $j$ must also follow hers. In this case, by the induction hypothesis, the signal profile $(s_1,\ldots,s_{j-1})$
is uniquely determined.  Thus, when this signal profile is realized, agent $j$ and all subsequent agents can infer that the signals received by agents $1,\ldots,j-1$ were $(s_1,\ldots,s_{j-1})$.

Take this signal profile and let
\[
n_j=\left|\{k\in \{1,\dots,j-1\}\mid s_k=s_h\}\right| \quad \text{ and } \quad m_j=\left|\{k\in \{1,\dots,j-1\}\mid s_k=s_l\}\right|.
\]
Note that $n_j+m_j=j-1$. If $x^{n_j+1}y^{m_j+1}\leq K_r$, agent $j+1$ always chooses $a_0$ if $s_j=s_l$ because $x^{n_j}y^{m_j+2}<x^{n_j+1}y^{m_j+1}\leq  K_r$, and agent $j+1$ follows his private signal if $s_j=s_h$ because $x^{n_j+1}y^{m_j+1}\leq  K_r$ and $x^{n_j+2}y^{m_j}>x^{n_j+1}y^{m_j}>  K_r$, where the last inequality holds by the fact that agent $j$ chooses $a_1$ if she receives $s_h$ after $(s_1,\dots,s_{j-1})$. Similarly, if $x^{n_j+1}y^{m_j+1}> K_r$, agent $j+1$ always chooses $a_1$ if $s_j=s_h$ because $x^{n_j+2}y^{m_j}>x^{n_j+1}y^{m_j+1}>  K_r$, and agent $j+1$ follows her private signal if $s_j=s_l$ because $x^{n_j+1}y^{m_j+1}>  K_r$ and $x^{n_j}y^{m_j+2}<x^{n_j}y^{m_j+1}\leq   K_r$, where the last inequality holds by the fact that agent $j$ chooses $a_0$ if she receives $s_l$ after $(s_1,\dots,s_{j-1})$.
In either case, there exists a unique signal profile under which agent $j+1$ follows her private signal. This completes the induction.
\end{proof}

Next, we establish two mathematical properties. Lemma \ref{lem: increasing with respect to x} shows that, within a given $r$-region, each agent's payoff is weakly increasing in $x=\frac{\pi(s_h| H)}{\pi(s_h| L)}.$
Lemma \ref{lem: decreasing with respect to y}, on the other hand, shows that, within the same $r$-region, each agent's payoff is weakly decreasing in $y=\frac{\pi(s_l\mid H)}{\pi(s_l\mid L)}.$
\begin{lemma}\label{lem: increasing with respect to x}
Take $i\in \mathbb{N}$. Let $\bm{b}$ be a sequence with length $i$. Suppose $b_j\in \{0,1\}$ for all  $j=1,\dots,i-1$ and $b_i=1$.  Assume that $p,q\in (0,1)$ satisfies $0<q<p<1.$
Define $c_j=\left|\{l\in \mathbb{N}\mid l<j,\,b_l=1\}\right|$. Suppose that $K_r\in \mathbb{R}$ satisfies
\begin{align*}
\begin{cases}
        p^{j-c_j}(1-p)^{c_j+1}<K_r  q^{j-c_j}(1-q)^{c_j+1} & \text{ if } b_j=0\\
        p^{j-c_j}(1-p)^{c_j+1}>K_r  q^{j-c_j}(1-q)^{c_j+1} & \text{ if } b_j=1\\
        1-p<K_r(1-q)\\
        p>K_rq
\end{cases}
\end{align*}
for all $j=1,\dots,i-1$.
Let $f(p)=\sum_{j=1}^ib_j p^{j-c_j}(1-p)^{c_j}$. Then, $(1-p)f'(p)<K_r(1-q)f'(q)$.
\end{lemma}
\begin{proof}[Proof of Lemma \ref{lem: increasing with respect to x}]
The proof is by induction. First, consider the case of $i=1$. Then, $b_1=1$, $c_1=0$, and $f(p)=p$. Thus, $(1-p)f'(p)=1-p<K_r(1-q)=K_r(1-q)f'(q)$.

Next, fix $n\in\mathbb{N}$ with $n\geq 2$ and suppose that the statement holds for all $i=1,\dots,n-1$. Consider the case $i=n$. Take an arbitrary sequence $\bm{b}$ of length $n$ such that $b_j\in\{0,1\}$ for all $j=1,\ldots,n-1$ and $b_n=1$.
Let $a$ be the largest integer satisfying $b_n=b_{n-1}=\cdots=b_{n-a}=1.$ By definition, we have $0\leq a\leq n-1$. 
Note that, by the definition of $c_{j}$, we have $j-c_j=n-c_n$ for all $j=n,n-1,\dots,n-a$. Then, it follows that
\begin{align*}
    f(p)&=\sum_{j=1}^{n-1-a}b_jp^{j-c_j}(1-p)^{c_j}+\sum_{j=n-a}^n b_jp^{j-c_j}(1-p)^{c_j}\\
    &=\sum_{j=1}^{n-1-a}b_jp^{j-c_j}(1-p)^{c_j}+\sum_{j=n-a}^n p^{n-c_n}(1-p)^{c_n-n+j}\\
    &=\sum_{j=1}^{n-1-a}b_jp^{j-c_j}(1-p)^{c_j}+ \left(p^{n-c_n-1}(1-p)^{c_n-a}[1-(1-p)^{a+1}]\right).
\end{align*}
Let $h(p)=\sum_{j=1}^{n-1-a}b_jp^{j-c_j}(1-p)^{c_j}$.\footnote{When $a=n-1$, we define $h(p)=0$.} Taking the derivative with respect to $p$ and multiplying by $1-p$ imply that $(1-p)f'(p)$ can be written as 
\begin{align*}
    & (1-p)h'(p) +p^{n-c_n-2}(1-p)^{c_n-a} \Bigl( \bigl[(n-c_n-1)(1-p)-(c_n-a)p\bigr] \bigl[1-(1-p)^{a+1}\bigr] \\
    &\hspace{8cm} +(a+1)p(1-p)^{a+1} \Bigr)\\
    &=(1-p)h'(p) \\
    &\quad +p^{n-c_n-1}(1-p)^{c_n-a} \Bigl(-(c_n-a)+(n-c_n-1)\sum_{j=1}^{a}(1-p)^j+n(1-p)^{a+1}\Bigr).
\end{align*}
Since $n-1-a\leq n-1$, the induction hypothesis implies that $h(p)=0$ or $(1-p)h'(p)<K_r(1-q)h'(q)$. We will later show that the following inequalities hold. 
(i): $p^{n-c_n-1}(1-p)^{c_n-a}>K_rq^{n-c_n-1}(1-q)^{c_n-a}$ if $c_n\neq a$ and (ii): $p^{n-c_n-1}(1-p)^{c_n-a+j}<K_rq^{n-c_n-1}(1-q)^{c_n-a+j}$ for all $j\geq 1$.
Assuming for the moment that these inequalities hold, and noting that $-(c_n-a)=0$ if $c_n=a$ and $-(c_n-a)<0$ if $c_n\neq a$, $n-c_n-1\geq 0$, and $n>0$, we obtain
\[
    (1-p)f'(p)<K_r(1-q)f'(q).
\]

Therefore, it remains to verify inequalities (i) and (ii). We begin with (i). Since $c_n\neq a$, we have $a<n-1$ and $c_{n-a-1}=c_{n-a}=c_n-a\neq 0$.\footnote{Note that $b_{n-a-1}=0$ holds by the definition of $a$. Thus, $c_{n-a-1}=c_{n-a}.$} 
Hence, there exists at least one entry equal to $1$ among $b_{n-a-2}, b_{n-a-3},\dots.$
Let $b_m$ be the first such entry when moving backward from $n-a-2$. 
By assumption,
\[
p^{m-c_m}(1-p)^{c_m+1} > K_rq^{m-c_m}(1-q)^{c_m+1}.
\]
Since $m$ is the last index prior to $n-a-1$ at which $b_m=1$, we have $c_m=c_{n-a-1}-1=c_n-a-1$.
Hence,
\[
p^{m-c_n+a+1}(1-p)^{c_n-a} > K_rq^{m-c_n+a+1}(1-q)^{c_n-a}. 
\]
$1-p<K_r(1-q)$ and $p>K_r q$ implies $p>q$.
Since  $n-m-a-2\geq n-(n-a-2)-a-2=0$ and $p>q$, it follows that
\begin{align*}
    p^{n-c_n-1}(1-p)^{c_n-a}
    & = p^{n-m-a-2}p^{m-c_n+a+1}(1-p)^{c_n-a}\\
    & \geq q^{n-m-a-2}p^{m-c_n+a+1}(1-p)^{c_n-a}\\
    & > K_r q^{n-m-a-2}q^{m-c_n+a+1}(1-q)^{c_n-a} = K_r q^{n-c_n-1}(1-q)^{c_n-a}.
\end{align*}
This establishes (i).

Next, show (ii). Take $j\geq 1$. By assumption, $b_{n-a-1}=0$ implies\footnote{If $a=n-1$, this inequality is equivalent to $1-p<K_r(1-q)$, which is already assumed. Hence, we can ignore this case.}
\[
p^{(n-a-1)-c_{n-a-1}}(1-p)^{c_{n-a-1}+1}<K_rq^{(n-a-1)-c_{n-a-1}}(1-q)^{c_{n-a-1}+1}.
\]
Since $c_{n-a-1}=c_n-a$, we have
\[
p^{n-c_n-1}(1-p)^{c_{n}-a+1}<K_rq^{n-c_n-1}(1-q)^{c_{n}-a+1}.
\]
Note that $1-p<1-q$, then, it follows that 
\begin{align*}
    p^{n-c_n-1}(1-p)^{c_n-a+j}
    &= p^{n-c_n-1}(1-p)^{c_n-a+1}(1-p)^{j-1}\\
    &\leq  p^{n-c_n-1}(1-p)^{c_n-a+1}(1-q)^{j-1}\\
    &<K_rq^{n-c_n-1}(1-q)^{c_n-a+1}(1-q)^{j-1} = K_r q^{n-c_n-1}(1-q)^{c_n-a+j}.
\end{align*}
Therefore, (ii) holds.

Hence, $(1-p)f'(p)<K_r(1-q)f'(q),$ as desired.
\end{proof}
\begin{lemma}\label{lem: decreasing with respect to y}
Take $i\in \mathbb{N}$. Let $\bm{b}$ be a sequence with length $i$. Suppose $b_j\in \{0,1\}$ for all  $j=1,\dots,i-1$ and $b_i=1$.  Assume that $p,q\in (0,1)$ satisfies $0<q<p<1.$
Define $c_j=\left|\{l\in \mathbb{N}\mid l<j,\,b_l=1\}\right|$. Suppose that $K_r\in \mathbb{R}_+$ satisfies
\[
\begin{cases}
        p^{c_j+1}(1-p)^{j-c_j}>K_r  q^{c_j+1}(1-q)^{j-c_j} & \text{ if } b_j=0\\
       p^{c_j+1}(1-p)^{j-c_j}<K_r  q^{c_j+1}(1-q)^{j-c_j} & \text{ if } b_j=1\\
        1-p<K_r(1-q)\\
        p>K_rq
\end{cases}
\]
for all $j=1,\dots,i-1$.
Let $g(p)=\sum_{j=1}^ib_j p^{c_j}(1-p)^{j-c_j}$. Then, $pg'(p)<K_rqg'(q)$.
\end{lemma}
\begin{proof}[Proof of Lemma \ref{lem: decreasing with respect to y}]
Prove is by induction. 
First, consider the case of $i=1$. Then, $b_1=1$, $c_1=0$, and $g(p)=1-p$. Thus, $pg'(p)=-p<-K_rq=K_rqg'(q)$.

Next, fix $n\in\mathbb{N},n\geq 2$ and suppose that the statement holds for all $i=1,\dots,n-1$.
Consider the case $i=n$.
Take an arbitrary sequence $\bm{b}$ of length $n$ such that $b_j\in\{0,1\}$ for all $j=1,\ldots,n-1$ and $b_n=1$.
Let $a$ be defined as in Lemma \ref{lem: increasing with respect to x}. Then, it follows that
\begin{align*}
    g(p)&=\sum_{j=1}^{n-1-a}b_jp^{c_j}(1-p)^{j-c_j}+ p^{c_n-a}(1-p)^{n-c_n-1}[1-p^{a+1}].
\end{align*}
Let $\overline{h}(p)=\sum_{j=1}^{n-1-a}b_jp^{c_j}(1-p)^{j-c_j}$.\footnote{When $a=n-1$, we define $\overline{h}(p)=0$.} Taking the derivative with respect to $p$ and multiplying by $p$ yields
\[
    pg'(p) = p\overline{h}'(p) + p^{c_n-a}(1-p)^{n-c_n-1} \Bigl( (c_n-a)-(n-c_n-1)\sum_{j=1}^{a}p^j-np^{a+1} \Bigr).
\]
Since $n-1-a\leq n-1$, the induction hypothesis implies that $\overline{h}(p)=0$ or $p\overline{h}'(p)<K_rq\overline{h}'(q)$. We will later show that the following inequalities hold. 
(i): $p^{c_n-a}(1-p)^{n-c_n-1}<K_rq^{c_n-a}(1-q)^{n-c_n-1}$ if $c_n\neq a$ and (ii): $p^{c_n-a+j}(1-p)^{n-c_n-1}>K_rq^{c_n-a+j}(1-q)^{n-c_n-1}$ for all $j\geq 1$.
Assuming for the moment that these inequalities hold, and noting that $c_n-a>0$ if $c_n\neq a$, $-(n-c_n-1)\leq 0$, and $-n<0,$ we obtain 
\begin{align*}
    pg'(p)<K_rqg'(q).
\end{align*}

Therefore, it remains to verify inequalities (i) and (ii). We begin with (i). Since $c_n\neq a$, we can take $b_m$ in the same way as Lemma \ref{lem: increasing with respect to x}. 
By assumption, $b_m=1$ implies
\begin{align*}
p^{c_m+1}(1-p)^{m-c_m}
<
K_rq^{c_m+1}(1-q)^{m-c_m}.
\end{align*}
Since  $c_m=c_{n-a-1}-1=c_n-a-1$,
\begin{align*}
p^{c_n-a}(1-p)^{m-c_n+a+1}
<
K_rq^{c_n-a}(1-q)^{m-c_n+a+1}. 
\end{align*}
Since  $n-m-a-2\geq n-(n-a-2)-a-2=0$ and $0<1-p<1-q$, it follows that
\begin{align*}
    p^{c_n-a}(1-p)^{n-c_n-1}&=p^{c_n-a}(1-p)^{m-c_n+a+1}(1-p)^{n-m-a-2}\\
    &\leq p^{c_n-a}(1-p)^{m-c_n+a+1}(1-q)^{n-m-a-2}\\
    &<K_r  q^{c_n-a}(1-q)^{m-c_n+a+1}(1-q)^{n-m-a-2} = K_r q^{c_n-a}(1-q)^{n-c_n-1}.
\end{align*}
This establishes (i).

Next, show (ii). Take $j\geq 1$. By assumption, $b_{n-a-1}=0$ implies\footnote{If $a=n-1$, this inequality is equivalent to $p>K_rq$, which is already assumed. Hence, we can ignore this case.}
\begin{align*}
p^{c_{n-a-1}+1}(1-p)^{(n-a-1)-c_{n-a-1}}>K_rq^{c_{n-a-1}+1}(1-q)^{(n-a-1)-c_{n-a-1}}.
\end{align*}
Since $c_{n-a-1}=c_n-a$, we have
\begin{align*}
p^{c_{n}-a+1}(1-p)^{n-c_n-1}>K_rq^{c_{n}-a+1}(1-q)^{n-c_n-1}.
\end{align*}
Note that $p>q$, then, it follows that 
\begin{align*}
    p^{c_n-a+j}(1-p)^{n-c_n-1}&=p^{j-1}p^{c_n-a+1} (1-p)^{n-c_n-1}\\
    &\geq  q^{j-1}p^{c_n-a+1} (1-p)^{n-c_n-1}\\
    &>K_rq^{j-1}q^{c_n-a+1} (1-q)^{n-c_n-1} = K_r  q^{c_n-a+j}(1-q)^{n-c_n-1}.
\end{align*}
Therefore, (ii) holds.

Hence, $pg'(p)<K_rqg'(q),$ as desired.
\end{proof}
\begin{proof}[Proof of Proposition \ref{prop: same r region}]
First, show the continuity of $V_i^{\mathcal{D}_{r}}(\pi)$.
Suppose that $\pi \in \overline{\Pi}^{B}$ and $\pi' \in \overline{\Pi}^{B}$ belong to the same $r$-region.
  Let $\bm{\sigma}^{*}$ be an equilibrium under $(\mathcal{D}_{r},\pi)$. Since equilibrium is essentially unique, every signal profile induces a unique action profile. Let $\tau^{*}=(\tau_1^{*},\ldots,\tau_N^{*})$
denote the induced mapping from signal histories to actions. In particular, 
$\tau_i^{*}:S^i\to A $
gives the action chosen by agent $i$ as a function of the realized signal history $(s_1,\ldots,s_i)$. Then, we have
\begin{align*}
    V_i^{\mathcal{D}_{r}}(\pi)=\mu_0(1-r)\sum_{(s_1,\dots,s_i)\in S^i}\prod_{j=1}^i \pi(s_j|H)\chi_{T_i^{*}}-(1-\mu_0)r\sum_{(s_1,\dots,s_i)\in S^i}\prod_{j=1}^i \pi(s_j|L)\chi_{T_i^{*}},
\end{align*}
where $T_i^{*}=\{(s_1,\dots,s_i)\in S^i\mid \tau^{*}_i(s_1,\dots,s_i)=a_1\}$, and $\chi$ is an indicator function. As noted above, the equilibrium strategy profile is unchanged on the equilibrium path within a given $r$-region. Hence, the induced mapping $\tau^{*}$ is constant throughout the region. Since $V_i^{\mathcal{D}_{r}}(\pi)$ is determined by a fixed finite sum whose coefficients depend continuously on $\pi$, it follows that $V_i^{\mathcal{D}_{r}}(\pi)$ is continuous on each $r$-region. In addition, since the mapping $\varphi:(\pi(s_h|H),\pi(s_h|L))\mapsto\left(\frac{\pi(s_h|H)}{\pi(s_h|L)}, \frac{\pi(s_l|H)}{\pi(s_l|L)}\right)$ is a diffeomorphism, $V_i^{\mathcal D_r}(\pi)$ can be viewed as a differentiable function of $(x,y)$.

Next, show that $V_i^{\mathcal D_r}(\pi)\geq V_i^{\mathcal D_r}(\pi')$ if  $\pi \in \overline{\Pi}^{B}$ and $\pi' \in \overline{\Pi}^{B}$ belong to the same $r$-region, and $\pi \succsim_B \pi'$.
Suppose that $x<K_r$. In this case, under both $\pi$ and $\pi'$, every agent always chooses action $a_0$. Therefore, the expected payoff is zero under the two information structures, and the desired inequality holds trivially. Similarly, when $y>K_r$, every agent always chooses action $a_1$, and hence the expected payoff is the same under both $\pi$ and $\pi'$. Therefore, it suffices to consider the case
$y<K_r<x$.

Since $\pi \succsim_B\pi'$ is equivalent to
$x\geq x'$ and $y\leq y'$, it is enough to show that $V_i^{\mathcal D_r}(\pi)$ is non-decreasing with respect to $x$ and non-increasing with respect to $y$. First, show that $V_i^{\mathcal D_r}(\pi)$ is non-decreasing with respect to $x$.
Note that $\mathcal{D}_r$ induces no tie-break under $\pi$. Hence, by Lemma \ref{lem: unique history under no cascade}, there exists a unique signal profile $(s_1,\ldots,s_{i-1})$
under which agent $i$ follows her private signal. Fix this signal profile. Define a binary sequence $(b_1,\ldots,b_i)$ by
\begin{align*}
b_j=
\begin{cases}
0 & \text{ if } s_j=s_h,\\
1 & \text{ if } s_j=s_l,
\end{cases}
\qquad j=1,\ldots,i-1,
\end{align*}
and let $b_i=1$. For each $j$, define $c_j=\left|\{l\in \mathbb{N}\mid l<j,b_l=1\,\}\right|,$ $p=\pi(s_h|H)$, and $q=\pi(s_h|L)$. Note that $1>p>q>0$.
Agent $i$ chooses $a_1$ in either of the following cases.
(I): the realized signal profile $\bm{s}^*$ satisfies
$s_j^*=s_j$ for all  $j=1,\ldots,i-1$, and $s_i^*=s_h$. 
The probability of this event is
$(1-p)^{c_i}p^{i-c_i}$ at the state $H$ and $(1-q)^{c_i}q^{i-c_i}$ at $L$.
(II): there exists a smallest index $j\in\{1,\ldots,i-1\}$ such that
$s_j^*\neq s_j$, and this first deviation is favorable, namely
$s_j^*=s_h$. Equivalently, $b_j=1$. The probability of this event is $b_j(1-p)^{c_j}p^{j-c_j}$ at $H$ and $b_j(1-q)^{c_j}q^{j-c_j}$ at $L$.
Therefore,
\[
\mathbb{P}(a_i=a_1\mid H)=\sum_{j=1}^{i}
b_j(1-p)^{c_j}p^{j-c_j} \,\text{ and }\,
\mathbb{P}(a_i=a_1\mid L)=\sum_{j=1}^{i}
b_j(1-q)^{c_j}q^{j-c_j}.
\]
Let $f(p)=\sum_{j=1}^{i}
b_jp^{j-c_j}(1-p)^{c_j}$. Then, it follows that
\begin{align*}
    V_i^{\mathcal{D}_r}(\pi)=\mu_0(1-r)f(p)-(1-\mu_0)rf(q).
\end{align*}

Next, we verify that the conditions of Lemma~\ref{lem: increasing with respect to x} are satisfied. If $b_1=0$, agent 2 always chooses $a_0$ if agent 1 receives $s_l$. Hence, $xy<K_r$, or $p(1-p)<K_r q(1-q)$. Since $c_1=0$, this is equivalent to $p^{1-c_1}(1-p)^{c_1+1}<K_rq^{1-c_1}(1-q)^{c_1+1}$. If $b_1=1$, agent 2 always chooses $a_1$ if agent 1 receives $s_h$. Hence, $xy>K_r$, or $p(1-p)>K_r q(1-q)$. Since $c_1=0$, this is equivalent to $p^{1-c_1}(1-p)^{c_1+1}>K_rq^{1-c_1}(1-q)^{c_1+1}$.  Take $j\in\{2,\dots,i-1\}$. Suppose that $s_n^*=s_n$ for all $n=1,\dots,j-1$ and $s_j^*\neq s_j$.
In this case, agent $j$ is the first agent at whom a cascade occurs. If $s_j^*=s_h$, then agent $j+1$ chooses $a_1$ even when her private signal is $s_l$. Since among agents $1,\ldots,j-1$, exactly $c_j$ have received $s_l$ and the remaining $j-1-c_j$ have received $s_h$, the posterior likelihood ratio of agent $j+1$ after observing $s_l$ is $x^{\,j-1-c_j}y^{\,c_j}xy,$
which must exceed $K_r$. Hence, $x^{\,j-1-c_j}y^{\,c_j}xy>K_r.$ Equivalently, we have
\begin{align*}
    p^{j-c_j}(1-p)^{c_j+1}>K_rq^{j-c_j}(1-q)^{c_j+1}.
\end{align*}
If $s_j^*=s_l$, the analogous argument yields
\begin{align*}
    p^{j-c_j}(1-p)^{c_j+1}<K_rq^{j-c_j}(1-q)^{c_j+1}.
\end{align*}
Since $y<K_r<x$, we also have $1-p<K_r(1-q)$ and $p>K_rq$.
Hence, all the assumptions of Lemma \ref{lem: increasing with respect to x} are satisfied.

Since $p=\frac{x(1-y)}{x-y}$ and $q=\frac{1-y}{x-y}$, we have $\frac{\partial p}{\partial x}=-\frac{y(1-y)}{(x-y)^2}$ and $\frac{\partial q}{\partial x}=-\frac{1-y}{(x-y)^2}$. Hence,
\begin{align*}
    \frac{\partial }{\partial x}V_i^{\mathcal{D}_r}&=\mu_0(1-r)f'(p)\frac{\partial p}{\partial x}-(1-\mu_0)rf'(q)\frac{\partial q}{\partial x}\\
    &=\frac{1-y}{(x-y)^2}\frac{1}{1-q}\mu_0(1-r)[K_r(1-q)f'(q)-(1-p)f'(p)]\geq 0.
\end{align*}
The last inequality comes from Lemma \ref{lem: increasing with respect to x}.

Next, show that $V_i^{\mathcal D_r}(\pi)$ is non-increasing with respect to $y$. Take the same signal profile $(s_1,\dots,s_{i-1})$. 
Define a binary sequence $(b_1,\ldots,b_i)$ by
\begin{align*}
b_j=
\begin{cases}
0 & \text{ if } s_j=s_l,\\
1 & \text{ if } s_j=s_h,
\end{cases}
\qquad j=1,\ldots,i-1,
\end{align*}
and let $b_i=1$. For each $j$, define $c_j=\left|\{l\in \mathbb{N}\mid l<j,b_l=1\,\}\right|.$
Agent $i$ chooses $a_0$ in either of the following cases.
(I): the realized signal profile $\bm{s}^*$ satisfies
$s_j^*=s_j$ for all  $j=1,\ldots,i-1$, and $s_i^*=s_l$. 
The probability of this event is
$(1-p)^{i-c_i}p^{c_i}$ at the state $H$ and $(1-q)^{i-c_i}q^{c_i}$ at $L$.
(II): there exists a smallest index $j\in\{1,\ldots,i-1\}$ such that
$s_j^*\neq s_j$, and $s_j^*=s_l$. Equivalently, $b_j=1$. The probability of this event is $b_j(1-p)^{j-c_j}p^{c_j}$ at $H$ and $b_j(1-q)^{j-c_j}q^{c_j}$ at $L$.
Therefore,
\[
\mathbb{P}(a_i=a_0\mid H)=\sum_{j=1}^{i}
b_j(1-p)^{j-c_j}p^{c_j} \;\text{ and }\;
\mathbb{P}(a_i=a_0\mid L)=\sum_{j=1}^{i}
b_j(1-q)^{j-c_j}q^{c_j}.
\]
Let $g(p)=\sum_{j=1}^{i}
b_jp^{c_j}(1-p)^{j-c_j}$. Then, it follows that
\begin{align*}
    V_i^{\mathcal{D}_r}(\pi)=\mu_0(1-r)[1-g(p)]-(1-\mu_0)r[1-g(q)].
\end{align*}

A symmetric argument to that in Lemma \ref{lem: increasing with respect to x} shows that all the assumptions of Lemma \ref{lem: decreasing with respect to y} are satisfied. Since $\frac{\partial p}{\partial y}=-\frac{x(x-1)}{(x-y)^2}$ and $\frac{\partial q}{\partial y}=-\frac{x-1}{(x-y)^2}$, we have
\begin{align*}
     \frac{\partial }{\partial y}V_i^{\mathcal{D}_r}&=-\mu_0(1-r)g'(p)\frac{\partial p}{\partial y}+(1-\mu_0)rg'(q)\frac{\partial q}{\partial y}\\
    &=\frac{x-1}{(x-y)^2}\frac{1}{q}\mu_0(1-r)[pg'(p)-K_rq g'(q)]\leq 0.
\end{align*}
The last inequality comes from Lemma \ref{lem: decreasing with respect to y}.

Therefore, the payoff is weakly increasing in $x$ and weakly decreasing in $y$. Hence, the Blackwell order is respected within each $r$-region.
\end{proof}

The continuity property can be summarized in the following statement.
\begin{proposition}\label{prop: continuity}
      Take $\pi\in \overline{\Pi}^{B}$ and $r\in (0,1)$. Suppose that $y<K_r<x$. Let $\bm{\sigma}^0$ and  $\bm{\sigma}^{1}$ be pure-strategy equilibria of $(\mathcal{D}_r,\pi)$ such that, whenever a tie occurs, all agents choose the same action.  Then, $V_i^{\mathcal{D}_r}(\pi,\bm{\sigma}^0)=V_i^{\mathcal{D}_r}(\pi,\bm{\sigma}^{1})$.
\end{proposition}
\begin{proof}[Proof of Proposition \ref{prop: continuity}]
Without loss of generality, we assume that all agents break ties in favor of $a_0$ under $\bm{\sigma}^0$ and in favor of $a_1$ under $\bm{\sigma}^1$. If $\mathcal{D}_r$ induces no tie-break, the statement holds since there is an essentially unique equilibrium. Therefore, we consider the case $\{(n,m)\in \mathbb{Z}_+^2 \mid x^n y^m=K_r, 1\leq n+m\leq i\}\neq \emptyset$. Among the elements of this set, let $(n_1,m_1)$ be one that minimizes $n+m$. Then, the strategies of agents $1,\ldots,n_1+m_1-1$ coincide on the equilibrium path under both $\bm{\sigma}^0$ and $\bm{\sigma}^1$.  If $n_1+m_1=i$, then agent $i$ is indifferent between $a_0$ and $a_1$ at the tie, so her expected payoff is independent of the tie-breaking rule. Therefore, $V_i^{\mathcal{D}_r}(\pi,\bm{\sigma}^0)=V_i^{\mathcal{D}_r}(\pi,\bm{\sigma}^{1})$. Thus, in what follows, we restrict attention to the case $n_1+m_1\leq i-1$. Note that $n_1+m_1\geq 2$ since $y<K_r<x$.

By Lemma \ref{lem: unique history under no cascade}, there exists a unique signal profile $(s_1,\ldots,s_{n_1+m_1-2})$ under which agent $n_1+m_1-1$ follows her private signal. 
We first show that among these $n_1+m_1-2$ signals, exactly $n_1-1$ are equal to $s_h$ and exactly $m_1-1$ are equal to $s_l$. That is, $\left|\{\,j<n_1+m_1-1 \mid s_j=s_h\,\}\right|=n_1-1$ and $\left|\{\,j<n_1+m_1-1 \mid s_j=s_l\,\}\right|=m_1-1.$ Suppose that $s_h$ appears at least $n_1$ times. Then, when agent $n_1+m_1-1$ receives $s_l$, her likelihood ratio is at least $x^{n_1}y^{m_1-1}$. Since $x^{n_1}y^{m_1}=K_r$ and $y<1$, we have $x^{n_1}y^{m_1-1}=\frac{K_r}{y}>K_r.$
Hence, agent $n_1+m_1-1$ would choose $a_1$ regardless of her private signal, contradicting the fact that she follows her private signal. Therefore, $s_h$ appears fewer than $n_1$ times. By a symmetric argument, it is also impossible for $s_l$ to appear at least $m_1$ times.

Under $\bm{\sigma}^0$, if the signal profile $(s_1,\ldots,s_{n_1+m_1-2})$ is realized and agent $n_1+m_1-1$ receives $s_l$, then even if agent $n_1+m_1$ receives $s_h$, the likelihood ratio is $x^{n_1-1}y^{m_1-1}yx=x^{n_1}y^{m_1}=K_r.$
Since ties are broken in favor of $a_0$ under $\bm{\sigma}^0$, all agents from agent $n_1+m_1$ onward choose $a_0$. On the other hand, if the signal profile $(s_1,\ldots,s_{n_1+m_1-2})$ is realized and both agents $n_1+m_1-1$ and $n_1+m_1$ receive $s_h$, then even if agent $n_1+m_1+1$ receives $s_l$, the likelihood ratio is $x^{n_1+1}y^{m_1}>K_r.$
Hence, all agents from agent $n_1+m_1+1$ onward choose $a_1$.
Therefore, under $\bm{\sigma}^0$, agent $i$ chooses $a_1$ if and only if one of the following events occurs.
\begin{enumerate}
\setlength{\itemsep}{0.5pt}
\setlength{\parskip}{0.5pt}
\item[0-(i)] A cascade toward $a_1$ occurs before agent $n_1+m_1-1$.
\item[0-(ii)] The signal profile $(s_1,\ldots,s_{n_1+m_1-2})$ is realized, and both agents $n_1+m_1-1$ and $n_1+m_1$ receive $s_h$.
\item[0-(iii)] The signal profile $(s_1,\ldots,s_{n_1+m_1-2})$ is realized, agent $n_1+m_1-1$ receives $s_h$, agent $n_1+m_1$ receives $s_l$, and thereafter a sequence of signals is realized under which agent $i$ chooses $a_1$.
\end{enumerate}

Under $\bm{\sigma}^1$, if the signal profile $(s_1,\ldots,s_{n_1+m_1-2})$ is realized and agent $n_1+m_1-1$ receives $s_h$, then even if agent $n_1+m_1$ receives $s_l$, the likelihood ratio is $x^{n_1-1}y^{m_1-1}xy=x^{n_1}y^{m_1}=K_r.$
Since ties are broken in favor of $a_1$ under $\bm{\sigma}^1$, all agents from agent $n_1+m_1$ onward choose $a_1$. On the other hand, if the signal profile $(s_1,\ldots,s_{n_1+m_1-2})$ is realized and both agents $n_1+m_1-1$ and $n_1+m_1$ receive $s_l$, then even if agent $n_1+m_1+1$ receives $s_h$, the likelihood ratio is $x^{n_1}y^{m_1+1}<K_r.$
Hence, all agents from agent $n_1+m_1+1$ onward choose $a_0$.
Therefore, under $\bm{\sigma}^1$, agent $i$ chooses $a_1$ if and only if one of the following events occurs.
\begin{enumerate}
\setlength{\itemsep}{0.5pt}
\setlength{\parskip}{0.5pt}
\item[1-(i)] A cascade toward $a_1$ occurs before agent $n_1+m_1-1$.
\item[1-(ii)] The signal profile $(s_1,\ldots,s_{n_1+m_1-2})$ is realized, and agent $n_1+m_1-1$ receives $s_h$.
\item[1-(iii)] The signal profile $(s_1,\ldots,s_{n_1+m_1-2})$ is realized, agent $n_1+m_1-1$ receives $s_l$, agent $n_1+m_1$ receives $s_h$, and thereafter a sequence of signals is realized under which agent $i$ chooses $a_1$.
\end{enumerate}
Since no tie occurs before agent $n_1+m_1$, the probabilities of event 0-(i) and 1-(i) are the same.

Moreover, note that the probability that $(s_1,\ldots,s_{n_1+m_1-2})$ is realized is $p^{n_1-1}(1-p)^{m_1-1}$ under state $H$, and $q^{n_1-1}(1-q)^{m_1-1}$ under state $L$, where $p=\pi(s_h|H)$ and $q=\pi(s_h|L)$.
Therefore, under $\bm{\sigma}^0$, the probability of event 0-(ii) is $p^{n_1+1}(1-p)^{m_1-1}$ conditional on state $H$, and $q^{n_1+1}(1-q)^{m_1-1}$ conditional on state $L$.
Similarly, under $\bm{\sigma}^1$, the probability of event 1-(ii) is $p^{n_1}(1-p)^{m_1-1}$ conditional on state $H$, and $q^{n_1}(1-q)^{m_1-1}$ conditional on state $L$.

Next, consider the probabilities of events 0-(iii) and 1-(iii). In either case, after the realization of $(s_1,\ldots,s_{n_1+m_1-2})$, the next two agents receive $s_h$ once and $s_l$ once. Consequently, by agent $n_1+m_1$, the signal $s_h$ has occurred exactly $n_1$ times, and the signal $s_l$ has occurred exactly $m_1$ times.
Furthermore, all agents from agent $n_1+m_1+1$ onward can infer this signal history from the observed action history. Indeed, under both 0-(iii) and 1-(iii), no cascade occurs before agent $n_1+m_1$.
Since $x^{n_1}y^{m_1}=K_r,$ the posterior likelihood ratio after these observations is exactly equal to the cutoff $K_r$.
Therefore, after conditioning on the history up to agent $n_1+m_1$ described above, the continuation game is equivalent to a social learning process with cutoff $K_r=1$. Hence, the conditional probability that agent $i$ chooses $a_1$ is equal to the probability that agent $i-(n_1+m_1)$ chooses $a_1$ in that process.
Let $\alpha^0$ and $\beta^0$ denote the probabilities that agent $i-(n_1+m_1)$ chooses action $a_1$ under the equilibrium that is consistent with $\bm{\sigma}^0$, that is, breaking ties in favor of $a_0$,
conditional on the state being $H$ and $L$, respectively, in the social learning process with cutoff $K_r=1$.
Similarly, let $\alpha^1$ and $\beta^1$ denote the corresponding probabilities under the equilibrium that is consistent with $\bm{\sigma}^1$.

Then, it follows that, conditional on state $H$, the probability that agent $i$ chooses $a_1$ under $\bm{\sigma}^1$ exceeds the corresponding probability under $\bm{\sigma}^0$ by
\begin{align*}
    & p^{n_1}(1-p)^{m_1-1}-p^{n_1+1}(1-p)^{m_1-1}+p^{n_1}(1-p)^{m_1}\alpha^1-p^{n_1}(1-p)^{m_1}\alpha^0\\
    & = p^{n_1}(1-p)^{m_1}(1+\alpha^1-\alpha^0).
\end{align*}
Similarly, conditional on state $L$, the probability that agent $i$ chooses $a_1$ under $\bm{\sigma}^1$ exceeds the corresponding probability under $\bm{\sigma}^0$ by
\begin{align*}
 q^{n_1}(1-q)^{m_1}(1+\beta^1-\beta^0).
\end{align*}
Hence, we have
\begin{align*}
    &V_i^{\mathcal{D}_r}(\pi,\bm{\sigma}^1)-V_i^{\mathcal{D}_r}(\pi,\bm{\sigma}^0)\\
    &=\mu_0(1-r)p^{n_1}(1-p)^{m_1}[1+\alpha^1-\alpha^0]-(1-\mu_0)rq^{n_1}(1-q)^{m_1}[1+\beta^1-\beta^0]\\
    &=\mu_0(1-r)q^{n_1}(1-q)^{m_1}[x^{n_1}y^{m_1}(1+\alpha^1-\alpha^0)-K_r(1+\beta^1-\beta^0)]\\
    &=\mu_0(1-r)q^{n_1}(1-q)^{m_1}K_r(\alpha^1-\alpha^0-\beta^1+\beta^0).
    \end{align*}
    The last equality holds by $x^{n_1}y^{m_1}=K_r$. 
    Hence, what we need to show is $\alpha^1-\alpha^0=\beta^1-\beta^0$.

Now suppose that $\{(n,m)\in \mathbb{Z}_+^2 : x^n y^m=K_r,\; n_1+m_1<n+m\le i\}=\emptyset.$
Then no tie occurs in the new social learning process starting from agent $n_1+m_1+1$. Hence, $\alpha^0=\alpha^1$ and $\beta^0=\beta^1$, and the desired claim follows immediately.
We therefore consider the case in which $\{(n,m)\in \mathbb{Z}_+^2 : x^n y^m=K_r,\; n_1+m_1<n+m\leq i\}\neq\emptyset.$
Let $(n_2,m_2)$ be an element of this set that minimizes $n+m$.
If $n_2+m_2=i$, then only agent $i-(n_1+m_1)$ in the new social learning problem may face a tie. Since the tie-breaking rule does not affect the payoff of an agent who is indifferent, we have
\begin{align*}
    \mu_0(1-r')\alpha^1-(1-\mu_0)r'\beta^1=\mu_0(1-r')\alpha^0-(1-\mu_0)r'\beta^0,
\end{align*}
where $r'\in \mathbb{R}$ is chosen so that $\frac{(1-\mu_0)r'}{\mu_0(1-r')}=1.$
Then, it follows that $\alpha^1-\alpha^0=\beta^1-\beta^0.$

If $n_2+m_2<i$, applying the same argument as above to the new social learning process, we obtain
\begin{align*}
  &  \alpha^1-\alpha^0=p^{\,n_2-n_1}(1-p)^{m_2-m_1}
\bigl(1+\alpha^{\prime 1}-\alpha^{\prime 0}\bigr)\\
&\beta^1-\beta^0=q^{\,n_2-n_1}(1-q)^{m_2-m_1}
\bigl(1+\beta^{\prime 1}-\beta^{\prime 0}\bigr),
\end{align*}
where $\alpha^{\prime 0}$ and $\beta^{\prime 0}$ denote the probabilities that agent $i-(n_2+m_2)$ chooses action $a_1$ under the equilibrium that is consistent with $\bm{\sigma}^0$
conditional on the state is $H$ and $L$, respectively, in the social learning process with cutoff $K_r=1$, and $\alpha^{\prime 1}$ and $\beta^{\prime 1}$ are the corresponding probabilities under the equilibrium that is consistent with $\bm{\sigma}^1$.
Since $x^{n_2-n_1}y^{m_2-m_1}=x^{n_2}y^{m_2}(x^{n_1}y^{m_1})^{-1}=1$, we have
\begin{align*}
  &  (\alpha^1-\alpha^0)-(\beta^1-\beta^0)\\
    &=q^{n_2-n_1}(1-q)^{m_2-m_1}[x^{n_2-n_1}y^{m_2-m_1}\bigl(1+\alpha^{\prime 1}-\alpha^{\prime 0}\bigr)-\bigl(1+\beta^{\prime 1}-\beta^{\prime 0}\bigr)]\\
   &= q^{n_2-n_1}(1-q)^{m_2-m_1}(\alpha^{\prime 1}-\alpha^{\prime 0}-\beta^{\prime 1}+\beta^{\prime 0}).
\end{align*}
Hence, what we need to show is $\alpha^{\prime 1}-\alpha^{\prime 0}=\beta^{\prime 1}-\beta^{\prime 0}$.

By repeatedly applying the same argument, we obtain analogous expressions at each subsequent tie. Since only finitely many agents can face a tie, this procedure must terminate after finitely many steps. Consequently, we have
\begin{align*}
    V_i^{\mathcal D_r}(\pi,\bm{\sigma}^1)-V_i^{\mathcal D_r}(\pi,\bm{\sigma}^0)=0.
\end{align*}
\end{proof}

\begin{proof}[Proof of Theorem \ref{prop: binary binary}]
Take arbitrary $\pi,\pi' \in \Pi^B$ and $\mathcal D \in \mathscr D^B$ satisfying the conditions of Theorem \ref{prop: binary binary}. 
Let $r=\frac{u(a_0,L)-u(a_1,L)}{u(a_0,L)-u(a_1,L)+u(a_1,H)-u(a_0,H)}$.
Then, by Lemma \ref{lem: simple decision problem}, it suffices to show that
$V_i^{\mathcal D_r}(\pi)\geq V_i^{\mathcal D_r}(\pi')$.

We first consider the case in which $\pi$ contains a conclusive signal. If $\pi$ is fully informative, that is, $\pi(s_h|H)=1$ and $\pi(s_h|L)=0$, then the claim is immediate.
Next, suppose that $\pi(s_h|H)=1$ and $\pi(s_h|L)\in (0,1).$ 
In this case, at the equilibrium under $\pi$, each agent's behavior must be one of the following: (i) The agent always chooses $a_0$, regardless of both the history and her private signal. (ii) The agent chooses $a_1$ if and only if all preceding agents and the agent herself have received $s_h$.
If case (i) occurs, then every agent also always chooses $a_0$ under $\pi'$ by $\pi\succsim_B \pi' $. Hence, the expected payoff is the same under $\pi$ and $\pi'$.
If case (ii) occurs, then agent $i$ chooses $a_1$ if and only if the first $i$ agents all receive $s_h$. Therefore, agent $i$'s expected payoff under $\pi$ is equal to $\overline{V}_i^{D_r}(\pi).$
Thus, $V_i^{\mathcal{D}_r}(\pi)=\overline{V}_i^{D_r}(\pi)\geq \overline{V}_i^{D_r}(\pi')\geq V_i^{\mathcal{D}_r}(\pi').$ A symmetric argument establishes the case in which $\pi(s_h|H) \in (0,1)$ and $\pi(s_h|L)=0$.

Moreover, if agent $1$ chooses an action independently of her private signal under $\pi'$, then every subsequent agent, including agent $i$, also chooses an action independently of her private signal and histories. Hence, $V_i^{\mathcal{D}_r}(\pi')$ is equal to the payoff obtained under no information. It therefore follows immediately that $V_i^{\mathcal{D}_r}(\pi)\geq V_i^{\mathcal{D}_r}(\pi')$.

Therefore, it remains to consider the case in which $\pi,\pi' \in \overline{\Pi}^{B}$ and $(x,y)=\varphi(\pi)$ and $(x',y')=\varphi(\pi')$ satisfy $y\leq y' < K_r < x'\leq x$. We henceforth focus on this case.

Define a continuous function $c:[0,2]\to \{(a,b)\in \mathbb{R}^2 \mid 1<a,0<b<1\}$ by $c(t)=(x'+t(x-x'),y')$  and $c(t+1)=(x,y'+t(y-y'))$ for $t\in [0,1]$. Let $\pi^t=\varphi^{-1}(c(t))$ for $t\in [0,2]$. By construction, we have $\pi^s\succsim_B \pi^u$ if $s\geq u$.
For each $(n,m)\in\mathbb Z_+^2 $, the curve $a^n b^m = K_r$
intersects the path $c(t)$ at most once on $[0,1]$ and at most once on $[1,2]$, since the slope of the curve is negative. Hence, it has at most two intersections with $c(t)$.
Therefore,
\begin{align*}
    \left\{(a,b)\in\mathbb{R}^2 \mid a>1,\,0<b<1,\,a^n b^m = K_r\text{ for some }(n,m)\in\mathbb{Z}_+^2\text{ with }n+m\leq N\right\}
\end{align*}
and the path $(a,b)=c(t)$ have only finitely many intersection points.
Let $0<\underline{t}_1<\cdots<\underline{t}_n<1$ be all values of $t$ corresponding to such intersection points with $0<t<1$, and let $1<\overline{t}_1<\cdots<\overline{t}_m<2$ be all values of $t$ corresponding to such intersection points with $1<t<2$.
Finally, define $\underline{t}_0:=0$, $\underline{t}_{n+1}:=1$, $\overline{t}_{0}=1$, and $\overline{t}_{m+1}:=2$.

Note that if $t^1\geq t^2$ satisfies $(t^1,t^2)\in (\underline{t}_j,\underline{t}_{j+1})^2\cup(\overline{t}_j,\overline{t}_{j+1})^2\cup[0,\underline{t}_{1})^2\cup(\overline{t}_m,2]^2$ for some $j$, $\pi^{t^1}$ and $\pi^{t^2}$ are in the same $r$-region. Then, by Proposition \ref{prop: same r region}, $V_i^{\mathcal{D}_r}(\pi^{t^1})\geq V_i^{\mathcal{D}_r}(\pi^{t^2})$ holds.

Since each $r$-region is an open set, there exists $\underline{t}^*\in (\underline{t}_n,1)$ and $\overline{t}^*\in (1,\overline{t}_1)$ such that $\underline{t}^*$ and $\overline{t}^*$ are in the same $r$-region. Then, it follows that $V_i^{\mathcal{D}_r}(\pi^{\overline{t}^*})\geq V_i^{\mathcal{D}_r}(\pi^{\underline{t}^*})$ by Proposition \ref{prop: same r region}.

By construction, for all $t<\underline{t}_j$ sufficiently close to $\underline{t}_j$, the equilibrium action profile under $\pi^t$ coincides with the equilibrium under $\pi^{\underline{t}_j}$ in which ties are broken in favor of $a_0$. Therefore,
\begin{align*}
    \lim_{t\nearrow \underline{t}_j}V_i^{\mathcal{D}_r}(\pi^{t})=V_i^{\mathcal{D}_r}(\pi^{\underline{t}_j},\bm{\sigma}^0),
\end{align*}
 where $\bm{\sigma}^0$ is an equilibrium under $\pi^{\underline{t}_j}$ that breaks ties in favor of $a_0$.
 Similarly, we have
 \begin{align*}
    \lim_{t\searrow \underline{t}_j}V_i^{\mathcal{D}_r}(\pi^{t})=V_i^{\mathcal{D}_r}(\pi^{\underline{t}_j},\bm{\sigma}^1),
\end{align*}
 where $\bm{\sigma}^1$ is an equilibrium under $\pi^{\underline{t}_j}$ that breaks tie in favor of $a_1$.
Thus, Proposition \ref{prop: continuity} implies
  \begin{align*}
   \lim_{t\nearrow \underline{t}_j}V_i^{\mathcal{D}_r}(\pi^{t})=  \lim_{t\searrow \underline{t}_j}V_i^{\mathcal{D}_r}(\pi^{t})
\end{align*}
for all $j$. Analogously, for all $j$,
\begin{align*}
   \lim_{t\nearrow \overline{t}_j}V_i^{\mathcal{D}_r}(\pi^{t})=  \lim_{t\searrow \overline{t}_j}V_i^{\mathcal{D}_r}(\pi^{t})
\end{align*}

Repeatedly applying the monotonicity within each $r$-region and the continuity at each boundary point, we obtain
\[
    V_i^{\mathcal D_r}(\pi^{0})
    \leq V_i^{\mathcal D_r}(\pi^{\underline{t}^*}) 
    \leq V_i^{\mathcal D_r}(\pi^{\overline{t}^*}) 
    \leq V_i^{\mathcal D_r}(\pi^{2}).
\]
Therefore, we obtain $V_i^{\mathcal D_r}(\pi)\geq V_i^{\mathcal D_r}(\pi')$, which completes the proof.
\end{proof}
\begin{remark}
Note that, even in the binary-signal case, the above properties may fail when tie-breaking arises, depending on the equilibrium selection. The following examples demonstrate these phenomena.
\begin{example}
Let $\mu_0=1/2$ and $r=3/4$. Suppose that $\pi(s_h|H)=6/11$ and $\pi(s_h|L)=1/11$.
Then, $\varphi(\pi)=(6,1/2)$. For each $\varepsilon>0$, let $\pi'(\varepsilon),\pi''(\varepsilon)\in\overline{\Pi}^B$ be information structures satisfying $\varphi(\pi''(\varepsilon))=(6+\varepsilon,1/2)$ and $\varphi(\pi'(\varepsilon))=(6-\varepsilon,1/2)$. Note that $\pi''(\varepsilon)\succsim_B\pi \succsim_B  \pi'(\varepsilon)$ for all $\varepsilon>0$.

Define three equilibrium strategy profiles for the first three agents in $(\mathcal D_r,\pi)$, denoted by $\bm{\sigma}^0$, $\bm{\sigma}^+$, and $\bm{\sigma}^-$. 
First, let $\bm{\sigma}^0$ be an equilibrium in which every agent chooses $a_0$ whenever she is indifferent between the two actions.
Under $\bm{\sigma}^+$, agent 2 follows her private signal whenever she is indifferent: she chooses $a_0$ after observing $s_l$ and $a_1$ after observing $s_h$. Agent 3 chooses $a_0$ whenever she is indifferent.
Under $\bm{\sigma}^-$, agent 2 takes the action opposite to her private signal whenever she is indifferent: she chooses $a_1$ after observing $s_l$ and $a_0$ after observing $s_h$. Agent 3 chooses $a_0$ whenever she is indifferent.

It is straightforward to verify that agent~3 chooses $a_1$ if and only if the realized signal profile is $(s_h,s_h,s_h)$, $(s_h,s_h,s_l)$, or $(s_h,s_l,s_h)$ under $\bm{\sigma}^0$; $(s_h,s_h,s_h)$, $(s_h,s_h,s_l)$, $(s_h,s_l,s_h)$, or $(s_l,s_h,s_h)$ under $\bm{\sigma}^+$; and $(s_h,s_h,s_h)$ or $(s_h,s_l,s_h)$ under $\bm{\sigma}^-$.
Hence, it can be calculated that
\[
V_3^{\mathcal{D}_r}(\pi,\bm{\sigma}^0)=\frac{513}{10648}, \quad
V_3^{\mathcal{D}_r}(\pi,\bm{\sigma}^+)=\frac{663}{10648}, \quad
V_3^{\mathcal{D}_r}(\pi,\bm{\sigma}^-)=\frac{363}{10648}.
\]
By Proposition \ref{prop: continuity}, we know that
\begin{align*}
\lim_{\varepsilon \searrow 0}V_3^{\mathcal{D}_r}(\pi'(\varepsilon))=\lim_{\varepsilon \searrow 0}V_3^{\mathcal{D}_r}(\pi''(\varepsilon))=V_3^{\mathcal{D}_r}(\pi,\bm{\sigma}^0).
\end{align*}
Therefore, for sufficiently small $\varepsilon>0$, it follows that
$V_3^{\mathcal{D}_r}(\pi'(\varepsilon))>V_3^{\mathcal{D}_r}(\pi,\bm{\sigma}^-)$ although $\pi\succsim_B \pi'(\varepsilon)$, and $V_3^{\mathcal{D}_r}(\pi''(\varepsilon))<V_3^{\mathcal{D}_r}(\pi,\bm{\sigma}^+)$ although $\pi''(\varepsilon)\succsim_B \pi$. \qed
\end{example}
\end{remark}

\end{document}

\endgroup

\end{document}